%% file: main.tex
\documentclass[a4paper,11pt]{article}
\pdfoutput=1

\usepackage{tcs}

\newcommand{\anote}[1]{}
\newcommand{\gnote}[1]{}
\begin{document}

%\title{Time-efficient learning of quantum Hamiltonians \\ from Gibbs states at any temperature}
\title{Efficient Certificates of Anti-Concentration \\
Beyond Gaussians}
\author{
Ainesh Bakshi\thanks{Supported by the NSF TRIPODS program (award DMS2022448) and Ankur Moitra’s ONR grant.} \\
\texttt{ainesh@mit.edu} \\
MIT
\and
Pravesh K. Kothari \\
\texttt{kothari@cs.princeton.edu} \\
IAS \& Princeton University
\and
Goutham Rajendran \\
\texttt{gouthamr@cmu.edu} \\
CMU
\and
Madhur Tulsiani\thanks{ Supported by the NSF grants CCF-1816372 and CCF-2326685.}\\
\texttt{madhurt@ttic.edu} \\
TTIC
\and
Aravindan Vijayaraghavan\thanks{Supported by NSF grants CCF-1652491, CCF-2154100 and ECCS-2216970.}\\
\texttt{aravindv@northwestern.edu}\\
Northwestern
}
\date{}

\maketitle

\begin{abstract}
A set of high dimensional points $X=\{x_1, x_2,\ldots, x_n\} \subseteq \R^d$ in isotropic position is said to be $\delta$-anti concentrated if for every direction $v$, the fraction of points in $X$ satisfying $|\iprod{x_i,v}|\leq \delta$ is at most $O(\delta)$. 
%
%
%Perhaps unsurprisingly, finding such a direction given $X$ %(or certifying that there isn't any) is easily shown to be NP-hard. 
%
Motivated by applications to list-decodable learning and clustering, three recent works \cite{karmalkar2019list,raghavendra2020list,bakshi2021list} considered the problem of constructing efficient certificates of anti-concentration in the average case, when the set of points $X$ corresponds to samples from a Gaussian distribution.
Their certificates played a crucial role in several subsequent works in algorithmic robust statistics on list-decodable learning and settling the robust learnability of arbitrary Gaussian mixtures.
Unlike related efficient certificates of concentration properties that are known for wide class of distributions~\cite{kothari2018robust}, the aforementioned approach has been limited only to \textit{rotationally invariant} distributions (and their affine transformations) with the only prominent example being Gaussian distributions.

This work presents a new (and arguably the most natural) formulation for anti-
concentration. 
Using this formulation, we give quasi-polynomial time verifiable sum-of-squares certificates of anti-concentration that hold for a wide class of \emph{non-Gaussian} distributions including anti-concentrated bounded product distributions and uniform distributions over $L_p$ balls (and their affine transformations). Consequently, our method upgrades and extends results in algorithmic robust statistics e.g., list-decodable learning and clustering, to such distributions. %\avnote{I think we should mention quasipolynomial time here}
%
 %\avnote{e.g., list-decodable learning, clustering}

As in the case of previous works, our certificates are also obtained via relaxations in the sum-of-squares hierarchy. 
However, the nature of our argument differs significantly from prior works that formulate anti-concentration as the non-negativity of an explicit polynomial.
%, and prove non-negativity by expressing the polynomial as a sum-of-squares.
Our argument constructs a canonical integer program for anti-concentration and analyzes an SoS relaxation of it. The certificates used in prior works can be seen as specific dual certificates to this program. 

From a technical standpoint, unlike existing works that explicitly construct sum-of-squares certificates, our argument relies on duality and analyzes a pseudo-expectation on large subsets of the input points that take a small value in some direction. Our analysis uses the method of polynomial \emph{reweightings} to reduce the problem to analyzing only \emph{analytically dense} or \emph{sparse} directions.

%Unlike essentially all prior works, we do not know of an argument that  establish the existence of sum-of-squares certificates by explicitly constructing one. Instead, our argument relies on duality and analyzes a pseudo-expectation on large subsets of the input points that take a small value in some direction. Our analysis uses the method of polynomial \emph{reweightings} to reduce the problem to analyzing only analytically \emph{dense} or \emph{sparse} directions. 

% with a \emph{dual formulation} of anti-concentration, which can be viewed as a sum-of-squares relaxation of a canonical integer program.
% %

% our argument relies on an indirect argument that analyzes the \emph{dual witness} or \emph{pseudo-expectation} via low-degree \emph{reweightings} that have earlier appeared in the context of rounding high-degree sum-of-squares relaxations.  

\end{abstract}

\thispagestyle{empty}
\clearpage
\newpage

\microtypesetup{protrusion=false}
\tableofcontents{}
\thispagestyle{empty}
\microtypesetup{protrusion=true}
\clearpage
\setcounter{page}{1}

\input{intro}

\input{prelims}

\input{anti-concentration-program}
\input{sos-anti-conc-dense-directions}
\input{re-weighting-pseudo-distributions}
\input{clustering-mixtures}
\input{list-decodable-regression}

\input{distributions}

\paragraph{Acknowledgements}

We thank Xue Chen for extensive discussions.  
G. R. was supported in part by NSF grants CCF-1816372
and CCF-2008920.

\printbibliography

\appendix

\input{appendix}
\end{document}

%% file: intro.tex
% \textcolor{red}{To Dos.}

%%%%%%%%%%%%%%%%%%%%%% Text Macros
\newcommand{\ie}{i.e.,\xspace}
\newcommand{\eg}{e.g.,\xspace}
\newcommand{\etal}{et al.\xspace}
%%%%%%%%%%%%%%%%%%%

% \begin{enumerate}
%     \item abstract+ intro 
%     \item technical overview
%     \item clean up in section 9 (distribution families)
%     \item complete coefficient bounds
%     \item leftover comments
% \end{enumerate}

\section{Introduction}

%\ainesh{CAUTION: to the best of my current knowledge we don't have NP hardness for this problem so let's not say it in the abstract.}
%
A high-dimensional random variable $X \in \R^d$ is said to be $(\eps,\delta)$ anti-concentrated, if
along any direction $v \in S^{d-1}$ and $a \in \R$, we have $\Pr\Bracks{\Abs{\Iprod{X,v}-a} \leq \delta} \leq
\eps$. 
For example, a spherical Gaussian (zero mean and identity covariance) can easily be shown to be $(\delta, \delta)$ anti-concentrated for
any $\delta > 0$ (for sufficiently large $d$).
While concentration properties of distributions can be studied via tools such as moment bounds, the study of anti-concentration often requires more delicate structural information about the underlying
distribution~\cite{krishnapur2016anti}.

In this work, motivated by various statistical applications, we consider the problem of \emph{certifying} anti-concentration properties for a set of points. 
In analogy with the above definition of random variables, a set of high dimensional
points $X=\{x_1, x_2,\ldots, x_n\} \subseteq R^d$ in isotropic position (zero empirical mean, and identity empirical covariance) is said to be
$(\eps,\delta)$ anti-concentrated if for every unit vector $v$ and a shift $a \in \R$, the fraction of points
in $X$ satisfying $|\iprod{x_i,v}|\leq \delta$ is at most $\eps$  \ie $\Pr_{i \in
  [n]}\Bracks{\Abs{\Iprod{x_i,v}-a} \leq \delta} \leq \eps$. 
%
% In fact, for problems we will discuss in this work, we will be able to take $\eps$ to be $O(\delta)$ and simply refer to the set of points above as $\delta$ anti-concentrated.
%
\madhur{Actually, the move from $(\eps,\delta)$ anti-concentration to $\delta$ anti-concentration is not always without loss of generality, or even just polynomial loss of parameters, so may need to add more why we can do this. Moving it to later when are already in the statistical setting.}

The problem of deciding if a given set $X$ is $(\eps,\delta)$ anti-concentrated is
NP-hard in the worst case.\footnote{In the special case when $\delta=0$, this corresponds to deciding whether there exists an $\eps$ fraction of points that lie on a $d-1$ dimensional subspace. This is known to be NP-hard, see e.g. ~\cite{hardt2013algorithms, khachiyan1995}.} %\madhur{Citation? Even for worst-case of exact problem?}
%
%Moreover, approximate versions of this problem are also known to be hard under appropriate complexity-theoretic assumptions. Moitra and Hardt~\cite{hardt2013algorithms} \madhur{Actually, the ``no'' case in
%   Moitra and Hardt seems weaker, since they only show that no set of $\eps d/n$ point lie in a $d$
%   dimensional subspace. But they could all be very close to a subspace, which will violate
%   anticoncentration. Moreover, their reduction seems to work only for $\delta d$ dimensional
%   subspaces where $\delta$ is the scale of SSE, and not for $d-1$ dimensional subspaces which we
%   need.  Will look at this more carefully.}
% %
Moreover, for sufficiently small values of the anti-concentration scale $\delta \approx
1/\sqrt{d}$, efficient algorithms for certifying $\delta$ anti-concentration 
imply algorithms for the (homogeneous version of) the Continuous Learning With Errors (hCLWE)
problem considered by Bruna, Regev, Song, and Tang~\cite{bruna2021continuous}. 
%This would in turn imply
%efficient quantum algorithms for approximating worst-case lattice problems that are believed to be computationally intractable. 
%

In the average case, when the point set $X$ corresponds to independent samples of certain
distributions, the works~\cite{karmalkar2019list,raghavendra2020list} identified certificates of anti-concentration as a key 
 primitive for problems in algorithmic robust statistics. This eventually led to algorithms for list-decodable
linear regression~\cite{karmalkar2019list}, subspace recovery~\cite{raghavendra2020list2, bakshi2021list}, covariance estimation~\cite{ ivkov2022list}, and learning mixtures of
Gaussians~\cite{bakshi2022robustly, liu2021settling}.
Moreover, certifying anti-concentration is also shown to be \emph{necessary} for list-decodable
regression~\cite{karmalkar2019list}. 

% \avnote{The following paragraph needs rewording to make crisper. }
% The certificates developed in these prior works are in terms of proofs expressed in the sum-of-squares
% (SoS) hierarchy, which express anti-concentration in the direction $v$ as the non-negativity of a
% low-degree polynomial in $v$, and prove this by expressing the polynomial as a sum-of-squares of
% other low-degree polynomials.
% %
The certificates for anti-concentration in prior works~\cite{karmalkar2019list,raghavendra2020list} were obtained by formulating it as the non-negativity of an explicit polynomial that admits a sum-of-squares decomposition. In fact, the explicit forms needed additional tinkering to suit the application at hand. And the only natural examples of distributions that satisfied the requirements included Gaussian distributions and the uniform distribution over the Euclidean ball/sphere.
%The prior works formulate anti-concentration in the direction $v$ as the non-negativity of an explicit polynomial in $v$; they then prove that the explicit polynomial is non-negative by expressing it as a sum-of-squares of low-degree polynomials. 

%
%Thus, these certificates do not extend to several natural classes of distributions, such as 
%product distributions with subexponential tails.
% even the
% uniform distribitions over $\ell_p$ balls when $p \neq 2$ \gnote{TBD}.
%
% this served more as a design principle for the algorithms and their analysis, rather than an independent result which was directly usable for multiple applications. 
% %
% This is because the definition of what is an SoS certificate of anti-concentration actually varied with the application, and the analysis of the relevant algorithms used not just the fact that the certificate exists, but also the \emph{proof} that the certificate is non-negative.
%
In this work, we consider the following natural question:

\begin{quote}
    \begin{center}
        {\em Does there exist an application-agnostic certificate of anti-concentration that applies to a broad class of distributions? }
        %in particular beyond affine transformations of rotationally invariant distributions?}
    \end{center}
\end{quote}

% We are certainly not the first to consider this question, it appears as an open question in Bakshi's thesis (see Open Question 4 in~\cite{bakshi2022algorithms}). 
In this work, we obtain such certificates of anti-concentration via sum-of-squares relaxations of a natural integer program. Our certificates apply to a broad class of distributions that goes beyond Gaussians, %including distributions that are not necessarily rotationally invariant 
or rotationally invariant distributions (or their affine transformations).
We also unify several prior applications of anti-concentration and resolve other open problems raised in prior works.

% \gnote{Pravesh's thoughts: Sell it as conceptual contribution, e.g. beyond sub-gaussian. Hopefully can recover prior results. For applications, ideal thing is to recover prior works, but clustering should be enough (list-decodable regression might be much easier (lemma 4.1 of KKK'20 can be reproduced).}

% \subsection{Some notes}

% Motivation 

% - Key property of distributions, used in list-decodable learning and robust statistics
% - More challenging to study than concentration 
% - Showing anti-concentration is NP-hard using Hardt-Moitra and Hardt-Price

% Change over prior work

% - Previous approaches limited to “essentially” Gaussians
% - Application-dependent notions of anti concentration which needed modification with problem, and could not be applied black-box. 
% - Recover several previous results using anti-concentration certificates
% - Certificates work when shifted by a bounded point

% Results

% - Analysis of anti-concentration certificate for  degree ??
% - Clustering mixture of k-Gaussians
% - Subspace clustering

% Technical Overview

% - Anti-concentration along dense directions using Berry-Eseen computations
% - Anti-concentration along sparse directions using a dual argument:
%     - Can fix a few coordinates of a direction to ensure remaining vector has hypercontractive moments
%     - Show that neither the remaining dense part, nor the fixed sparse part, can contribute too much to anti concentration program
%     - Reweighting?
% - 

\subsection{Our Results}

We begin by defining the class of distributions that would be the focus of this paper (refer to \cref{def:reasonably-anti-concentrated-dist} for a quantitative definition). 

\begin{definition}[Reasonably anti-concentrated distributions]
\label{def:reasonably-anti-concentrated-dist-intro}
A distribution $\calD$ over $\R^d$ is \textit{reasonably anti-concentrated} if it satisfies the following: (a) $\calD$ is anti-concentrated (b) $\calD$ has \textit{strictly sub-exponential} tails, (c) $\calD$ satisfies \textit{hyper-contractivity} of linear forms and (d) $\calD$ is almost $k$-wise independent. 
\end{definition}

Apart from Gaussian-like distributions, we show that a broad family of distributions are \textit{reasonably anti-concentrated}, including arbitrary affine transformations of anti-concentrated bounded product distributions e.g., the solid hypercube (see \cref{thm:product-distributions-are}) and uniform distributions over $\ell_p$ balls for any $p>1$ (see \cref{thm:anti_concentrated_distributions}).

Now, given a point set $X = \{x_1, \ldots, x_n\} \subseteq \R^d$ of samples from a distribution with mean $\mu$ and covariance $\Sigma$, consider the following natural optimization problem to certify anti-concentration:
\begin{equation}
\label{eqn:cons-system-intro}
\begin{split}
    & \hspace{0.2in} \max_{w \in \mathbb{R}^n , v\in \mathbb{R}^d  } \hspace{0.2in} \frac{1}{n} \sum_{ i \in [n]} w_i  
    \\
     \textrm{s.t.} & \left \{\begin{aligned}
      &\forall i\in [n]
      & w_i^2
      & = w_i \\
      &\forall i\in [n] & w_i \Iprod{x_i-\mu, v}^2 & \leq w_i \cdot \delta^2  \cdot (v^\top \Sigma v) \\
    \end{aligned}\right\}
\end{split}
\end{equation}
Here, the $w_i$'s are indicators, and the constraints count how many samples are smaller than $\delta^2$ fraction of the variance along $v$. This program is \textit{the} natural integer program formulation of anti-concentration\footnote{This formulation was first suggested in a \href{https://youtu.be/xKoMFCQojCk?t=2496}{talk} at the Simons Institute Berkeley.}.  Intuitively, this program finds a direction $v$ such that projecting the samples along $v$ witnesses the largest fraction of \textit{concentrated} samples.

Our main result shows that we can efficiently \textit{certify} an upper bound on the objective value of this program as long as $X$ is drawn i.i.d. from any affine transformation of a \textit{reasonably anti-concentrated} distribution.  

\begin{theorem}[Efficient Anti-Concentration certificate, Informal version of \cref{thm:main-anti-concentration-thm}]
Given a point set $X = \Set{x_i }_{i\in [n]}$ in $\R^d$ such that $X$ is drawn iid from an unknown affine transformation of a \textit{reasonably anti-concentrated} distribution $\calD$ and $0<\delta <1$, there exists an algorithm that with high probability over the samples, certifies $\max_{w, v} \sum_{i \in [n]} w_i /n \leq \delta$ and runs in $n^{  \mathcal{O}_{\delta}(\log(d)) }$ time.
\footnote{Throughout, we use $\mathcal{O}_{\delta}$
 to suppress the explicit dependence on $\delta$}
\end{theorem}

\begin{remark}[Comparison to prior work]
Prior works~\cite{karmalkar2019list, raghavendra2020list, raghavendra2020list2,bakshi2021list,bakshi2020outlier} have exhibited similar certificates of anti-concentration, in the context of list-decodable regression, subspace recovery and clustering, but there are two main differences compared to our work.
    \begin{enumerate}
        \item While their certificates also also obtained via SoS relaxations, they have worked with non-standard formulations which are usually adapted to the problem at hand and consider explicit polynomial approximations to the box indicator function. Furthermore, the analyses of such certificates in the corresponding algorithmic applications have been ad-hoc. In contrast, our certificate considers the canonical integer program formulation \cref{eqn:cons-system-intro} that essentially mimics the definition of anti-concentration.
        \item The certificates in prior works only work for spherically symmetric distributions such as Gaussians or uniform distribution on the sphere. This essentially reduces certifying anti-concentration to a univariate problem. Our certificate on the other hand goes beyond spherical symmetry and allows us to handle the more general notion of {\em reasonably anti-concentrated} distributions instead (\cref{def:reasonably-anti-concentrated-dist-intro}), albeit at the cost of quasi-polynomial running time.
    \end{enumerate}
\end{remark}
Finally, we note that we indirectly infer the existence of a sum-of-squares proof by showing that any low-degree pseudo-distribution consistent with the program in \cref{eqn:cons-system-intro} must in fact satisfy that the pseudo-expected objective value is at most $\delta$. It then follows from SDP duality that a low-degree sum-of-squares certificate of this inequality must exist. Our inability to obtain a direct sum-of-squares proof perhaps illustrates why finding such certificates has remained elusive in the past.

\subsection{Applications} 
We show that we can use our certificate of anti-concentration to obtain efficient algorithms for various statistics problems, including the first algorithms for clustering mixtures of \textit{reasonably anti-concentrated} distributions with arbitrary covariances, and list-decodable regression for such distributions. We note that all prior work on these problems was limited to affine transformations of rotationally invariant distributions, specifically Gaussians.

\paragraph{Warmup: Spectrally separated mixture with two clusters.} As a warmup, we use our new anti-concentration certificate to give a quasi-polynomial time algorithm for clustering any mixture of two reasonably anti-concentrated distributions that are spectrally separated i.e., there is a direction $v$ along which the variances of the two components differ significantly.  

\begin{theorem}[Spectrally Separated Mixture, Informal version of \cref{thm:clusterin-2-spectrally-separated}]
Consider a set of $n$ i.i.d. samples drawn from $\calM = \frac{1}{2} \calD(0, \Sigma_1) + \frac{1}{2} \calD(0, \Sigma_2)$ such that there exists a direction $v$ satisfying $v^\top \Sigma_1 v < \delta^3  v^\top \Sigma_2 v$, for some $\delta <1$. Then, there exists an algorithm that runs in $n^{ \mathcal{O}_{\delta}(\log^2(d)) }$ time and with probability at least $99/100$ outputs a direction $z$ such that $z^\top \Sigma_1  z < \bigO{\delta} z^\top \Sigma_2 z$ , whenever $n = d^{ \Omega_{\delta}(\log(d)) }$.  
\end{theorem}

The above algorithm when given samples from such a spectrally separated mixture finds a direction $z$ where the variances of the two components differ by a factor of $\Omega(1/\delta)$. One can then project along this direction to recovers the two clusters up to $O(\delta)$ error. 

This setting of spectrally separated mixtures captures as a special case the well-studied problem of {\em subspace clustering} \cite{vidal2003generalized,bakshi2021list}, where the components are supported on two distinct subspaces. Hence there is a direction $v$ along which the first component has non-zero variance, and the second component has variance $0$ i.e., the spectral separation is arbitrarily large. 
The above algorithm provides quasi-polynomial time guarantees when the points in each cluster are drawn from {\em reasonably anti-concentrated} distributions supported on the corresponding subspace. To the best of our knowledge, previous guarantees only hold in the special case when the points in the subspace is drawn from a Gaussian (or a rotationally invariant distribution) restricted to the subspace, or assume other structural assumptions on the subspaces~\cite{bakshi2021list, chandra2024learning}. Moreover, Theorem~\ref{thm:clusterin-2-spectrally-separated} guarantees approximate recovery of the clusters even when the spectral separation is only a (sufficiently large) constant factor. 

%Next we show that these techniques can be used for the more general problem of clustering a mixture of $k$ {\em reasonably anti-concentrated} distributions. 

%where the variance of the first component along $v$ is $0$ while the variance of the second    
%In the well-studied problem of subspace clustering, the two clusters have covariance matrices which are particular, this captures the well-studied problem of subspace clustering   

% \begin{remark}
% \ainesh{add discussion of how this is similar to subspace recovery, clustering etc.}
% \ainesh{ref Open question 4 in Bakshi's thesis; also ~\cite{bakshi2021list}}
% \end{remark}

\paragraph{Clustering mixtures of \emph{ reasonably anti-concentrated} distributions.}
Next, we consider the more general setting when the points are drawn from a mixture of $k$ \textit{reasonably anti-concentrated } distributions. There is rich body of work on clustering mixture models where efficient polynomial time algorithms are known when the means of the clusters are separated sufficiently see e.g., \cite{vempala2004spectral, achlioptas2005spectral, kumar2010clustering, hopkins2018mixture, kothari2017outlier,diakonikolas2017learning}. In the specific setting when the component distributions are Gaussians, recent works provide polynomial time guarantees under significantly less stringent notion of separation called the ``clusterable setting'', which includes spectral separation, or Frobenius separation, or mean separation; see \cref{def:clusterable-mixture} for a formal definition. 

To the best of our knowledge, there are no known algorithms in the clusterable setting for more general distributions beyond Gaussians. We make progress on this question by giving an algorithm that
%, with an additional technical condition of certified hypercontractivity of quadratic forms (\cref{def:certifiable-hypercontractivity})), 
recovers the clusters up to $O(\delta)$ error when the underlying mixture is clusterable.

% We show that if the distributions also satisfy an extra technical condition , then there exists an efficient algorithm to recover the clusters.

\begin{theorem}[Clustering Mixtures, Informal version of \cref{thm:main-clusterin-theorem}]
 Given $0< \eps< 1$ and  $n$ samples from a mixture of $k=O(1)$ reasonably anti-concentrated distributions that also have certifiably hypercontractive quadratic forms (\cref{def:certifiable-hypercontractivity}) such that the resulting mixture is \textit{clusterable} (\cref{def:clusterable-mixture}), there exists an algorithm that runs in $n^{ \mathcal{O}_{k,\eps}(\log^2(d) ) }$ time and outputs an $\eps$-accurate clustering  whenever $n  = d^{ \Omega_{k,\eps}\Paren{\log(d)} }$.
\end{theorem}

We note that affine transformations of product distributions and uniform distributions over $\ell_p$ balls, for $p \in (1, \infty)$ satisfy all the technical conditions for our statement. We prove this in \cref{thm:anti_concentrated_distributions}. 

%The above theorem gives a quasi-polynomial time guarantee for constant $k$. 
\begin{remark}[Outlier-robust algorithms]
    We remark that using standard machinery in robust statistics, the sum-of-squares relaxation we analyze can be made to work in the strong contamination model. Here, an $\eta$-fraction of input samples are adversarially corrupted, and the rest are drawn from a mixture of reasonably anti-concentrated distributions. We refer the reader to~\cite{fleming2019semialgebraic, bakshi2020outliermegre} for details.
\end{remark}

\paragraph{List-decodable linear regression.}
Finally, we consider the setting of list-decodable linear regression. Here an $\alpha$-fraction of the samples are \textit{inliers} and are drawn i.i.d. from some \textit{reasonably anti-concentrated} distribution $\calD$. The remaining $(1-\alpha)$ fraction of the samples are arbitrary, and potentially adversarially generated. The underlying statistical model is simply $y_i = \Iprod{ \Theta, x_i }$, where $x_i \sim \calD$, and $\Theta$ is some unknown hyperplane. For the formal definition, see \cref{model:regression}. The following theorem gives a quasi-polynomial time guarantee for this problem when the inliers are drawn from any resonably anti-concentration distribution. 

\begin{theorem}[List-Decodable Regression, Informal version of \cref{thm:list-decodable-regression}]
\label{thm:list-decodable-informal}
For any $0< \alpha, \eps < 1$, there exists an algorithm that takes $n$ samples from the above model where an $\alpha$-fraction are from a \textit{reasonably anti-concentrated} distribution, runs in $ n^{ \mathcal{O}_{\eps, \alpha}(\log(d))}$ time and outputs a list of $\bigO{1/\alpha}$ vectors such that with probability at least $99/100$, the list contains at least one vector $\hat{\Theta}$ such that $\norm{ \hat{\Theta} - \Theta }_2^2\leq \eps$, when $n= d^{ \Omega_{\eps, \alpha} ( \log(d)  )}$. 
\end{theorem}

% \begin{remark}
%     \ainesh{ KKK and RY only work for Gaussians and uniform over the sphere, but our algorithm is quasi-poly. }
%     \gnote{check and edit this, i don't like that the section is ending with a negative remark}
%     The works \cite{karmalkar2019list, raghavendra2020list} also study and prove polynomial time guarantees for list-decodable regression, however their results crucially only apply to Gaussians and rotationally invariant distributions (for technical reasons related to SoS proofs). 
%     %In contrast, our results provide quasi-polynomial time guarantees for 
%     In contrast, our techniques extend to a much broader class of distributions, e.g., product distributions with strictly sub-exponential marginals and gives quasi-polynomial time guarantees.
%     %We remark that our algorithm incurs a quasi-polynomial dependence as opposed to the polynomial time guarantees for algorithms that were specially designed for the Gaussian distribution.
% \end{remark}

\subsection{Technical Overview}

In this section, we provide an outline of the techniques we introduce to obtain our certificate of anti-concentration. 

\paragraph{Inefficient certificates of anti-concentration.} As a starting point, we note that the uniform distribution over any convex body is anti-concentrated, and therefore the optimization problem in \cref{eqn:cons-system-intro} achievies a value of at most $\delta$, for any $\delta <1$. The proof of anti-concentration simply follows from observing that the density of any univariate projection is upper bounded by $\exp(-t/c)$, for a fixed constant $c$. One na\"ive certificate that demonstrates $\sum_i w_i$ is bounded in all directions is one where we simply create an $\eps$-net and check each univariate projection, i.e. for each vector in the net, we project the distribution along this vector and check that it is anti-concentrated. However, it is easy to see that the size of such a certificate is at least $\exp(d)$ and our goal is to avoid this exponential dependence on the dimension. 

\paragraph{Existing approaches to certifying anti-concentration.} As alluded to earlier, the notion of certifiable anti-concentration was introduced in two concurrent works, one by Raghavendra and Yau~\cite{raghavendra2020list} and one by Karmalkar, Klivans and Kothari~\cite{karmalkar2019list}, and they formulate anti-concentration as a polynomial inequality. First, let $f(x)$ be the univariate box-indicator function: 
\begin{equation*}
    f(x) = \begin{cases}
    1  & \textrm{  if  }  -\delta \leq x \leq \delta    \\
    0 & \textrm{ otherwise}
    \end{cases}
\end{equation*}
Then, there exists a polynomial $p(x)$ that approximates $f$ everywhere in the interval $[-1,1]$, and the degree of this polynomial is roughly $\bigO{1/\delta^2}$ (in fact, there is a sum-of-squares polynomial satisfying these criteria). The formulation of anti-concentration considered in prior works is as follows:
\begin{equation}
\label{eqn:box-indicator-formulation-intro}
    \begin{split}
    \forall v \in \R^d,  \hspace{0.2in}\frac{1}{n}\sum_{i \in [n]} p( \Iprod{x_i, v}^2 )  \leq \delta \norm{v}^2_2.
    \end{split}
\end{equation}
Since $p$ is large in the interval $[-\delta, \delta]$, such a certificate implies that the number of samples for which $\abs{\Iprod{x_i, v}} \leq \delta$ is bounded. To avoid brute force search over an $\eps$-net,  \cite{raghavendra2020list, karmalkar2019list, bakshi2020outlier} show that there is a sum-of-squares proof of the inequality in \cref{eqn:box-indicator-formulation-intro} in the indeterminate $v$ (see \cref{subsec:sos-framework} for background on sum-of-squares proofs) and this in turn implies an efficient certificate. 

However, all known proofs of such a certificate assume that the input distribution is rotationally invariant. Therefore, $p(\Iprod{x_i , v}^2 ) = p( g \norm{v}^2 )$, where $g$ is a standard Gaussian, implying $p$ is a univariate function, and any univariate inequality admits a sum-of-squares proof (see \cref{fact:univariate-sos-proofs}). It is easy to see that all prior works thus heavily rely on the input distribution being rotationally invariant.

\paragraph{Anti-concentration along analytically dense directions.} It is well known that the marginal distribution obtained by projecting the uniform distribution over a convex body along a random direction is well-approximated by a Gaussian. Restricting our attention to well-behaved convex bodies,  we develop an analytic formulation of this statement: For any direction $v$ such that $\norm{v}_4^4 \leq \lambda \norm{v}_2^4$, we show that there is a sum-of-squares proof in $v$ that $\frac{1}{n} \sum_{i \in [n] } p^2( \Iprod{x_i, v}^2 )\leq \delta $, where $p^2$ is the sum-of-squares polynomial that approximates the box-indicator function. In the language of sum-of-squares proofs, 
\begin{equation}
    \Set{\norm{v}_4^4 \leq \exp\Paren{-1/\delta^2} \norm{v}_2^4 } \sststile{}{v} \Set{ \frac{1}{n} \sum_{i \in [n] } p^2( \Iprod{x_i, v}^2 )\leq \delta }, 
\end{equation}
i.e. under the axiom that the $\ell_4^4$ norm of $v$ is bounded, there is an efficient certificate that $p^2$ is small on the input samples in all directions (see~\cref{thm:anti-concentration-certificates-with-shifts}). The captures the intuition that $v$ cannot be a sparse direction or have large magnitude coordinates. Our proof proceeds by explicitly bounding the higher moments of $\Iprod{x ,v}^2$ and showing that under the axiom $\Set{\norm{v}_4^4 \leq \exp\Paren{-1/\delta^2} \norm{v}_2^4 }$, the moments are $\exp\Paren{-1/\delta^2}$-close to that of a Gaussian, additively (see \cref{lem:lower-bound-k-th-moment,lem: upper_bound_moment_independent})\footnote{Our proof is inspired by Lemma 2.6 in Kindler, Naor and Schechtman~\cite{kindler2010ugc}.}.

Additionally, we show that $p^2$ certifies anti-concentration around any bounded point $\phi$, as opposed to just the origin, i.e. 
\begin{equation*}
    \Set{\norm{v}_4^4 \leq \exp\Paren{-1/\delta^2} \norm{v}_2^4 } \sststile{}{v} \Set{ \frac{1}{n} \sum_{i \in [n] } p^2( \Iprod{x_i, v}^2 -\phi )\leq \delta },
\end{equation*}
only at the cost of $p^2$ having degree $\bigO{ 1/\delta^2 + \phi }$. This is a qualitative improvement over prior works, and is necessary to the remainder of our argument.

\paragraph{Anti-concentration on the remaining directions.} Directions that are not analytically dense include, in particular, sparse directions, and marginals along sparse directions do not behave like the Gaussian distribution. Nevertheless, it is tempting to try to show that the remaining set of directions is small, and perhaps admits a small $\eps$-net. Certainly, the set of sparse directions satisfies this property. However, consider the direction $v = [1/2, 1/2, c/\sqrt{n}, \ldots , c/\sqrt{n}  ]$, for a fixed constant $c$. This direction is neither sparse nor analytically dense, and we cannot discard the large or the small coordinates, since they contribute an equal amount of $\ell_2$-mass. Surprisingly, we do not know how to obtain an explicit sum-of-squares proof of anti-concentration for such directions. Instead, we show that we can still indirectly certify anti-concentration.

Inspired by this example, let us assume we have the ability to condition indeterminates, i.e. coordinates of $v$ that are large. Given this ability, it is not too hard to show that if we condition on all the coordinates that are at least $\lambda$, and denote this set by $u_\calS$, the resulting residual vector, $v - u_\calS$ would have $\ell_\infty$ norm bounded by $\lambda$, and in turn $\norm{  v - u_\calS }_4^4 \leq \lambda$. Further, $u_\calS$ will have at most $1/\lambda$ non-zero entries.  We could then appeal to univariate anti-concentration on $\Iprod{x, u_\calS}$ and analytically dense certificates on $\Iprod{x, v - u_\calS}$. While we do not posess the ability to condition large coordinates of $v$ in the proof system, we show how to execute such a proof strategy by looking at dual objects to sum-of-squares proofs, i.e. \textit{pseudo-distributions}.

\paragraph{Re-weighting pseudo-distributions.} Consider a pseudo-distribution $\mu$ over a vector valued indeterminate $v$ such that $\norm{v}^2 \leq 1$. Since $v$ is a unit vector, it does not take values in a discrete set and we cannot condition $\mu$ on a specific value (unlike the boolean setting, e.g.~\cite{barak2011rounding}). Instead, we show that we via polynomial re-weighting (a generalization of ``conditioning" introduced in~\cite{barak2017quantum}) the pseudo-distribution by a low degree polynomial, such as $v_i^{2t}$, where $v_i$ is the $i$-th coordinate of $v$ and $t$ is a large integer, and this morally has the same effect as that of conditioning (we refer the reader to~\cref{subsec:re-weighting-background} for background on re-weighting pseudo-distributions). 

% \pravesh{We should cite BKS for first finding and developing the idea of polynomial rweightings here and also at the first place we discuss techniques in the intro. David could be our reviewer. :) }

The main theorem we obtain is that after an iterative re-weighting process (\cref{algo:re-weighting}), there exists a new pseudo-distribution $\mu'$ and a sparse vector $u_\calS$ (this is a vector without any indeterminates) such that either $\pexpecf{\mu'}{\norm{v - u_\calS}_2^{2t}} \leq \eta$, i.e. all the low-degree moments of the residual vector are small, or $\pexpecf{\mu'}{ \norm{v - u_\calS}_{2t}^{2t} } \leq \eta$, i.e. the residual vector is analytically dense. Here $\eta$ can be made arbitrarily small at the cost of running time, however, $t$ must be $\Omega(\log(d))$, which results in the quasi-polynomial running time in all our results. 

The iterative re-weighting proceeds as follows: given the pseudo-distribution $\mu$ we find a coordinate $v_i$ such that $\pexpecf{\mu}{ v_i^{2t} } \geq \lambda^{2t} \pexpecf{\mu}{ \norm{v}_2^{2t} }$. If no such coordinate exists, it's easy to show that $\pexpecf{\mu}{ \norm{v}_{2t}^{2t}  }$ is small, since we essentially have an $\ell_{\infty}$ bound on $v$. Otherwise, we re-weight $\mu$ by $v_i^{2t}$, and set the new vector $v = v - \pexpecf{\mu}{v_i}$. We iterate until there is either no such coordinate or the norm of the residual vector is small. 

The analysis has two key components. The first is a scalar re-weighting statement that was obtained by Barak, Kothari and Steurer~\cite{barak2017quantum}: Given a bounded indeterminate $z$, there exists a re-weighting of $\mu$, denoted by $\mu'$, such that $ \pexpecf{\mu'}{ (z- \pexpecf{\mu}{ z})^{2t} } \leq \eps^{2t} \pexpecf{\mu}{ z}^{2t}$. Using this statement, we can conclude that in a single iteration, the re-weighting process essentially fixes the $i$-th coordinate. However, since the pseudo-distribution gets re-weighted over and over again, we might undo this conditioning. We instead design a carefully chosen potential that monotonically decreases, showing that the re-weighting process must terminate after a few iterations (see the proof of \cref{thm:key-re-weighting-theorem} for details).  

Finally, we note that $t=\Omega(\log(d))$ is necessary for our re-weighting strategy. Consider the case where $\mu$ is simply the uniform distribution over the standard basis vectors (a special case of a pseudo-distribution). Each element in the support has only one non-zero coordinate, but $\expecf{\mu}{ v_i^{2t} }$ remains a negligible fraction of the $\ell_2$ norm until $t=\Omega(\log(d))$.
% \ainesh{does anyone want to flesh this out more?}

\paragraph{Pseudo-expectations vs. axioms.}
Given our re-weighting scheme, it is not too hard to show that for the re-weighted pseudo-distribution $\mu'$, $\pexpecf{\mu'}{ \norm{v  - u_\calS }_4^4 } \leq \eta$. We would now like to invoke our certificates of anti-concentration along the direction $v-u_\calS$. However, we only obtain the axiom $\Set{ \norm{v - u_\calS }_4^4 \leq \eta }$ if the pseudo-expectation remains small for all low-degree sum-of-squares polynomials multiplied by  $\norm{v - u_\calS }_4^4$. Instead, we show that our certificates used the axiom by only multiplying it by $\norm{v - u_\calS}^{2\ell}_2$, for different integral $\ell$, and we can indeed prove that $\pexpecf{\mu'}{ \norm{v  - u_\calS}_2^{2\ell} \cdot \norm{v  - u_\calS }_4^4 } \leq \eta$ (see \cref{lem:upgraded-muirhead-sos,}). This allows us to invoke our certificates for directions that satisfy $\pexpecf{\mu'}{ \norm{v  - u_\calS }_4^4 } \leq \eta$.

\paragraph{Objective is bounded under a re-weighting.}
With the aforementioned components in place, we are ready to analyze the objective value of the anti-concentration program. Let $\mu$ be any pseudo-distribution that is consistent with the constraints in \cref{eqn:cons-system-intro}. For ease of exposition, we assume that $w_i \Iprod{x_i, v}^2 = 0$, and this constraint can be easily relaxed. Then, there exists a re-weighting of $\mu$ and a sparse vector $u_\calS$ such that either $\pexpecf{\mu'}{ \norm{v-u_\calS }_4^4 }$ is small or $\pexpecf{\mu'}{ \norm{v-u_\calS }_2^{2t} }$ is small. Consider the first case. Since $p^2(0)=1$, we can bound the objective as follows: 
\begin{equation}
\label{eqn:bounding-theobjective-intro}
\begin{split}
    \frac{1}{n} \pexpecf{\mu'}{ \sum_{i \in [n]} w_i} & =\frac{1}{n}\pexpecf{\mu'}{\sum_{i \in [n]} w_i p^2(0)}\\
    & = \frac{1}{n}\pexpecf{\mu'}{\sum_{i \in [n]} w_i p^2(w_i \Iprod{x_i, v})}\\
    & \leq \frac{1}{n} \pexpecf{\mu'}{ \sum_{i \in [n]} p^2\Paren{  \Iprod{x_i, u_\calS} + \Iprod{x_i, v-u_\calS} } }
\end{split}
\end{equation}
where the second equality follows from $p^2$ being a sum-of-squares polynomial and our constraint. We can then drop the indicators since they are at most $1$. Now, we treat the vector $\Iprod{x_i, u_\calS}$ as a shift (since this is a fixed vector) and invoke our analytic certificate on the direction $v- u_\calS$. For most samples $\Iprod{x_i, u_\calS}$ is bounded, and therefore,
the objective value is small (see \cref{lem:re-weighting-to-certificate} for details). 

Now, consider the second case. We again proceed as \cref{eqn:bounding-theobjective-intro} and split into the $p^2( \Iprod{x_i, u_\calS } + \Iprod{x, v-  u_\calS })$. We then just expand out the polynomial by writing it in the monomial basis, and show that we can treat all the terms involving $\Iprod{x, v-  u_\calS }$ as additive error. Since all low-degree moments of this vector are bounded, the additive error remains small (see \cref{lem:case1-decomposition}). We now invoke anti-concentration properties on $p^2(\Iprod{x, u_\calS})$, but since $u_\calS$ is sparse, and fixed, we no longer need sum-of-squares certificates, and can simply reduce to the univariate case. Taking a step back, we managed to show that any pseudo-distribution $\mu$ can be re-weighted to one for which $\frac{1}{n} \sum_{i \in [n]} w_i$ is small.

\paragraph{Inferring the existence of SoS proof.}
Semi-definite programming duality implies that if an inequality is true for all pseudo-distributions of degree-$k$, then there must be a degree-$k$ sum-of-squares proof of the inequality (see \cref{fact:sos-completeness}). Thus far, we have managed to show that any pseudo-distribution $\mu$ can be re-weighted to obtain the desired inequality, i.e. $\frac{1}{n}\pexpecf{\mu'}{\sum_{i\in [n]} w_i  } \leq \delta$. In general, the existence of a re-weighting need not imply a sum-of-squares proof. However, we use the specific structure of our re-weighting to show that such a proof must indeed exist. 

We proceed via contradiction. Let $\mu$ be a pseudo-distribution under which $\frac{1}{n}\pexpecf{\mu}{\sum_{i\in [n]} w_i  } > 10 \delta$. Then, we re-weight $\mu$ by $(\sum_{i \in [n]} w_i)^{2t}$ for a large $t$, to essentially fix the objective value, and let the resulting pseudo-distribution be $\mu_1$. The scalar fixing lemma described above ensures that the objective value remains large even if we apply low-degree polynomials to the objective. Now, we invoke our re-weighting result on $\mu_1$, we can show that there is a re-weighting of $\mu_1$, say $\mu_2$, such that $\frac{1}{n}\pexpecf{\mu_2 }{\sum_{i\in [n]} w_i  } < \delta/100 $. However, since this re-weighting is a low-degree polynomial, it could not have altered the objective value under $\mu_1$ by too much, therefore yielding a contradiction (we refer the reader to \cref{thm:main-anti-concentration-thm} for a complete proof). Formally, we can infer a certificate of anti-concentration of the following form:
\begin{equation}
\label{eqn:parametric-form-cert-intro}
    \delta - \frac{1}{n}\sum_{i\in [n]} w_i = \textrm{sos}(v,w) + \sum_{i \in [n]} q_i(v,w)\Paren{ w_i ( \delta \norm{v}^2 - \Iprod{x_i, v}^2) } + \sum_{i \in [n]} r_i( w_i^2 - w_i ) ,
\end{equation}
where $\textrm{sos}(v,w), q_i(v,w)$ are sum-of-squares polynomials and $r_i$ are arbitrary low-degree polynomials.

\paragraph{Clustering via certificates.}
We note that the certificates we obtain do not go via the certificates obtained in prior works that were discussed in \eqref{eqn:box-indicator-formulation-intro} and therefore, does not imply the algorithmic applications of anti-concentration in a black-box manner. We describe how to use our certificates to design algorithms. For the purposes of the overview, we consider a simple example: let $X = \Set{x_i }_{i\in[n}$ be $n$ samples from a mixture of two \textit{reasonably anti-concentrated} distributions, denoted by $\calM = \frac{1}{2}\calD(0, \Sigma_1) + \frac{1}{2} \calD(0, \Sigma_2)$. The two distributions are concentric (both means are zero) and the mixture is isotropized, so the mixture covariance is $I$. Further, there is some unknown direction $v$ such that $v^\top \Sigma_1 v \leq \delta v^\top \Sigma_2 v$, for some sufficiently small $\delta$. Our goal is to simply recover some direction that continues to separate $\Sigma_1$ and $\Sigma_2$. Note, we can use such a direction to cluster by simply projecting all the samples along this direction and thresholding them. 

To find such a direction we consider the following program: 
\begin{equation}
\label{eqn:cons-system-clustering-intro}
\begin{split}
    \calC_{\delta} = & \left \{\begin{aligned}
      &\forall i\in [n]
      & w_i^2
      & = w_i \\
      && \sum_{i \in [n]} w_i & =  n/2\\
      & \forall i\in [n] & w_i \Iprod{x_i, v}^2 & \leq c  \delta w_i \norm{v}^2    \\
    \end{aligned}\right\}
\end{split}
\end{equation}
This program is feasible, since we can set the $w_i$'s to indicate the samples from $\calD( 0, \Sigma_1)$. We then compute a pseudo-distribution that is consistent with $\calC_\delta$ and our goal is to obtain a separating direction. A natural candidate is to compute the covariance matrix of the indeterminates, i.e. $ \pexpecf{\mu}{ v v^\top}$ and then sample $g \sim \calN(0,  \pexpecf{\mu}{ v v^\top})$ (often referred to as Gaussian rounding). However, there can be several directions where the variance of $\calD(0, \Sigma_2)$ is significantly smaller than $\calD(0, \Sigma_1)$ and $\mu$ could be the uniform distribution over all such directions. In fact, it is not to hard to construct an instance where $\pexpecf{\mu}{ v v^\top} = I$ and therefore, our program would not succeed in finding a separating direction.

Instead, we pick one of the samples uniformly at random and re-weight the pseudo-distribution by $\Iprod{x_i, v}^{2t}$. With probability $1/2$, we sample $x_i$ from $\calD(0, \Sigma_2)$ and the re-weighting essentially conditions the pseudo-distribution to be supported on directions where $\Sigma_2$ has larger variance, and the aforementioned Gaussian rounding scheme succeeds.

A bit more formally, using the certificate from~\cref{eqn:parametric-form-cert-intro} we prove the following key sum-of-squares inequality (see \cref{lem:key-sos-identity} for a formal statement and proof): 
\begin{equation}
\label{eqn:key-sos-identity}
    \begin{split}
        \calC_{\delta} \sststile{}{} \Set{ \Iprod{ \Sigma_1, vv^\top }^2 \leq \bigO{1} \cdot \Biggl( { \delta^2 \norm{v}^4 - \sum_{ x_i \in \calD(0, \Sigma_2) }  \underbrace{ q_i(w,v)w_i\Paren{ \delta v^\top \Sigma_2 v  - \frac{\delta \norm{v}^2_2 }{100} } }_{ \eqref{eqn:key-sos-identity}.(1) } \Biggr)   }},
    \end{split}
\end{equation}
where the term \eqref{eqn:key-sos-identity}.(1) need not axiomatically be non-negative. Since we compute a pseudo-distribution consistent with $\calC_\delta$, we can conclude 
\begin{equation*}
    \Iprod{ \Sigma_1, \pexpecf{\mu}{ vv^\top  } }^2  \leq \pexpecf{\mu}{ \Iprod{ \Sigma_1, vv^\top }^2 } \leq \delta^2 \norm{v}^4_2 - \sum_{i} \pexpecf{\mu}{ q_i(w,v) w_i \Paren{ \delta v^\top \Sigma_2 v  - \frac{\delta \norm{v}^2_2 }{100} }   }.
\end{equation*}
Observe, as long as $\pexpecf{\mu}{ q_i(w,v) v^\top \Sigma_2 v } \geq 0.1$, we are in good shape, since the additive term disappears, and we have $\Iprod{\Sigma_1, \pexpecf{\mu}{ v v^\top }} \leq \delta \norm{v}_2^2$, an obtaining a separating direction from here is straight-forward. 

As a thought experiment, consider the case where we know the scalar quantity $v^\top \Sigma_2 v$. Then, re-weighting $\mu$ by $(v^\top \Sigma_2 v)^{2t}$ outputs a pseudo-distribution $\mu'$ such that the expected value of $v^\top \Sigma_2 v$ remains large, even when multiplied by low-degree sum-of-squares polynomials, in particular, the $q_i$'s and therefore $\pexpecf{\mu}{ q_i(w,v) v^\top \Sigma_2 v } \geq 0.1$ (see~\cref{lem:fixing-the-quad-form}).  Therefore, we can now simply perform Gaussian rounding on $\pexpecf{\mu'}{ vv^\top }$. Of course, re-weighting by $v^\top \Sigma_2 v$ is as hard as learning $\Sigma_2$, but if we pick a uniformly random sample, and condition on $x_i$ being from $\calD(0, \Sigma_2)$, in expectation, we are re-weighting $v^\top \Sigma_2 v$ and the above argument continues to work (see \cref{lem:simulating-fixing-the-quad-form} for a full proof).

\subsection{Related work}

Anti-concentration is a vast field and has been subject to intense study in probability theory (e.g. \cite{erdos1945lemma, costello2006random, lovett2010elementary, bakshi2021list}) and statistics (e.g. \cite{li2001gaussian, chernozhukov2015comparison}). Similar to concentration inequalities \cite{tropp2015:book, medarametla2016bounds, rajendran2023concentration}, anti-concentration inequalities have found numerous applications in computer science, e.g. \cite{simchowitz2015dictionary, karmalkar2019list, raghavendra2020list, aaronson2011computational, chakrabarti2011optimal, vidick2012concentration, sherstov2012communication, nie2022matrix}. 
Among such applications, the works most closely related to ours are the ones exploiting anticoncentration certificates to design approximation algorithms \cite{karmalkar2019list, raghavendra2020list, bakshi2020outlier, raghavendra2020list2, diakonikolas2020robustly, bakshi2021list, ivkov2022list}.
Notably, many of these certificates are based on the Sum of Squares (SoS) hierarchy. In particular, these SoS certificates of anti-concentration have been modified in various ways to design statistical algorithms. However, they suffer from two drawbacks -- Most of the programs used have been application-specific, requiring ad-hoc analyses designed for the problem at hand; and for technical reasons, they work only for spherical distributions, e.g. Gaussians. In this work, we present a unified and the most canonical way to certify anti-concentration phenomena for a wide class of distributions that go beyond spherical distributions.

The sum of squares hierarchy (SoS) \cite{shor1987approach, nesterov2000squared, parrilo2000structured, grigoriev2001complexity} uses insights from proof complexity to design algorithms, leading to the powerful framework known as ``proofs to algorithms'' \cite{FKP19}. This has been highly effective in theoretical computer science, capturing state-of-the-art approximation algorithms for many
problems, e.g. \cite{AroraRV04, GW94, hopkins2015tensor, Raghavendra08}. Due to its effectiveness in capturing many algorithmic reasoning techniques \cite{FKP19}, lower bounds against SoS is also an active area of study \cite{BHKKMP16, KothariMOW17, sklowerbounds, potechin2020machinery, rajendran2022nonlinear, kunisky2020positivity, potechin2022sub, jones2022symmetrized, jones2022sum, jones2023sum}. 
In recent years, many works have exploited SoS hierarchy for applications in algorithmic robust statistics \cite{diakonikolas2023algorithmic}, leading to to breakthrough algorithms for
long-standing open problems \cite{bakshi2022robustly, liu2021settling, hopkins2020mean, klivans2018efficient, FKP19, kothari2017outlier, bakshi2020outlier, bakshi2021list, schramm2017fast}. Highlights include robustly learning mixtures of high dimensional Gaussians (a long line of works culminating in \cite{bakshi2022robustly, liu2021settling}), efficient algorithms for the fundamental problems of regression \cite{klivans2018efficient}, moment estimation \cite{kothari2017outlier}, clustering \cite{bakshi2020outlier} and subspace recovery \cite{bakshi2021list} in the presence of outliers. In recent years, many more applications are being discovered, e.g. in differential privacy \cite{hopkins2022efficient, georgiev2022privacy, kothari2022private,   hopkins2023robustness} and 
quantum information \cite{barak2012hypercontractivity, barak2017quantum, bakshi2023learning}.

Finally, we discuss works related to our applications. The literature on clustering and regression is incredibly vast.
These fundamental problems have many connections to other standard statistical tasks such as mixture modeling, density estimation, latent variable modeling, etc. For the most closely related works on clustering and regression, see \cite{bakshi2020outlier, liu2022clustering, diakonikolas2022clustering, karmalkar2019list, diakonikolas2017statistical, kothari2017better, hopkins2018mixture, chen2020learning, raghavendra2020list} and citations therein.

\paragraph{Organization of the paper.}

In \cref{sec: prelims}, we present the preliminaries of the sum-of-squares framework and define the distributional properties we use. In \cref{sec: anticonc_program}, we present our main anti-concentration program and state our main theorem, with the proof details filled out in Sections \ref{sec: anticonc_program}, \ref{sec:dense-cert-anti-conc} and \ref{sec:re-weighting-pseudo-dist}.
We then show the application of our program to spectral clustering in \cref{sec:clustering-warmup},
$k$-clustering in \cref{sec:clustering-with-all-separation} and list decodable regression in
\cref{sec:list-decodable-regression}. 
% Finally, in \cref{sec:distribution_families}, we show that a natural class of distributions (beyond Gaussian), namely uniform distributions on $\ell_p$ balls, satisfy our required distributional properties.\gnote{TBD}

%%% Local Variables:
%%% mode: latex
%%% TeX-master: "main"
%%% End:

%% file: prelims.tex
\section{Preliminaries}\label{sec: prelims}

\subsection{The Sum-of-Squares framework}
\label{subsec:sos-framework}

We now provide an overview of the Sum-of-Squares proof system.
We closely follow the exposition as it appears in the lecture notes of Barak~\cite{barak2016proofs}.   

\paragraph{Pseudo-Distributions.}
A discrete probability distribution over $\R^m$ is defined by its probability mass function, $D\from \R^m \to \R$, which must satisfy $\sum_{x \in \mathrm{supp}(D)} D(x) = 1$ and $D \geq 0$.
We extend this definition by relaxing the non-negativity constraint to merely requiring that $D$ passes certain low-degree non-negativity tests.
We call the resulting object a pseudo-distribution.

\begin{definition}[Pseudo-distribution]
A \emph{degree-$\ell$ pseudo-distribution} is a finitely-supported function $D:\R^m \rightarrow \R$ such that $\sum_{x} D(x) = 1$ and $\sum_{x} D(x) p(x)^2 \geq 0$ for every polynomial $p$ of degree at most $\ell/2$, where the summation is over all $x$ in the support of $D$.
\end{definition}
Next, we define the related notion of pseudo-expectation.
\begin{definition}[Pseudo-expectation]
The \emph{pseudo-expectation} of a function $f$ on $\R^m$ with respect to a pseudo-distribution $\mu$, denoted by $\pexpecf{\mu(x)}{f(x)}$,  is defined as
\begin{equation*}
  \pexpecf{\mu(x)}{f(x)} = \sum_{x} \mu(x) f(x).
\end{equation*}
\end{definition}
We use the notation $\pexpecf{\mu(x)}{(1,x_1, x_2,\ldots, x_m)^{\otimes \ell}}$ to denote the degree-$\ell$ moment tensor of the pseudo-distribution $\mu$.
In particular, each entry in the moment tensor corresponds to the pseudo-expectation of a monomial of degree at most $\ell$ in $x$. 

\begin{definition}[Constrained pseudo-distributions]
\label{def:constrained-pseudo-distributions}
Let $\mu$ be a degree-$\ell$ pseudo-distribution over $\mathbb{R}^d$. Let $\calA = \Set{ p_1\geq 0 , p_2\geq0 , \ldots p_m\geq 0}$ be a system of $m$ polynomial inequality constraints of degree at most $r$. We define $\mu$ to satisfy $\calA$ at degree $\ell \ge1$ if for every subset $\calS \subset [m]$, and every sum-of-squares polynomial $q$ such that $\textrm{deg}(q) + \sum_{i \in \calS } \max\Paren{p_i, r } \leq \ell$, $\pexpecf{\mu}{ q \prod_{i \in \calS} p_i } \geq 0$. Further, we define $\mu$ to \textit{ approximately satisfy} the system of constraints $\calA$ if the above inequalities are satisfied up to additive error $2^{-n^{\ell} } \norm{q} \prod_{i \in \calS} \norm{p_i}$, where $\norm{\cdot}$ denotes the euclidean norm of the coefficients of the polynomial, represented in the monomial basis.  
\end{definition}

Crucially, there's an efficient separation oracle for moment tensors of constrained pseudo-distributions. 

\begin{fact}[\cite{shor1987approach, nesterov2000squared, parrilo2000structured, grigoriev2001complexity}]
 \label{fact:sos-separation-efficient}
  For any $m,\ell \in \N$, the following set has a $m^{O(\ell)}$-time weak separation oracle (in the sense of \cite{grotschel1981ellipsoid}):
  \begin{equation*}
    \Set{  \pexpecf{\mu(x)} { (1,x_1, x_2, \ldots, x_m)^{\otimes \ell } } \Big\vert \text{ $\mu$ is a degree-$\ell$ pseudo-distribution over $\R^m$}}
  \end{equation*}
\end{fact}
This fact, together with the equivalence of weak separation and optimization \cite{grotschel1981ellipsoid} forms the basis of the sum-of-squares algorithm, as it allows us to efficiently approximately optimize over pseudo-distributions. 

Given a system of polynomial constraints, denoted by $ \calA$, we say that it is \emph{explicitly bounded} if it contains a constraint of the form $\{ \|x\|^2 \leq M\}$. Then, the following fact follows from  \cref{fact:sos-separation-efficient} and \cite{grotschel1981ellipsoid}:
%\ewin{If we only use the definition of explicitly bounded here, we should just not define it.}

\begin{fact}[Efficient Optimization over Pseudo-distributions]
There exists an $(m+t)^{O(\ell)} $-time algorithm that, given any explicitly bounded and satisfiable system $ \calA$ of $t$ polynomial constraints in $m$ variables, outputs a level-$\ell$ pseudo-distribution that satisfies $ \calA$ approximately (see~\cref{def:constrained-pseudo-distributions}) \footnote{Here, we assume that the bit complexity of the constraints in $ \calA$ is $(m+t)^{O(1)}$.}\label{fact:eff-pseudo-distribution}.
\end{fact}

We now define basic facts for pseudo-distributions, which extend facts that hold for standard probability distributions, which can be found in several prior works listed above.

\begin{fact}[Cauchy-Schwarz for Pseudo-distributions]
Let $f,g$ be polynomials of degree at most $d$ in indeterminate $x \in \R^d$. Then, for any degree d pseudo-distribution $\mu$,
$\pexpecf{\mu}{fg}  \leq \sqrt{ \pexpecf{\mu}{ f^2} } \cdot \sqrt{ \pexpecf{\mu}{ g^2} }$.
 \label{fact:pseudo-expectation-cauchy-schwarz}
\end{fact}

\begin{fact}[Hölder's Inequality for Pseudo-Distributions] \label{fact:pseudo-expectation-holder}
Let $f,g$ be polynomials of degree at most $d$ in indeterminate $x \in \R^d$. 
Fix $t \in \N$. Then, for any degree $dt$ pseudo-distribution $\mu$,
\begin{equation*}
    \pexpecf{\mu}{ f^{t-1}  g} \leq \Paren{ \pexpecf{\mu}{ f^t }  }^{\frac{t-1}{t}} \cdot  \Paren{  \pexpecf{\mu }{ g^t }  }^{\frac{1}{t}}.  
\end{equation*}
In particular, for all even integers $k$,
$\pexpecf{\mu}{f}^k \leq \pexpecf{\mu}{ f^k }$.
\end{fact}

% \paragraph{Reweighting Pseudo-Distributions}

% The following fact is easy to verify and has been used in several works (see~\cite{DBLP:conf/stoc/BarakKS17} for example).  
% \begin{fact}[Reweightings] \label{fact:reweightings}
% Let $\tmu$ be a pseudo-distribution of degree $k$ satisfying a set of polynomial constraints $ \calA$ in variable $x$. 
% Let $p$ be a sum-of-squares polynomial of degree $t$ such that $\pE[p(x)] \neq 0$.
% Let $\tmu'$ be the pseudo-distribution defined so that for any polynomial $f$, $\pE_{\tmu'}[f(x)] = \pE_{\tmu}[ f(x)p(x)]/\pE_{\tmu}[p(x)]$. Then, $\tmu'$ is a pseudo-distribution of degree $k-t$ satisfying $ \calA$. 
% \end{fact}

\paragraph{Sum-of-squares proofs.}

Let $f_1, f_2, \ldots, f_r$ and $g$ be multivariate polynomials in the indeterminate $x$.
Given constraints $\{f_1 \geq 0, \ldots, f_m \geq 0\}$,   a \emph{sum-of-squares proof} that implies the identity $\{g \geq 0\}$ consists of  polynomials $(p_S)_{S \subseteq [m]}$ such that
\begin{equation*}
g = \sum_{S \subseteq [m]} p^2_S \cdot \prod_{i \in S} f_i
\end{equation*}
We say that this proof has \emph{degree-$\ell$} if for every set $S \subseteq [m]$, the polynomial $p^2_S \Pi_{i \in S} f_i$ has degree at most $\ell$.
If there is a degree $\ell$ SoS proof that $\{f_i \geq 0 \mid i \leq r\}$ implies $\{g \geq 0\}$, we write:
\begin{equation}
  \{f_i \geq 0 \mid i \leq r\} \sststile{\ell}{}\{g \geq 0\}
\end{equation}
For all polynomials $f,g\colon\R^n \to \R$ and for all functions $F\colon \R^n \to \R^m$, $G\colon \R^n \to \R^k$, $H\colon \R^{p} \to \R^n$ such that each of the coordinates of the outputs are polynomials of the inputs, we have the following inference rules.

The first one derives new inequalities by addition/multiplication:
\begin{equation} \label{eq:sos-addition-multiplication-rule}
\frac{ \calA \sststile{\ell}{} \{f \geq 0, g \geq 0 \} } { \calA \sststile{\ell}{} \{f + g \geq 0\}}, \frac{ \calA \sststile{\ell}{} \{f \geq 0\},  \calA \sststile{\ell'}{} \{g \geq 0\}} { \calA \sststile{\ell+\ell'}{} \{f \cdot g \geq 0\}} \tag{Addition/Multiplication Rule} 
\end{equation}
The next one derives new inequalities by transitivity: 
\begin{equation} \label{eq:sos-transitivity}
\frac{ \calA \sststile{\ell}{}  \calB,  \calB \sststile{\ell'}{} C}{ \calA \sststile{\ell \cdot \ell'}{} C}   \tag{Transitivity Rule} 
\end{equation}
Finally, the last rule derives new inequalities via substitution:
\begin{equation} \label{eq:sos-substitution}
\frac{\{F \geq 0\} \sststile{\ell}{} \{G \geq 0\}}{\{F(H) \geq 0\} \sststile{\ell \cdot \deg(H)} {} \{G(H) \geq 0\}}\tag{Substitution Rule} 
\end{equation}

% Low-degree sum-of-squares proofs are sound and complete if we take low-level pseudo-distributions as models.
Sum-of-squares proofs allow us to deduce properties of pseudo-distributions that satisfy some constraints.
\begin{fact}[Soundness]
  \label{fact:sos-soundness}
  If $\mu \sdtstile{\ell}{}  \calA$ for a level-$\ell$ pseudo-distribution $\mu$ and there exists a sum-of-squares proof $ \calA \sststile{\ell'}{}  \calB$, then $\mu \sdtstile{\ell \cdot \ell'+\ell'}{}  \calB$.
\end{fact}
If the pseudo-distribution $D$ satisfies $ \calA$ only approximately, soundness continues to hold if we require an upper bound on the bit-complexity of the sum-of-squares $ \calA \sststile{r'}{} B$  (number of bits required to write down the proof). In our applications, the bit complexity of all sum of squares proofs will be $n^{O(\ell)}$ (assuming that all numbers in the input have bit complexity $n^{O(1)}$). This bound suffices in order to argue about pseudo-distributions that satisfy polynomial constraints approximately.

The following fact shows that every property of low-level pseudo-distributions can be derived by low-degree sum-of-squares proofs.
\begin{fact}[Completeness]
  \label{fact:sos-completeness}
  Suppose $d \geq r' \geq r$ and $ \calA$ is a collection of polynomial constraints with degree at most $r$, and $ \calA \vdash \{ \sum_{i = 1}^n x_i^2 \leq B\}$ for some finite $B$.
  Let $\{g \geq 0 \}$ be a polynomial constraint.
  If every degree-$d$ pseudo-distribution that satisfies $D \sdtstile{r}{}  \calA$ also satisfies $D \sdtstile{r'}{} \{g \geq 0 \}$, then for every $\epsilon > 0$, there is a sum-of-squares proof $ \calA \sststile{d}{} \{g \geq - \epsilon \}$.
\end{fact}

Next, we recall the fact that any univariate polynomial inequality admits a sum-of-squares proof over the reals. 
\begin{fact}[Univariate Polynomial Inequalities admit SoS Proofs,~\cite{laurent2009sums}]
\label{fact:univariate-sos-proofs}
Given a univariate degree-$t$ polynomial $p$ such that for all $x \in \mathbb{R}$, $p(x)\geq 0$, we have $\sststile{d}{x} \Set{ p(x) \geq 0 } $. Further, if for all $ a \leq x \leq b$, we have $p(x) \geq 0$, then $\Set{ x\geq a, x\leq b } \sststile{d}{x} \Set{ p(x) \geq 0 }$.
\end{fact}

\begin{fact}\label{fact:nonnegative-quadratic}(Non-negative Quadratic Polynomials Inequalities admit SoS Proofs)
Given a multivariate polynomial $p(x_1, x_2, \ldots , x_m) \in \mathbb{R}[x_1, \ldots, x_m]$ such that $p$ has degree $2$ and $p(x_1, x_2, \ldots , x_m) \geq 0$ for all $x_1, x_2, \ldots x_m \in \mathbb{R}$. Then, $\sststile{2}{x} \Set{ p(x) \geq 0  }$. 
\end{fact}
\begin{proof}
Note that there is a unique $(m + 1) \times (m+1 )$ symmetric matrix $M$ such that if we let $v(x) = (1,x_1, \dots , x_m)^\dagger$ then
\[
p(x_1, \dots , x_m) = v(x)^\dagger M v(x) \,.
\]
Now we claim that $M$ must be PSD.  Indeed for any vector $v  = (v_1, \dots , v_{m+1}) \in \R^{m+1}$, if $v_1 \neq 0$, then we can consider plugging in $x_1 = v_2/v_1, \dots , x_m = v_{m+1}/v_1$ and since $p(x_1, \dots , x_m) \geq 0$ we get that $v^\dagger M v \geq 0$.  If $v_1 = 0$, then we can consider plugging in $x_1 = cv_2, \dots , x_m = cv_{m+1}$ and take the limit as $c \rightarrow \infty$ and again we deduce that $v^\dagger M v \geq 0$.  Thus, $M$ must be PSD.  We can now write $M = \sum_{i = 1}^{m+1} u_iu_i^\dagger$ for some vectors $u_i \in \R^{m+1}$.  Thus, we can write
\[
p(x_1, \dots , x_m) = v(x)^\dagger M v(x) = \sum_{i = 1}^{m+1} \langle u_i, v(x)\rangle^2
\]
which is a degree-$2$ SoS polynomial and we are done.
\end{proof}

\paragraph{Basic Sum-of-Squares Proofs.} Next, we use the following basic facts regarding sum-of-squares proofs. For further details, we refer the reader to a recent monograph~\cite{fleming2019semialgebraic}.

\begin{fact}[Operator norm Bound]
\label{fact:operator_norm}
Let $A$ be a symmetric $d\times d$ matrix and $v$ be a vector in $\mathbb{R}^d$. Then,
\[
\sststile{2}{v} \Set{ v^{\top} A v \leq \norm{A}_2\norm{v}_2^2}.
\]
\end{fact}

\begin{fact}[SoS Almost Triangle Inequality] \label{fact:sos-almost-triangle}
Let $f_1, f_2, \ldots, f_r$ be indeterminates. Then,
\[
\sststile{2t}{f_1, f_2,\ldots,f_r} \Set{ \Paren{\sum_{i\leq r} f_i}^{2t} \leq r^{2t-1} \Paren{\sum_{i =1}^r f_i^{2t}}}.
\]
\end{fact}

\begin{fact}[SoS Hölder's Inequality]\label{fact:sos-holder}
Let $w_1, \ldots w_n$ be indeterminates and let $f_1,\ldots f_n$ be polynomials of degree $m$ in vector valued variable $x$. 
Let $k$ be a power of 2.  
Then, 
\[
\Set{w_i^2 = w_i, \forall i\in[n] } \sststile{2km}{x,w} \Set{  \Paren{\frac{1}{n} \sum_{i = 1}^n w_i f_i}^{k} \leq \Paren{\frac{1}{n} \sum_{i = 1}^n w_i}^{k-1} \Paren{\frac{1}{n} \sum_{i = 1}^n f_i^k}}. 
\]
\end{fact}

\begin{fact}[Almost square-root]
\label{fact:squared-value-to-magntitude}
Given a scalar indeterminate $v$, $   \Set{v^2 \leq 1 } \sststile{}{v} \Set{  -1 \leq v\leq 1  }$.
\end{fact}
\begin{proof}
We know that $\Set{ \Paren{1-v}^2 = 1 + v^2 -2v \geq 0 }$ and  $\Set{ \Paren{1 + v}^2 = 1 + v^2 + 2v \geq 0 }$. Further, by assumption, $\Set{ 1-v^2 \geq 0}$ and by the addition rule we have $\Set{ 2 + 2v\geq0  }$ and  $\Set{ 2 - 2v \geq0  }$. Rearranging yields the claim. 
\end{proof}

\subsection{Re-weighting Pseudo-distributions}
\label{subsec:re-weighting-background}
Let $\mu$ be a pseudo-distribution of degree $\geq t$ on $d$-dimensional vector valued indeterminate $v$. Let $q$ be a sum-of-squares polynomial in $v$ of degree $t'<t$. Then, $\mu'$ defined by $\mu'(v) = \mu(v) \cdot q(v)$ is called the \emph{reweighting} of $\mu$ by the polynomial $q$. It is easy to observe that $\mu'$ is a pseudo-distribution of degree at least $t-t'$. Further, if $\mu$ satisfying a polynomial equality constraint $\{r=0\}$ and degree of $r$ is at most $t-t'$, then, so does $\mu'$.  See \cite{barak2017quantum} for background on reweighting pseudo-distributions.

We also require the following straightforward fact that ensures subsequent re-weightings of pseudo-distributions do not decrease the pseudo-expectation of the indeterminate:

\begin{fact}[Scalar Re-weightings are Monotone]
\label{fact:scalar-reweighting-monotone}
Let $\mu$ be a pseudo-distribution over an indeterminate $z$. Let $\mu'$ be a re-weighting of $\mu$ obtained by using the polynomial $z^{2\ell}$, for some $\ell \in \mathbb{N}$. Then, $\pexpecf{\mu'}{z^2} \geq \pexpecf{\mu}{ z^2}$, given $\mu$ is has degree at least $3\ell+3$.     
\end{fact}
\begin{proof}
Recall, by the definition of re-weighting a pseudo-distribution, we have
\begin{equation}
    \pexpecf{\mu'}{ z^2 } = \frac{ \pexpecf{\mu}{ z^2 \cdot z^{2\ell} } }{ \pexpecf{\mu}{ z^{2\ell}} }. 
\end{equation}
By H\"older's inequality for pseudo-distributions (Fact~\ref{fact:pseudo-expectation-holder} applied with $f=z^2$ and $k=\ell$ ), we have $\pexpecf{\mu}{z^2 }  \leq \Paren{ \pexpecf{\mu}{ z^{2\ell} } }^{1/\ell}$
and thus
\begin{equation}
\label{eqn:holder's-to-z-2l}
\begin{split}
    \pexpecf{\mu }{ z^2  } \pexpecf{\mu}{z^{2\ell}} & \leq   \pexpecf{\mu}{ z^{2\ell}}^{\frac{\ell+1}{\ell}}  \leq \Paren{ \pexpecf{\mu}{ z^{2\ell + 2} }^{\frac{2\ell}{2\ell+2 }} \cdot \pexpecf{\mu}{1}^{\frac{2}{2\ell+2}}   }^{\frac{\ell+1}{\ell}}  = \pexpecf{\mu}{ z^{2\ell +2} },
\end{split}
\end{equation}
where the second inequality follows from applying H\"older's inequality again.
Therefore, we can rearrange Equation~\eqref{eqn:holder's-to-z-2l} and conclude $\pexpecf{\mu'}{ z^2  } \geq \pexpecf{\mu}{z^2}$.
\end{proof}

\begin{fact}[Bounding moments of subtracting a coordinate]
\label{fact:moment-inequality-sos}
Let $v$ be a vector-valued indeterminate over $\mathbb{R}^d$ and $y  \in \mathbb{N}$. Then, for any $i \in [d]$,
\begin{equation*}
    \sststile{2y}{v} \Set{  \Paren{ \norm{ v}^2_2 - v_i^2  }^{y} \leq \norm{v}_2^{2y} - v_i^{2y}  }
\end{equation*}
\end{fact}
\begin{proof}
We use the fact that $\sststile{2t}{a,b} \Set{ a^{2t} + b^{2t} \leq \Paren{a+b}^{2t} }$ as follows:
    \begin{equation}
        \sststile{2y}{v} \Set{  \Paren{ \norm{ v}^2_2 - v_i^2  }^{2y}   + v_i^{2y} \leq \Paren{ \norm{v}_2^2 - v_i^2 + v_i^2  }^{2y} =  \norm{v}_2^{2y} }
    \end{equation}
Rearranging yields the claim.
\end{proof}

% Given vectors $a, b \in \mathbb{N}^d$ introduce the notation $b \preceq a$ to denote $a$ majorizes $b$, i.e. for all $k \in [d]$, $\sum_{i \in [k]} a_i \geq \sum_{i \in [k]} b_i$.  Next, we introduce Reynolds operator $\mathcal{R}$ of the symmetric group
% $\mathfrak{S}_n$, given a polynomial $f(x_1, x_2, \ldots, x_n)$, $\mathcal{R}(f) = \frac{1}{n!} \sum_{\sigma \in \mathfrak{S}_n } f\Paren{ x_{\sigma(1)}, \ldots, x_{\sigma(n)} } $. We then define $x^{a} = \prod_{j \in [d]} x_j^{a_j} $ and $[a] = \mathcal{R} x^a$.

% \begin{fact}[Muirhead Inequality~\cite{}]
% \label{fact:muirhead-sos}
% For non-negative variables $x_1,\ldots, x_n$, if $b\preceq a$, $$[a]-[b] = \sum_{j_1 =0 }^1 \sum_{j+2 = 0}^1 \ldots \sum_{j_n=0}^1 \textrm{sos}(x) \prod_{ i \in [n] } x^{j_i }_i.$$ 
% \end{fact}

\subsection{Distributions}

% \begin{definition}[$(c_k, k)$-certifiably hypercontractive]
% \gnote{Fill in; definition copy pasted below}
% \end{definition}
% \ainesh{fix this section}

\begin{definition}[Certifiable Hypercontractivity of linear forms]
\label{def:certifiable-hypercontractivity-of-linear}
An  distribution $\calD$ on $\R^d$ is said to be $h$-certifiably $c_h$-hypercontractive if there's a degree-$h$ sum-of-squares proof of the following unconstrained polynomial inequality in $d$-dimensional vector-valued indeterminate $v$:
% \[
% \sststile{h}{Q} \Set{ \E_{x \sim \cD} \iprod{xx^{\top},Q}^{h} \leq \paren{Ch}^{h} \Paren{\E_{x \sim \cD} \iprod{xx^{\top},Q}^2}^{h/2}}\mper
% \]
\[ \E_{x \sim \calD} \iprod{x,v}^{2h} \leq \paren{c_h}^{2h} \Paren{\E_{x \sim \calD} \iprod{ x,v }^2 }^{h},
\]
for $h$ being an even integer.
A set of points $X \subseteq \R^d$ is said to be  $(c_h, h)$-certifiably hypercontractive if the uniform distribution on $X$ is  $h$-certifiably $c_h$-hypercontractive.
\end{definition}

\begin{definition}[Certifiable Hypercontractivity of quadratic forms]
\label{def:certifiable-hypercontractivity}
An isotropic distribution $\calD$ on $\R^d$ is said to be $h$-certifiably $c_h$-hypercontractive if there's a degree $h$ sum-of-squares proof of the following unconstrained polynomial inequality in $d \times d$ matrix-valued indeterminate $Q$:
% \[
% \sststile{h}{Q} \Set{ \E_{x \sim \cD} \iprod{xx^{\top},Q}^{h} \leq \paren{Ch}^{h} \Paren{\E_{x \sim \cD} \iprod{xx^{\top},Q}^2}^{h/2}}\mper
% \]
\[ \E_{x \sim \calD} \Paren{x^{\top}Qx - \expecf{x\sim \calD}{ x^\top Q x} }^{2h} \leq \paren{c_h}^{2h} \Paren{\E_{x \sim \calD} \Paren{x^{\top}Qx - \expecf{x\sim \calD}{ x^\top Q x} }^2}^{h},
\]
for $h$ being an even integer.
A set of points $X \subseteq \R^d$ is said to be  $(c_h, h)$-certifiably hypercontractive if the uniform distribution on $X$ is  $h$-certifiably $c_h$-hypercontractive.
\end{definition}

\begin{definition}[Almost $k$-wise independent]
\label{def:k-wise-independent}
For an integer $k \ge 1$, define a distribution $\calD$ to be almost $k$-wise independent if for all subsets $S \subseteq [n]$ of size at most $k$, we have $\EE_{x \sim \calD} [\prod_{i \in S} x_i^2] = (1 \pm \gam_S) \prod_{i \in S} \EE_{x \sim \calD}[x_i^2]$ where $\gam_S = o_d(1)$.
\end{definition}

\begin{definition}[Strictly sub-exponential distribution]
\label{def:strictly-sub-exponential}
A distribution $\calD$ with mean $\mu$ and covariance $\Sigma$ is $\eps$-strictly sub-exponential if for all $v \in \R^d$, 
\begin{equation*}
    \Pr[ \abs{ \Iprod{x - \mu, v} } \geq t \sqrt{ v^\top \Sigma v} ] \leq \exp\Paren{ -t^{1+\eps} / c  },
\end{equation*}
for some constant $c$. 
\end{definition}

\begin{fact}[Sub-Exponential Distributions are bounded]
\label{fact:sub-exponential-dist-bounded}
Let $x\sim \calD$ be a sample from a $d$-dimensional sub-exponential distribution with mean $0$ and covariance $I$. Then, for any $0<\delta<1$, with probability at least $1-\delta$, for all $u \in \mathbb{R}^d$,  $\Iprod{x , u}^2 \leq \log(1/\delta) \norm{u}_2^2$.  
\end{fact}

The main definition we consider in this paper is as follows:

\begin{definition}[$(\delta,k, \eps)$-Reasonably anti-concentrated distributions]
\label{def:reasonably-anti-concentrated-dist}
A distribution $\calD$ with mean $0$ and covariance $\Sigma$ is $(\delta,k,\eps)$-reasonably anti-concentrated if it satisfies the following:
\begin{enumerate}
    \item $\calD$ is $\delta$-anti-concentrated
    \item $\calD$ is $\eps$-strictly sub-exponential.
    \item $\calD$ has $\exp\Paren{ (1/\delta^2)^{ \exp( 1/\delta^4) } }$-certifiably $O(1)$-hypercontractive linear forms. 
    \item $\calD$ is almost $k$-wise independence.
\end{enumerate}
\end{definition}

% \gnote{We need stritly sub-exponential tails, so we can use Hanson-Wright}

The following lemma shows that bounded product distributions, a natural class of non-spherical distributions are reasonably anti-concentrated.

\begin{lemma}[Product Distributions]
\label{thm:product-distributions-are}
A $\delta$-anti-concentrated product distribution $\calD$ with \textit{strictly sub-exponential} coordinate  marginals is \textit{reasonably anti-concentrated}.
\end{lemma}
\begin{proof}
The coordinates are $k$-wise independent by definition.
To show strict sub-exponentiality, we use an analogue of the results from \cite[Lemma 3.4.2]{vershynin2020high} which states that sub-gaussian norm of the vector is bounded by the maximum sub-gaussian norm of each coordinate (upto a constant). This argument directly extends to sub-exponential norm (for instance, see \cite[Exercise 2.7.3]{vershynin2020high}).
Certifiable hypercontractivity of linear forms of product distributions follows from~\cite{kothari2018agnostic}.
\end{proof}

\subsection{Polynomial approximation to the box indicator}
\label{subsec:boxpolynomial}

% \ainesh{add all the polynomial facts we need.}
% \gnote{I added most of the facts we use, may need a 2-minute pass-through to ensure we have the right kind of statements.}

Our arguments hinge on polynomial approximations to the box indicator function $\mathbb{1}(x \in [\zeta - \delta, \zeta + \delta])$. Standard approximation theory provides us versatile tools to construct such polynomials and indeed, this has been used extensively in related prior works.

In this work, we work with the following explicit polynomial: 

\begin{definition}[Box Indicator Polynomial]
\label{def:box-indicator}
For $\eta,\zeta, \delta >0$ and $L \geq 1$, choose an even $d$ and let $p$ be a degree $O(d)$ polynomial such that 
\begin{equation*}
    p(x) = \eta  T_d\Paren{ \frac{ (x-\zeta)^2}{ L^2 } - 1 - \delta^2 },
\end{equation*}
where $T_d$ is the $d$-th Chebyshev polynomial of the first kind. We assume that the constants are chosen sufficiently small for the following facts to hold. 
\end{definition}

\begin{fact}
    If $1 \le x \le 1.4$,  and we write $x = 1 + \eps$, then we have $T_d(x) = e^{\Theta(d\sqrt{\eps})}$
\end{fact}

\begin{proof}
    Note that for $y = O(1)$, $\cosh(y) = 1 + O(y^2)$. Therefore, setting $\cosh(y) = x$ gives $\arccosh x = \Theta(\sqrt{\eps})$ which implies
	\[T_d(x) = \cosh(d\arccosh x) = \frac{e^{d\arccosh x} + e^{-d\arccosh x}}{2} = e^{\Theta(d\sqrt{\eps})}\]
\end{proof}

\begin{lemma}
    When $|x - \zeta| \le \Delta L < \delta L$, we have that $p(x) \ge \Omega(\eta e^{\Theta(d\sqrt{\delta^2 + \Delta^2})})$.
\end{lemma}

\begin{proof}
    When $x \in (\zeta-\Delta L, \zeta+ \Delta L)$, we have $p(x) = \eta T_d(y)$ where $y$ lies in the range $[-1-\delta^2, -1-\delta^2+\Delta^2]$ but $T_d$ monotonically decreases in this range, therefore $p(x) \ge \eta T_d(-1-\delta^2 + \Delta^2) \ge \eta e^{\Theta(d\sqrt{\delta^2 + \Delta^2})}$.
\end{proof}

\begin{lemma}\label{lem: cheby_coeff_bound}
    For constants $\eta, \zeta, \delta, L$, the coefficients of $p(x)$ are bounded by $2^{O(d)}$.
\end{lemma}

\begin{proof}
We have
\[T_d(y) = \sum_{k = 0}^{d/2} \binom{d}{2k} (y^2-1)^ky^{n-2k}\]
Note that we have each binomial coefficient above being at most $2^d$. Finally, we plug in $y = \frac{(x - \zeta)^2}{L^2} - 1 - \delta^2$ and expand further using the binomial theorem. Using the same bounds and putting them together implies the result.
\end{proof}

\begin{fact}
    If $x \le \zeta - \sqrt{2 + \delta^2}L$ or $x \ge \zeta + \sqrt{2 + \delta^2}L$, we have $p(x) \le \eta \left(2\frac{x - \zeta}{L}\right)^{2d}$
\end{fact}

\begin{proof}
    We have $p(x) = \eta T_d(y)$ where $y = (x - \zeta)^2/L^2 - 1 - \delta^2$. For the given range of $x$, we have $y \ge 1$ and therefore,
    \begin{align*}
    T_d(y) &= \frac{1}{2}((y - \sqrt{y^2 - 1})^d + (y + \sqrt{y^2 - 1})^d)\\
    &\le (y + \sqrt{y^2 - 1})^d\\
    &\le (2y)^d\le \left(2\frac{x - \zeta}{L}\right)^{2d}
\end{align*}
which implies the bound.
\end{proof}

\begin{lemma}[Gaussian anti-concentration with shifts]
\label{thm:anti-conc-with-shifts-gaussian}
Given $\delta, \zeta >0$, let $L = 1/\delta + \zeta, \eta = \delta / L$, and let  $p$ be as above. Then, 
\begin{equation*}
    \expecf{x\sim \calN(0,1)}{ p^2\Paren{ x-\zeta } } \leq \bigO{ \delta },
\end{equation*}
when $d = \Omega\Paren{ \log\Paren{ \frac{1}{\eta \delta} + \frac{\zeta}{\delta} }  \cdot \Paren{ \frac{1}{ 1/\delta^2  } + \zeta^2 } } $.
\end{lemma}
% \ainesh{we should be able to generalize this theorem to hold for $\zeta^2 \leq \Delta$, i.e. just an upper bound on the magnitude of the shift $\zeta$. }

\begin{proof}
% Without loss of generality, we assume $\zeta\geq 0$. 
% We work with the following explicit polynomial: for $\eta,\zeta, \delta >0$ and $L \geq 1$, choose an even $d$ and let $p$ be a degree $d$ polynomial such that 
% \begin{equation*}
    % p(x) = \eta  T_d\Paren{ \frac{ (x-\zeta)^2}{ L^2 } - 1 - \delta^2 },
% \end{equation*}
% where $T_d$ is the $d$-th Chebyshev polynomial of the first kind. 

The proof idea is as follows. We split $\mathbb{R}$ into the following intervals: $\calI_0 = (-\infty, \zeta-\sqrt{2 + \delta^2} L)$, $\calI_1 = [\zeta- \sqrt{2 + \delta^2} L, \zeta-\delta L]$,   $\calI_2 = (\zeta-\delta L, \zeta+ \delta L]$, $\calI_3 = (\zeta+ \delta L, \zeta + \sqrt{2 + \delta^2} L]$ and $\calI_4 = ( \zeta+ \sqrt{2 + \delta^2} L, \infty)$.
% \ainesh{the analysis for $\calI_0$ is annoying.}
For intervals $\calI_1, \calI_2, \calI_3$, we use the bound $\mu(x)   \leq 1$, where $\mu(x)$ is the pdf at $x$.
% \ainesh{this bound is loose as $\zeta$ gets large.}  
For $\calI_0$ and $\calI_4$, we use the tail estimate. 

Firstly, if $\delta^2 \leq 0.1$,
\begin{equation*}
    p(\zeta) = \eta T_d \Paren{-1-\delta^2} \leq \eta e^{\Theta(d\delta ) } 
\end{equation*}

For the interval $\calI_2$, $p(x) = \eta T_d(y)$ where $y$ lies in the range $[-1-\delta^2, -1]$ but $T_d$ monotonically decreases in this range, therefore $p(x)$ 
 lies in the interval $[\eta, \eta e^{\Theta(d\delta )}]$, implying 
\begin{equation*}
    \int_{x \in \calI_2} p^2\Paren{x} \mu(x)dx  \leq \int_{\calI_2} 1 dx  \leq 2\delta L .
\end{equation*}
Next, observe that for $x \in \calI_1 \cup \calI_3$, we have $p(x) = \eta T_d(y)$ where $y \in [-1, 1]$ but for such a $y$, we have $T_d(y) \in [-1, 1]$ (simply because $T_d(\cos\theta) = \cos (d\theta)$), which means $p(x) \in [-\eta, \eta]$. Therefore,  
\begin{equation*}
    \int_{x \in \calI_1} p^2\Paren{x} \mu(x)dx + \int_{x \in \calI_3} p^2\Paren{x} \mu(x)dx \leq \int_{x \in \calI_1 \cup \calI_3} \eta^2  \mu(x) dx  \leq 4 \eta^2 (L + \zeta)  .
\end{equation*}

Finally, for $\calI_0 \cup \calI_4$, it follows from \cite[Fact A.3]{karmalkar2019list}, and that in this interval,  $T_d\Paren{ \Paren{ \frac{x-\zeta}{L}}^2 - 1 - \delta^2 } \leq \Paren{ 2 \frac{x-\zeta}{L} }^{2d}$, 
\begin{equation*}
\begin{split}
    \int_{x \in \calI_0 \cup \calI_4} p^2\Paren{x} \mu(x)dx & \leq  \int_{x \in \calI_0 \cup \calI_4 } \frac{ \eta^2 2^{4d} }{L^{4d}}  (x+\zeta)^{4d} \exp\Paren{-\abs{x}^{\alpha} }dx  \\
    & \leq \int_{x \in \calI_0 \cup \calI_4 } \frac{ \eta^2 2^{8d} }{L^{4d}}  \cdot \Paren{  \zeta^{4d} +  x^{4d} } \exp\Paren{- \abs{x}^{\alpha} }dx  \\
    & \leq  \frac{ \eta^2 2^{8d} \zeta^{4d} }{L^{4d}} \int_{x\in \calI_0} \exp\Paren{ - \abs{x}^a } dx  + \frac{ \eta^2 2^{2d} }{L^{4d}} \exp\Paren{ -L^{\alpha} } \Paren{ L^{8d} + \Paren{32 d/\alpha}^{4d/\alpha} }
\end{split}
\end{equation*}
Then, combining all bounds above, and setting $L =1/\delta + \zeta$, and $\eta = \delta/L$, we have
\begin{equation*}
    \begin{split}
        \expecf{x\sim \calD}{ p^2\Paren{x} } & \leq  p(0) + 2 \int_{x \in \calI_1} p^2\Paren{x} \mu(x)dx + 2\int_{x \in \calI_2} p^2\Paren{x} \mu(x)dx + \int_{x \in \calI_3} p^2\Paren{x} \mu(x)dx\\
        & \leq \eta e^{\Theta(d\delta ) }  + 2\delta L + \eta^2 (L+\zeta) + \frac{ \eta^2 2^{2d} }{L^{4d}} \exp\Paren{ -L^{\alpha} } \Paren{ L^{8d} + \Paren{32 d/\alpha}^{4d/\alpha} } + \frac{ \eta^2 2^{2d} \zeta^{4d} }{L^{4d}}\\
        %& \leq O(\eta) +  \frac{L}{d} + \eta^2 L + \eta^2 \zeta   + \eta^2  \exp\Paren{ -L^{\alpha}  + 8d\log(L) + 32 \frac{d}{\alpha}\log(d/\alpha) } + \frac{ \eta^2 2^{2d} \zeta^{4d} }{L^{4d}}. 
        & \leq O(\delta),
    \end{split}
\end{equation*}
when  $d = \Omega( \log(1/\eta) \Paren{ L^2 }) = \Omega\Paren{ \log\Paren{ \frac{1}{\eta \delta} + \frac{\zeta}{\delta} }  \cdot \Paren{ \frac{1}{ 1/\delta^2  } + \zeta^2 } } $.
\end{proof}

We also have the following fact, which can be shown using the same argument as \cite{karmalkar2019list}.

\begin{fact}[Anti-Concentration for Sub-exponential distributions]
\label{lemma:box-indicator-small-for-samples}
% \ainesh{add a fact that the box indicator is small}
Given samples $\Set{x_i }_{i \in [n]}$ from a $\delta$ anti-concentrated strictly sub-exponential distribution $\calD$, it follows that
\begin{equation*}
    \frac{1}{n}\sum_{i \in [n]} p^2\Paren{\Iprod{x_i, v}^2 } \leq \delta.
\end{equation*}
\end{fact}

%% file: anti-concentration-program.tex
\section{Anti-concentration program}\label{sec: anticonc_program}

In this section, we introduce a polynomial system that certifies an upper bound on anti-concentration for a set of $n$ samples drawn from a \textit{reasonably anti-concentrated} distribution (\cref{def:reasonably-anti-concentrated-dist}).  
Formally, given $n$ i.i.d. samples $\Set{x_i}_{i \in [n]}$ drawn from a \textit{reasonably anti-concentrated} distribution $\calD$ with mean $0$ and covariance $I$,  and $\delta>0$, consider the following program:
% \begin{equation}
% \label{eqn:cons-system}
% \begin{split}
%     & \max_{w \in \mathbb{R}^n , v\in \mathbb{R}^d  } \hspace{0.2in} \frac{1}{n} \sum_{ i \in [n]} w_i  
%     \\
%     & \begin{aligned}
%       &\forall i\in [n]
%       & w_i^2
%       & = w_i \\
%       &\forall i\in [n] & w_i \Iprod{x_i, v}^2 & \leq w_i \delta^2  \\
%       && \norm{v}_2^2 &=1 \\
%     \end{aligned}
% \end{split}
% \end{equation}
\begin{equation}
\label{eqn:cons-system}
\begin{split}
    %& \hspace{0.2in} \max_{w \in \mathbb{R}^n , v\in \mathbb{R}^d  } \hspace{0.2in} \frac{1}{n} \sum_{ i \in [n]} w_i  
    %\\
    \calA_\delta = & \left \{\begin{aligned}
      &\forall i\in [n]
      & w_i^2
      & = w_i \\
      &\forall i\in [n] & w_i \Iprod{x_i, v}^2 & \leq w_i \delta^2  \\
      && \norm{v}_2^2 &=1 \\
    \end{aligned}\right\}
\end{split}
\end{equation}

Intuitively, $w_i$ are indicators indicating the concentrated points. When $w_i = 1$, we must have $\Iprod{x_i, v}^2 \leq \delta^2$ which ensures that $x_i$ is concentrated in the direction $v$, where we search over unit vectors $v$.
The main theorem we establish is that we can certify an upper bound of $O(\delta)$ on $\sum w_i$ over all feasible regions of the program above. Formally,

\begin{theorem}[Certifying Anti-Concentration]
\label{thm:main-anti-concentration-thm}
Given $\delta>0$, and $n \geq n_0$ samples $\Set{x_1, x_2, \ldots , x_n} \subseteq \mathbb{R}^d$ for some $n_0 \ge d^t$, sampled from a $(\delta, \exp\paren{1/\delta^2}, O(1) )$ \textit{reasonably anti-concentrated} distribution $\calD$ (\cref{def:reasonably-anti-concentrated-dist}),  let $t = \Paren{ \log(d) \cdot {\exp\Paren{1/\delta^4}^{ \exp\Paren{ 1/\delta^2} } } }$. Then, there is a degree-$t$ certificate of anti-concentration, i.e.  
\begin{equation*}
    \calA_\delta \sststile{t}{w,v} \Set{  \frac{1}{n} \sum_{i \in [n]} w_i  \leq \delta },
\end{equation*}
\end{theorem}

At a high level, we do not provide a direct sum-of-squares proof of this inequality. Instead, we reason about degree-$t$ pseudo-distributions to infer the existence of a sum-of-squares proof (see \cref{fact:sos-completeness}).
Given a degree-$t$  pseudo-distribution $\mu$ that is consistent with the constraint system $\calA_\delta$ in Equation \eqref{eqn:cons-system}, we perform several re-weightings that correspond to conditions on entries in $v$ such that either the resulting vector is analytically dense (i.e. the $q$-norm is bounded by a small constant times the $2$-norm ) or $\exp(1/\delta)$ moments of  $\ell_2$ norm are small. Then, we use the analytic properties of the re-weighted pseudo-distribution to certify an upper bound on the objective in \eqref{eqn:cons-system}.

More concretely, we abstract out three lemmas that help us complete the proof of Theorem~\ref{thm:main-anti-concentration-thm}. We defer the proofs of these lemmas to the subsequent subsections.

We begin by showing that if the re-weighting from Algorithm~\ref{algo:re-weighting} outputs a pseudo-distribution $\mu'$ and a constant $u_\calS$ such that the pseudo-moments of $\norm{v- u_\calS}_2$ are small, then we can simply decompose the polynomial approximation to the box indicator function and appeal to anti-concentration along the direction $u_\calS$. Formally, we have the following lemma: 

% \ainesh{this lemma says that if the direction obtained by removing the "heavy coordinates" has small enough $\ell_2$ moments, we can directly bound the error in expanding the indicator around $u_\calS$.  }

\begin{restatable}[Decomposing the box-indicator for \emph{almost-sparse} directions]{lemma}{AlmostSparse}
% \begin{lemma}[Decomposing the box-indicator for \emph{almost-sparse} directions]
\label{lem:case1-decomposition}
Given $1> \gamma>0$, let  $t \in \mathbb{N}$, and let   $p$ be the degree-$2t$ polynomial from  \cref{def:box-indicator}. Let $\Set{x_i}_{i \in [n]}$ be $n$ i.i.d. samples from a \textit{reasonably anti-concentrated} distribution $\calD$ (see \cref{def:reasonably-anti-concentrated-dist}). 
Let $\mu'$ be a degree-$\phi$ pseudo-distribution such that $\phi\geq 4t$ and let $u_{\calS}$ be a fixed constant vector supported on $|\calS| \ll d$ coordinates.  If for all $\ell \in [t]$,  $\pexpecf{\mu'}{ \norm{v -  u_{\calS } }_2^{2\ell} } \leq \gamma^{2\ell}$, then 
\begin{equation*}
    \pexpecf{\mu'}{ \frac{1}{n} \sum_{i \in [n] }   p^2\Paren{\Iprod{x_i, v  }^2  } }  \leq   \frac{1}{n}\sum_{i \in [n] }  p^2\Paren{  \Iprod{x_i,  u_\calS }^2 }   + O\Paren{2^{O(t)} c_{2t} \gamma  },
\end{equation*}
where $c_{2t}$ is the hypercontractivity constant of $\calD$.
\end{restatable}
% \end{lemma}

% 
Next, we show that if instead, the re-weighting from Algorithm~\ref{algo:re-weighting} outputs a pseudo-distribution such that the pseudo-expectation of $\norm{ v - u_{\calS} }^{2z}_{2z}$ is small, we can invoke our sum-of-squares certificate of anti-concentration with shifts from \cref{sec:dense-cert-anti-conc} to certify a bound on the box-indicator polynomial.

% \ainesh{this lemma shows that if the shifts are bounded, and the direction $v - u_\calS$ is analytically dense,   we can invoke a sos proof of anti-concentration  }

\begin{restatable}[Anti-Concentration with Shifts]{lemma}{Shifts}
% \begin{lemma}[Anti-Concentration with Shifts]
\label{lem:case2-analytic-density}
Given $\zeta>0$, 
%a polynomial $z(v)$ corresponding to the polynomial identity in \cref{eqn:explicit-polynomial-identity-sos-cert}, 
a pseudo-distribution $\mu$ of degree $t = \Omega\Paren{ \log(d) \cdot \exp\Paren{ 1/\zeta^4 }^{\exp\Paren{1/\zeta^4} }  }  $,  and a vector $u_{\calS}\in \mathbb{R}^d$ supported on $\abs{\calS}$ coordinates for some $\calS \subseteq [d]$, 
% let $\mu'$ be a re-weighting of $\mu$ such that
% $\pexpecf{\mu'}{ \norm{v - u_{\calS}}_{2z}^{2z} }  \leq \lambda^{2z}$ and $\pexpecf{\mu'}{\norm{v - u_\calS }_2^{2k}  } \geq \eta$\ainesh{set $\lambda, \eta, $}+ \pexpecf{\mu'}{z(v) \Paren{ \zeta^2 \norm{v-  u_\calS }_2^4 - \norm{v-  u_\calS }_4^4 } }  .  
let $\Set{x_i}_{i \in [n]}$ be a set of $n$ points such that for all $i \in [n]$, $\Iprod{x_i, u_\calS }^2 \leq \bigO{\log(1/\zeta)}$. 
Let $q$ be the polynomial from Theorem~\ref{thm:anti-concentration-certificates-with-shifts} with degree $\bigO{ \exp\Paren{1/\zeta^2}}$. If $$\pexpecf{\mu}{\norm{ v - u_\calS }_{2z}^{2z} } \leq \zeta^{2z} \pexpecf{\mu}{\norm{ v - u_\calS }_2^{2z} },$$
for $z = \Omega\Paren{ \log(d) \exp\Paren{1/\zeta^2} }$,
then we have  
\begin{equation*}
   \pexpecf{\mu}{ \frac{1}{n} \sum_{i \in [n]} q^2 \Paren{ \Iprod{x_i , v} } } \leq \zeta  \cdot  \pexpecf{\mu}{\norm{v - u_\calS }^2_2 } + \frac{1}{\exp\Paren{1/\zeta^2 }} 
\end{equation*}
\end{restatable}

% \gnote{We stop the re-weighting process when the norm of $v - u_S \le \gam$ where we choose $\gam$ to be exponentially small in $\delta$, so that the $2^{O(t)}$ term in Lemma 5.1 is killed off. In Lemma 5.2, $q^2(\ip{x}{v})$ is being evaluated at this $\norm{v} \approx \gam$, which makes t he degree of $q^2$ to be $exp(exp(\delta))$. It's important that $q$ in Lem 5.3 is of a much higher degree than the $p$ in Lem $5.2$}

Combining \cref{lem:case1-decomposition,lem:case2-analytic-density}, we can now obtain the following re-weighting statement: 

% \ainesh{formal lemma for given any pseudo-distribution, there exists a re-weighting such that $\pexpecf{}{\sum_i w_i  } < \delta n$.  }

\begin{restatable}[Re-weighting certifies Anti-Concentration]{lemma}{Reweighting}
% \begin{lemma}[Re-weighting certifies Anti-Concentration]
\label{lem:re-weighting-to-certificate}
Given $0<\delta<1$ and  a degree $t$ pseudo-distribution $\mu$ over a vector valued indeterminate $v \in \mathbb{R}^d$, there exists a re-weighting of $\mu$, denoted by $\mu'$ such that 
\begin{equation*}
    \pexpecf{\mu'}{\sum_{i \in [n]} w_i  } < \delta n, 
\end{equation*}
if $t = \Omega\Paren{ \log(d) \Paren{ \exp\Paren{ 1/\delta^4 }^{\exp\Paren{1/\delta^4}} } }$. 
\end{restatable}

We defer the proofs of \cref{lem:case1-decomposition,lem:case2-analytic-density,lem:re-weighting-to-certificate} to subsequent sections. 
We are now ready to complete the proof of our main Theorem~\ref{thm:main-anti-concentration-thm} by inferring the existence of a sum-of-squares proof from the existence of the re-weighting in \cref{lem:re-weighting-to-certificate}. We note that such an implication is not true in general, and we heavily exploit the structure of our re-weighting strategy in the following proof.

\begin{proof}[Proof of Theorem~\ref{thm:main-anti-concentration-thm}]

We proceed via contradiction. We begin by assuming there exists some pseudo-distribution $\mu$ that witnesses an objective value that is large. We then re-weight this pseudo-distribution to obtain $\mu'$ such that this re-weighting fixes the objective value (up to small perturbations). Crucially, the objective value is fixed, even when the pseudo-distribution is re-weighted by significantly lower degree polynomials. Then, we use the iterative re-weighting scheme from \cref{lem:re-weighting-to-certificate} on $\mu'$ to obtain a pseudo-distribution $\mu''$ which must have a low objective value. Note, by fiat, the $\mu' \to \mu''$ re-weighting cannot change the objective value too much and this presents a contradiction. 

Assume for the sake of contradiction that there exists some degree $$t = \Omega\Paren{ \log(d) \cdot \Paren{ \exp\Paren{ 1/\delta^4 }^{\exp\Paren{1/\delta^4}} } }$$ pseudo-distribution $\mu$ such that $\pexpecf{\mu}{ \sum_{i \in [n]} w_i } > \delta n$. We now re-weight $\mu$ by $\Paren{ \frac{ \sum_{i \in [n]} w_i}{n} }^q$ for some $q\geq \Omega(t)$ to obtain $\mu'$, which essentially fixes the value of $\sum_{i \in [n] } w_i$. It follows from \cref{lem:scalar-reweighting} that 
\begin{equation}
    \pexpecf{\mu' }{ \Paren{ \frac{ \sum_{i \in [n]} w_i }{n}  - \pexpecf{\mu}{ \frac{ \sum_{i \in [n]} w_i }{n} } }^{2q} } \leq \eps^{2q} \Paren{ \pexpecf{\mu}{ \frac{ \sum_{i \in [n]} w_i }{n}} }^{2q}. 
\end{equation}

for $\eps$ as per that lemma. Next, it follows from \cref{lem:re-weighting-to-certificate} that there exists a re-weighting $\mu''$ of $\mu'$ such that $\pexpecf{\mu''}{ \sum_{i \in [n]} w_i  } \leq \delta n/10$.  
Now, as we will see in \cref{sec:re-weighting-pseudo-dist}, $\mu'$ was constructed by re-weighting $\mu$ iteratively by large powers of coordinates of $v$, and this process repeats $\gamma = \bigO{\exp\Paren{1/\delta^4}^{\exp\Paren{1/\delta^2}} }$ times.  In particular, let $\mu \to \mu'\to  \mu_1 \to \mu_2 \to \ldots \to \mu_\gamma  = \mu''$ be the intermediate re-weightings, where $\mu_\ell$ is obtained by picking a coordinate $i_\ell \in [d]$ such that $\pexpecf{\mu_{\ell-1}}{ v_{i_\ell}^{2t'} } \geq \zeta^{t'} \pexpecf{\mu_{\ell-1} }{ \norm{v}^{2t'} }$ where $\zeta = 1/\exp\Paren{1/\delta^2}$. Let $\calS \subset [d]$ be the subset of indices that are selected in this process and let $z = \prod_{i \in \calS} v_i^{2t'}$. Therefore,
\begin{equation}
\begin{split}
    \pexpecf{\mu'}{z}  = \pexpecf{\mu'}{ \prod_{i \in \calS} v_i^{2t'} } & =  \pexpecf{\mu_{\gamma-1} }{ v_{i_{\gamma} }^{2t'} } \geq \zeta^{2t'}  .
\end{split}
\end{equation}
% \ainesh{seems like i don't have to unroll, i can just appeal to the last reweighting. }
Next, since $\Set{ \norm{v}^2 = 1 } $ is a constraint in our constraint system, we have $\pexpecf{\mu'}{ z^2 } \leq 1$.

Combining these two observations, we note that
\begin{equation}
    \begin{split}
        \Paren{  \pexpecf{\mu'' }{ \sum_{i \in [n]} w_i - \pexpecf{\mu}{\sum_{i \in [n]} w_i } } }^{4q} & \leq \Paren{  \pexpecf{\mu'' }{ \Paren{ \sum_{i \in [n]} w_i - \pexpecf{\mu}{\sum_{i \in [n]} w_i }}^{2q} } }^{2} \\
        & = \frac{ \Paren{  \pexpecf{\mu' }{ \Paren{ \sum_{i \in [n]} w_i - \pexpecf{\mu}{\sum_{i \in [n]} w_i }}^{2q} z }  }^{2} }{ \Paren{ \pexpecf{\mu' }{ z } }^2 } \\
        & \leq \frac{ n^{4q} \Paren{  \pexpecf{\mu' }{ \Paren{ \sum_{i \in [n]} w_i/n - \pexpecf{\mu}{\sum_{i \in [n]} w_i/n }}^{4q}  }  } \Paren{ \pexpecf{\mu'}{ z^2 } } }{ \Paren{ \pexpecf{\mu' }{ z } }^2 } \\
        & \leq \eps^{2q} \cdot n^{4q}/\zeta^{4t'} 
    \end{split}
\end{equation}
Taking the $4q$-th root and setting $\eps = \zeta/(10\delta)$,  we have
\begin{equation}
    \pexpecf{\mu}{ \sum_{i \in [n]} w_i  } \leq \pexpecf{\mu''}{ \sum_{i \in [n]} w_i } + \delta n / 10, 
\end{equation}
which is a contradiction. This allows us to conclude that for any degree-$t$ pseudo-distribution $\pexpecf{\mu}{ \sum_{i \in [n]} w_i } \leq \delta n$, and therefore 

\begin{equation*}
    \calA_\delta \sststile{t}{v,w}\Set{ \delta n - \sum_{i \in [n]} w_i \geq 0 },
\end{equation*}
which concludes the proof.
\end{proof}

\subsection{Handling \emph{almost-sparse} directions}

In this subsection, we prove Lemma~\ref{lem:case1-decomposition}, restated for convenience.

\AlmostSparse*

\begin{proof}

We begin by recalling that by definition $p^2$ is an even polynomial, and therefore we have
\begin{equation}
\begin{split}
   \frac{1}{n} \sum_{i \in [n] } p^2\Paren{ \Iprod{x_i, v} } & = \frac{1}{n} \sum_{i \in [n] }  \sum_{j \in [t]}  c_j \Iprod{ x_i , v }^{2j} \\
   & =    \sum_{j \in [t]}  c_j \Paren{ \frac{1}{n} \sum_{i \in [n]}  \Iprod{ x_i , v }^{2j}  } 
\end{split}
\end{equation}
where for all $j$, $\abs{c_j} \leq 2^{O(t)}$ (Lemma~\ref{lem: cheby_coeff_bound}). For each sample $x_i$, consider the decomposition $\Iprod{x_i, v} = \Iprod{ x_i , v - u_{\calS} } + \Iprod{ x_i ,  u_{\calS} }$. 
Consider the case where $c_j \geq 0$. Now, using the binomial theorem, for $j \in [t]$, we have
\begin{equation*}
    \Iprod{x_i, v}^{2j} = \Iprod{ x_i ,  u_{\calS} }^{2j} + \sum_{\ell = 1}^{2j} \binom{2j}{ \ell } \Iprod{ x_i ,  u_{\calS} }^{\ell} \Iprod{ x_i , v - u_{\calS} }^{2j - \ell}
\end{equation*}
Then, summing over all $i \in [n]$ and taking pseudo-expectation, we have,
\begin{equation}
\label{eqn:binomial-expansion-1}
\begin{split}
    & \pexpecf{\mu' }{  \frac{ 1 }{n} \sum_{i \in [n] }  \Iprod{x_i, v}^{2j} } \\
    & =   \frac{ 1 }{n} \sum_{i \in [n] }  \Iprod{ x_i ,  u_{\calS} }^{2j} +  \underbrace{  \sum_{\ell = 1}^{2j} \binom{2j}{ \ell } \pexpecf{\mu'}{ \frac{1}{n} \sum_{i \in [n] }  \Iprod{ x_i ,  u_{\calS} }^{\ell}  \Iprod{ x_i , v - u_{\calS} }^{2j - \ell} } }_{\eqref{eqn:binomial-expansion-1}.1 }
\end{split}
\end{equation}

Observe, for a fixed $\ell$, the term in \eqref{eqn:binomial-expansion-1}.1 is a scalar quantity (since we take pseudo-expectation over $\mu'$) and therefore, it suffices to bound the absolute value of this scalar. Recall, by pseudo-Jensen's, for any indeterminate $x$, $\Paren{ \pexpecf{\mu}{ x} }^2 \leq \pexpecf{\mu}{ x^2 }$. Therefore, for a fixed $\ell$, we have

\begin{equation}
\label{eqn:pseudo-jensen-split}
    \begin{split}
        & \Paren{ \pexpecf{\mu'}{  \frac{1}{n} \sum_{i \in [n] }  \Iprod{ x_i ,  u_{\calS} }^{\ell}  \Iprod{ x_i , v - u_{\calS} }^{2j - \ell} }  }^2 \\
        & \hspace{0.4in}\leq \pexpecf{\mu'}{ \Paren{   \frac{1}{n} \sum_{i \in [n] }  \Iprod{ x_i ,  u_{\calS} }^{\ell}  \Iprod{ x_i , v - u_{\calS} }^{2j - \ell} }^2 }  \\
        &\hspace{0.4in} \leq \pexpecf{\mu'}{ \frac{1}{n} \sum_{i \in [n] }  \Iprod{ x_i ,  u_{\calS} }^{2\ell}   \cdot \frac{1}{n} \sum_{i \in [n]} \Iprod{x_i, v - u_{\calS} }^{4j-2\ell} }\\
        &\hspace{0.4in} = \underbrace{ \frac{1}{n} \sum_{i \in [n] }  \Iprod{ x_i ,  u_{\calS} }^{2\ell} }_{\eqref{eqn:pseudo-jensen-split}.(1)} \underbrace{ \pexpecf{\mu'}{\frac{1}{n} \sum_{i \in [n]} \Iprod{x_i, v - u_{\calS} }^{4j-2\ell}  } }_{\eqref{eqn:pseudo-jensen-split}.(2)}
    \end{split}
\end{equation}
To bound term \eqref{eqn:pseudo-jensen-split}.(1), we use true hypercontractivity of the samples:
\begin{equation}
\label{eqn:true-hypercontractivity-of-samples}
\begin{split}
    \frac{1}{n} \sum_{i \in [n] }  \Iprod{ x_i ,  u_{\calS} }^{2\ell} & \leq K_{2\ell} \Paren{  \frac{1}{n} \sum_{i \in [n]} \Iprod{ x_i, u_{\calS} }^2   }^{\ell} \\
    & = K_{2\ell} \Paren{  u_{\calS}^\top \Paren{ \frac{1}{n} \sum_{i \in [n]} x_i x_i^\top } u_{\calS}  }^{\ell}\\
    & \leq 2 K_{2\ell} \Norm{ u_{\calS} }^{2\ell}_2 \leq 2 K_{2\ell},
\end{split}
\end{equation}
where $K_{2\ell}$ is the hypercontractivity constant. 

To bound term \eqref{eqn:pseudo-jensen-split}.(2), we appeal to certifiable hypercontractivity of the uniform distribution over the samples: 
\begin{equation}
\label{eqn:certifiable-hypercontractivity-of-samples}
    \begin{split}
    & \pexpecf{\mu'}{\frac{1}{n} \sum_{i \in [n]} \Iprod{x_i, v - u_{\calS} }^{4j-2\ell}  } \\
    & \leq K_{2\ell} \pexpecf{\mu'}{  \Paren{  \frac{1}{n} \sum_{i \in [n]} \Iprod{ x_i,  v- u_{\calS} }^2   }^{2j - \ell } } \\
    & = K_{2\ell} \pexpecf{\mu'}{  \Paren{  \Paren{ v-  u_{\calS} }^\top \Paren{ \frac{1}{n} \sum_{i \in [n]} x_i x_i^\top } \Paren{v- u_{\calS} }  }^{2j -\ell} }\\
    & \leq 2 K_{2\ell} \pexpecf{\mu'}{  \Norm{ v-  u_{\calS} }^{4j -2\ell}_2 } \leq 2 K_{2\ell} \gamma^{4j - 2\ell} .
\end{split}
\end{equation}
Combining equations \eqref{eqn:true-hypercontractivity-of-samples} and \eqref{eqn:certifiable-hypercontractivity-of-samples}, we have
\begin{equation*}
    \abs{ \pexpecf{\mu'}{  \frac{1}{n} \sum_{i \in [n] }  \Iprod{ x_i ,  u_{\calS} }^{\ell}  \Iprod{ x_i , v - u_{\calS} }^{2j - \ell} }  } \leq 2 K_{2\ell} \gamma^{4j - 2\ell}
\end{equation*}
Now we
can bound the term  \eqref{eqn:binomial-expansion-1}.1 as follows: 
\begin{equation}
\begin{split}
     \sum_{\ell\in [2j]} \binom{2j}{\ell } \abs{ \pexpecf{\mu'}{  \frac{1}{n} \sum_{i \in [n] }  \Iprod{ x_i ,  u_{\calS} }^{\ell}  \Iprod{ x_i , v - u_{\calS} }^{2j - \ell} }  } &  \leq \sum_{\ell \in [2j]} \binom{2j}{\ell } 4 K_{2\ell} \gamma^{2j - \ell } \\
     & \leq 4 K_{2j} 2^{2j} \gamma  
\end{split}
\end{equation}
Therefore, for any fixed $j$, we have 
\begin{equation*}
    \pexpecf{\mu' }{  \frac{ c_{2j} }{n} \sum_{i \in [n] }  \Iprod{x_i, v}^{2j} }  =   \frac{ c_{2j} }{n} \sum_{i \in [n] }  \Iprod{ x_i ,  u_{\calS} }^{2j} + c_{2j} \eta_j,
\end{equation*}
where $\abs{ \eta_j } \leq 4 K_{2j} 2^{2j} \gamma$. Summing over all $j \in [t]$,

\begin{equation}
    \begin{split}
        \pexpecf{\mu'}{ \frac{1}{n} \sum_{i \in [n]} p^2_{\gamma} \Paren{ \Iprod{ x_i , v } } } & \leq \frac{1}{n} \sum_{i \in [n]} p^2_{\gamma} \Paren{ \Iprod{x_i,  u_{\calS} } }  + \sum_{j \in [t] } c_{2j} \eta_j \\ 
        & \leq \frac{1}{n} \sum_{i \in [n]} p^2_{\gamma} \Paren{ \Iprod{x_i,  u_{\calS} } }  + 2^{O(t)} K_{t} \gamma 
    \end{split}
\end{equation}

\end{proof}

% \ainesh{we don't even need the next lemma.}

% \begin{lemma}[True Anti-concentration along \emph{sparse} directions]
% \label{lem:case1-decomposition}

% Given $\delta>0, \Delta \geq 1$, let $\Set{x_i}_{i \in [n]}$ be $n= \Omega\Paren{}$\ainesh{insert number of samples} samples from a distribution $\calD$ over $\mathbb{R}^d$ such that \ainesh{insert definition of a reasonable distribution here}.    Then, with probability at least $.99$, for any pseudo-distribution $\mu$, we have
% \begin{equation*}
%     \pexpecf{\mu}{  \frac{1}{n}\sum_{i \in [n] }  p^2\Paren{  \Iprod{x_i,  \pexpecf{\mu}{v_S}} } }  \leq   O(\delta) \Norm{\expecf{\mu}{ v_S} }^2
% \end{equation*}

% \end{lemma}

% \begin{proof}
% First, note, since $p^2$ is a non-negative polynomial, we have 
% \begin{equation*}
%       \frac{1}{n}\sum_{i \in [n] }  p^2\Paren{  \Iprod{x_i,  \pexpecf{\mu}{v_S}} }  \leq  \frac{1}{n}\sum_{i \in [n] }  p^2\Paren{  \Iprod{x_i,  \pexpecf{\mu}{v_S}} } 
% \end{equation*}
% Observe $\expecf{\mu}{ v_S }$ is a fixed vector and no longer an indeterminate. Further, we know that the uniform distribution on the samples is anti-concentration and strictly sub-exponential. Therefore, invoking the right lemma from above bounds $p^2$. The pseudo-expectation doesn't change anything. \ainesh{check the scale at which $p^2$ is certifying anti-concentration.}

% \end{proof}

\subsection{Handling analytically dense directions}

% \ainesh{this section could use some clean up}

In this subsection, we prove Lemma~\ref{lem:case2-analytic-density}, restated below for convenience. 

\Shifts*

We begin by showing that if the pseudo-expectation of the $2z$-norm is bounded, then, the pseudo-expectation of the $4$ norm is bounded, even with additional vectors of the $\ell_2^2$ norm.

\begin{lemma}[$(2z,2)$-Hypercontractivity implies $(4, 2)$-Hypercontractivity]
\label{lem:upgraded-muirhead-sos}
Given a pseudo-distribution $\mu$ over a vector valued indeterminate $v$ such that $\norm{v}^2 \leq 1$, and an even integer $z \geq 2$ such that $\pexpecf{\mu}{ \norm{v}_{2z}^{2z} } \leq \lambda^{2z} \cdot \pexpecf{\mu}{ \norm{v}_2^{2z} }$, we have that for any integer $k \geq 0$, 
\begin{equation*}
    \pexpecf{\mu}{ \norm{v}_2^{2k} \cdot \norm{v}_4^{4} } \leq \lambda^2 \Paren{   \pexpecf{\mu }{  \norm{v}_{2}^{2z }  } .}^{1/z}
\end{equation*}
\end{lemma}
\begin{proof}
We first use induction on $z'$ to prove $ \Set{ \norm{v}_4^{4z'} \le \norm{v}_{2z' + 2}^{2z' + 2} \cdot \norm{v}_2^{2(z' - 1)} } $ for integers $z' \ge 1$. If $z' = 1$, then this is trivial and when $z'=2$, we have
\begin{equation*}
    \sststile{}{v} \Set{ \norm{v}_4^8 = \Paren{ \sum_{i \in [d]} v_i^4  }^2 = \Paren{ \sum_{i \in [d]} v_i^3 \cdot v_i   }^2 \leq \Paren{ \sum_{i \in [d]} v_i^6  } \cdot \Paren{ \sum_{i \in [d]} v_i^2} = \norm{v}_6^6 \cdot \norm{v}_2^2  },
\end{equation*}
where the inequality follows from Cauchy-Schwarz, and this completes the base case. 
For the induction step, assume $\Set{ \norm{v}_4^{4z'} \leq  \norm{v}_{2z' + 2}^{2z' + 2} \cdot \norm{v}_2^{2(z' - 1)}  } $. Then, 
using the inductive hypothesis
\begin{align*}
    \sststile{}{v} \Biggl\{ \norm{v}_4^{4z'+4} &= \norm{v}_4^{4z'} \cdot \norm{v}_4^{4} \le \norm{v}_{2z' + 2}^{2z' + 2} \cdot \norm{v}_2^{2(z' - 1)} \cdot \norm{v}_4^{4} \Biggr\} 
\end{align*}
We wish to prove this is at most $\norm{v}_{2z' + 4}^{2z' + 4} \cdot \norm{v}_2^{2z'}$ in SoS. For this, it suffices to prove
% \gnote{
% What \textit{is} true:
% \[\left(\norm{v}_{(2z' + 6)}^{(2z' + 6)}\right)^2 \le
%     \norm{v}_{2(2z' + 2)}^{2(2z' + 2)} \cdot \norm{v}_8^{8} &\le \norm{v}_{2(2z' + 4)}^{2(2z' + 4)} \cdot \norm{v}_4^{4}\]
%     }
\begin{align*}
    \norm{v}_{2z' + 2}^{2z' + 2} \cdot \norm{v}_2^{2(z' - 1)} \cdot \norm{v}_4^{4} &\le \norm{v}_{2z' + 4}^{2z' + 4} \cdot \norm{v}_2^{2z'}\\
    \qquad
    \norm{v}_{2z' + 2}^{2z' + 2} \cdot \norm{v}_4^{4} &\le \norm{v}_{2z' + 4}^{2z' + 4} \cdot \norm{v}_2^{2}\\
     \qquad
    \sum_{i \neq j} v_i^{2z' + 2} \cdot v_j^{4} &\le \sum_{i \neq j} v_i^{2z' + 4} \cdot v_j^{2}\\
     \qquad
    0 &\le \sum_{i \neq j} (v_i^{2z' + 4} \cdot v_j^{2} - v_i^{2z' + 2} \cdot v_j^{4})\\
     \qquad
    0 &\le \sum_{i < j} (v_i^{2z' + 4} \cdot v_j^{2} + v_j^{2z' + 4} \cdot v_i^{2} - v_i^{2z' + 2} \cdot v_j^{4} - v_j^{2z' + 2} \cdot v_i^{4})
\end{align*}

% \gnote{clarification for ainesh: Set $w_i = v_i^2$, then we are trying to prove $\sum w_i^{z' + 1}w_j^2 \le \sum w_i^{z' + 2}w_j$ which, by  Muirhead ($[z'+2, 1] \succeq [z'+1, 2]$), should be equal to $SoS + SoS * w_i$. We have $\sum w_i^{z'+1}w_j^2 = \frac{n!}{(n-2)!} \cdot \calR(w \to w_1^{z'+1}w_2) = \frac{1}{(n-2)!}[(z'+1, 1, 0, ...0)]$.}
% \ainesh{clean up this part.}\gnote{wdym?}

In the first step, we dropped the power of $\norm{v}_2$, and in the second step, we just expanded the norms. The remaining steps are just rearranging the given expression.
Therefore, it suffices to prove that for each $i < j$, $(v_i^{2z' + 4} \cdot v_j^{2} + v_j^{2z' + 4} \cdot v_i^{2} - v_i^{2z' + 2} \cdot v_j^{4} - v_j^{2z' + 2} \cdot v_i^{4})$ is SoS. But this is true because
\begin{align*}
    v_i^{2z' + 4} &\cdot v_j^{2} + v_j^{2z' + 4} \cdot v_i^{2} - v_i^{2z' + 2} \cdot v_j^{4} - v_j^{2z' + 2} \cdot v_i^{4}\\
    &= (v_i^2 - v_j^2)^2 (v_i^{2z'}v_j^2 + v_i^{2z'-2}v_j^4 +\ldots + v_i^2v_j^{2z'})
\end{align*}
which is clearly a sum of squares.
We verify the last equality in the following display.
\begin{align*}
    (v_i^2 - v_j^2)^2 &(v_i^{2z'}v_j^2 + v_i^{2z'-2}v_j^4 +\ldots + v_i^2v_j^{2z'})\\
    &=(v_i^2 - v_j^2)(v_i^{2z'+2}v_j^2 + \ldots + v_i^4v_j^{2z'} - v_i^{2z'}v_j^4 - \ldots - v_i^2v_j^{2z'+2})\\
    &= (v_i^2 - v_j^2)(v_i^{2z'+2}v_j^2 - v_i^2v_j^{2z'+2})\\
    &= v_i^{2z' + 4} \cdot v_j^{2} + v_j^{2z' + 4} \cdot v_i^{2} - v_i^{2z' + 2} \cdot v_j^{4} - v_j^{2z' + 2} \cdot v_i^{4}
\end{align*}
as desired.

Now, we complete the proof.
Using \cref{fact:pseudo-expectation-holder}, we have
\begin{equation}
    \begin{split}
    \Paren{ \pexpecf{\mu}{ \norm{v}_2^{2k}  \cdot \norm{v}_4^4 } }^z &\leq 
        \Paren{ \pexpecf{\mu}{ \norm{v}_2^{2k} \cdot  \norm{v}_4^4 } }^{z - 1}\\
        & \leq \pexpecf{\mu}{ \norm{v}_2^{2k(z-1)}\cdot  \norm{v}_4^{4z - 4}  }  \\
        & \leq \pexpecf{\mu}{ \norm{v}_{2z}^{2z} \cdot \norm{v}_2^{2 (z-2) + 2k(z-1)} } \\
        & \leq  \pexpecf{\mu}{ \norm{v}_{2z}^{2z}}\\
        & \leq  \lda^{2z}\pexpecf{\mu}{ \norm{v}_{2}^{2z}}\\
    \end{split}
\end{equation}

Taking $z$th roots completes the proof.
\end{proof}

We can now complete the proof of our lemma, building upon Theorem~\ref{thm:anti-concentration-certificates-with-shifts}, which certifies anticoncentration of dense directions with shifts.

\begin{proof}[Proof of Lemma~\ref{lem:case2-analytic-density}]

% \begin{enumerate}
%     \item Invoke lemma for $2z\to 2$ implying $4\to 2$.
%     \item Split $q^2\Paren{ \Iprod{x, v}  }$ into $q^2\Paren{ \Iprod{x, v - u_{\calS} + \Iprod{ x, u_{\calS} } }  }$.
%     \item Invoke sos proof of anti-concentration on $v - u_{\calS}$ with $\Iprod{x_i , u_\calS }$ being the shifts. 

% \end{enumerate}
Since for all $i \in [n]$, we know that $\Iprod{x_i, u_\calS} \leq O(\log(1/\zeta))$, it suffices to consider the box indicator polynomial $q^2$ corresponding to Theorem~\ref{thm:anti-concentration-certificates-with-shifts}, with $\Delta^2  =  \Theta(\log(1/\zeta ))$. Recall, we know that 
$$\pexpecf{\mu}{  \norm{v - u_\calS}_{2z}^{2z} } \leq \eta^{2z }  \pexpecf{\mu}{ \norm{v - u_\calS}_{2}^{2z} },$$ 
where $\eta  = 1/ \exp\Paren{1/\zeta^4}$. 
% \gnote{why is this not just $\zeta$?}
Let $k = \exp(1/\zeta^4)$. It follows from~\cref{lem:upgraded-muirhead-sos} that for any $\ell \in [k]$, 
\begin{equation}
\label{eqn:bounding-fourthmoment}
    \pexpecf{\mu}{ \norm{v - u_\calS}_2^{2\ell} \cdot  \norm{v-u_{\calS}}_4^{4} } \leq  \eta  . 
\end{equation}

However, observe such a statement is not sufficient to invoke the axiom $\Set{\norm{v}_4^4 \leq \lambda \norm{v}_2^4 }$ required to apply \cref{thm:anti-concentration-certificates-with-shifts}. Instead, writing out the guarantee of \cref{thm:anti-concentration-certificates-with-shifts} (from \cref{eqn:explicit-polynomial-identity-sos-cert} with the formal substitution $v = v - u_\calS$) for some $\lambda< 1/\exp\Paren{ 1/\zeta^4 }$, we have   
\begin{equation}
\label{eqn:main-anti-conc-certificate-shifts}
    \zeta \norm{v - u_\calS }^2  -\frac{1}{n} \sum_{i \in [n]} q^2\Paren{\Iprod{x_i, v - u_\calS} + \Iprod{x_i , u_\calS}   } = \textrm{sos}(v - u_\calS) + z(v - u_\calS)\Paren{ \lambda \norm{v - u_\calS}_2^4 - \norm{v - u_\calS}_4^4 }
\end{equation}

Recall, $z(v - u_{\calS}) = \sum_{a \in [\exp\Paren{1/\zeta^2}] } c_a \norm{ v - u_\calS }_2^{2a} $ (which we can see by tracking our proof of \cref{thm:anti-concentration-certificates-with-shifts}),
% \ainesh{add lemma for this},
where $\abs{c_a} \leq \exp\Paren{1/\zeta^2}$, and thus using \cref{eqn:bounding-fourthmoment}, we have  
\begin{equation}
\label{eqn:re-weighted-by-z-bound}
\begin{split}
    \pexpecf{\mu'}{ z(v-u_\calS) \norm{v-u_{\calS}}_4^{4} } & = \pexpecf{\mu'}{ \Paren{ \sum_{a \in [\exp\Paren{1/\zeta^2}] }   c_a \norm{ v - u_\calS }_2^{2a}}  \norm{v-u_{\calS}}_4^{4} } \\
    & \leq \eta \cdot \exp\Paren{2/\zeta^2 } \leq 1/\exp\Paren{1/\zeta^3},
\end{split}
\end{equation}
Similarly, since $\lambda$ is small and $\norm{v - u_{\calS} }^2 \leq 1$, we can bound 
\begin{equation*}
    \left |  \pexpecf{\mu}{z\Paren{v - u_\calS } \lambda \norm{v -u_\calS}_2^4 } \right| \leq  1/\exp\Paren{1/\zeta^3}.  
\end{equation*}
Taking pseudo-expectations on both sides in \cref{eqn:main-anti-conc-certificate-shifts} and plugging in the above bounds we have
\begin{equation*}
    \pexpecf{\mu}{ \frac{1}{n} \sum_{i \in [n]} q^2 \Paren{ \Iprod{x_i , v} } } \leq \zeta  \cdot  \pexpecf{\mu}{\norm{v - u_\calS }^2 } + \frac{1}{\exp\Paren{1/\zeta^2 }},
\end{equation*}
as desired. 
\end{proof}

\subsection{Re-weightings to  Refutation}
% \ainesh{this section should also be stable.}
In this subsection, we prove that given any pseudo-distribution $\mu$ we can find a re-weighting $\mu'$ such that $\pexpecf{\mu'}{ \sum_{i \in [n]} w_i  } \leq \delta n$, formally,

\Reweighting*

We proceed by case analysis, where we case on whether the re-weighting outputs a sparse vector $u_\calS$ and pseudo-distribution $\mu'$ for which pseudo-moment of $\norm{ v - u_\calS }_2^2$ are small, or $v - u_\calS$ is analytically dense, in the sense that $\pexpecf{\mu}{ \norm{ v - u_\calS }_{2z}^{2z} }$ is small. In the first case, we invoke \cref{lem:case1-decomposition} to appeal to anti-concentration of the sparse coordinates,  and in the second case we invoke the analytic sum-of-squares certificates from \cref{lem:case2-analytic-density}. We abstract out reweighting analyses to later sections and utilize \cref{thm:key-re-weighting-theorem} in our proof.

\begin{proof}[Proof of \cref{lem:re-weighting-to-certificate}]
 
We assume the hypercontractivity constant $K_{2\ell} \leq c^{2\ell}$ for a fixed constant $c$, where $\ell \leq \exp\Paren{1/\delta^2}$. 
Let $\mu$ be any degree-$\Omega\Paren{  \log(d) \cdot  \exp\Paren{1/\zeta^4 }^{\exp\Paren{1/\zeta^4} }   } $ pseudo-distribution that is consistent with $\calA_\delta$. Then, we can appeal to \cref{thm:key-re-weighting-theorem} with parameters $\delta = \delta/2 , k = 1/\exp\Paren{1/\delta^2}, z= \log(d)/\exp\Paren{1/\delta^2} , \eta = 1/\exp\Paren{1/\delta^2}$ to obtain a re-weighted pseudo-distribution $\mu'$ and a vector $u_\calS$ supported on at most $\abs{ \calS } = \bigO{1/\Paren{ \delta^{4k+2} \eta^2 } }$ coordinates. We first consider the case where the resulting pseudo-distribution $\mu'$ satisfies that for all $y\in[k]$, \begin{equation*}
    \pexpecf{\mu'}{ \Norm{v - u_\calS }_2^{2y} } \leq \eta. 
\end{equation*}
Next, observe that $p^2(z^2) \geq 1/2$ for all $z \in [-\delta, \delta]$ is a univariate inequality in a bounded interval. Further, it follows from the constraints $\calA_\delta$ that $w_i^2  \Iprod{x_i, v}^2  \leq w_i \delta^2 \leq \delta^2$ and therefore, invoking the univariate inequality with $z = w_i^2 \Iprod{x_i, v}^2$, and using \cref{fact:univariate-sos-proofs} we have 
\begin{equation}
\label{eqn:univariate-fact-main-proof}
    \Set{\calA_\delta} \sststile{}{ w_i^2 \Iprod{x_i, v}^2 } \Set{  \frac{1}{2}\leq p^2\Paren{w_i \Iprod{x_i , v}^2 } },
\end{equation}
and therefore,
\begin{equation}
\label{eqn:expanding-sum-w_i-case1}
    \begin{split}
        \pexpecf{\mu'}{ \frac{1}{n}\sum_{i \in [n]} w_i } 
 & \leq \pexpecf{\mu'}{\frac{2}{n}\sum_{i \in [n]} w_i \cdot p^2\Paren{  w_i \Iprod{x_i, v}^2 }}\\
 & \leq \underbrace{ \frac{1}{n} \sum_{i \in [n]} p^2\Paren{ \Iprod{x_i , u_\calS }^2 } }_{\eqref{eqn:expanding-sum-w_i-case1}.(1)}  + \delta,
\end{split}
\end{equation}
where the last inequality follows from \cref{lem:case1-decomposition} and our choice of $\eta \leq \delta/ \Paren{ 2^{O(\ell)} K_{2\ell}  }$  that makes the additive error small. Finally, we note that \eqref{eqn:expanding-sum-w_i-case1}.(1) is bounded by $\delta$ since \cref{lemma:box-indicator-small-for-samples} is a univariate inequality in $\Iprod{x, u_\calS}^2$ and thus admits a sum-of-squares proof.

% \ainesh{add stand alone lemma that bounds \eqref{eqn:expanding-sum-w_i-case1}.(1), this is real world anti-concentration. maybe this already is written somewhere? find it.}
% \gnote{$u_S$ is fixed, so we invoke true anti-conc of $x$ and hence of $\ip{x}{u_S}$.}

Next, consider the alternative, where $\mu'$ satisfies
\begin{equation*}
    \pexpecf{\mu'}{ \Norm{ v - u_\calS}_{2z}^{2z} } \leq \Paren{ 2\delta}^{2z} \cdot \pexpecf{\mu'}{ \norm{v - u_\calS}_2^{2z} } \textrm{ and } \pexpecf{\mu'}{\Norm{v - u_\calS }  } \geq \eta. 
\end{equation*}
It follows from \cref{fact:sub-exponential-dist-bounded}  that with for all but a $\delta$-fraction of the samples $x_i$, $\Iprod{x_i, u_\calS}^2 \leq \log(1/\delta)$. Conditioned on this event, we can partition the set of points $x_i$ into $\calP \subset [n]$ such that $$ \calP= \Set{ i \in [n] \hspace{0.05in} \vert  \hspace{0.05in}  \Iprod{x_i, u_\calS }^2 \leq \log(1/\delta) }. $$
Then, 
\begin{equation}
\label{eqn:expanding-sum-w-i-case2}
    \pexpecf{\mu'}{ \frac{1}{n} \sum_{i \in [n]} w_i  } \leq 
 \pexpecf{\mu'}{ \frac{1}{n} \sum_{i \in \calP } w_i } + \frac{1}{n}\left|  [n] \setminus \calP  \right| \leq \underbrace{ \pexpecf{\mu'}{ \frac{1}{n} \sum_{i \in \calP } w_i }}_{\eqref{eqn:expanding-sum-w-i-case2}.(1)} + \delta ,
\end{equation}
and thus it suffices to bound term \eqref{eqn:expanding-sum-w-i-case2}.(1). Recall, $q$ is the polynomial from Theorem~\ref{thm:anti-concentration-certificates-with-shifts} with degree $\bigO{ \exp\Paren{1/\zeta^2}}$. 
Using \cref{eqn:univariate-fact-main-proof} again, we have 
\begin{equation}
    \pexpecf{\mu'}{\frac{1}{n} \sum_{i \in \calP } w_i  } \leq \pexpecf{\mu'}{ \frac{1}{n} \sum_{i \in \calP } w_i q^2\Paren{ \Iprod{ x_i, v }^2 } } \leq \bigO{\delta},
\end{equation}
which follows from \cref{lem:case2-analytic-density}, since each point in $\calP$ has a bounded shift. 
\end{proof}

Finally, in \cref{sec: bounding_coeffs}, we will show how to bound the bit complexity of our SoS proofs. The idea is to instantiate the identity for various carefully chosen $w, v$ and then extract bounds on the coefficients.

\subsection{Certifying Anti-Concentration around a shift}

We observe that we can execute the same proof as above to certify anti-concentration w.r.t. the following set of constraints:

\begin{equation}
\label{eqn:cons-system-shift}
\begin{split}
    %& \hspace{0.2in} \max_{w \in \mathbb{R}^n , v\in \mathbb{R}^d  } \hspace{0.2in} \frac{1}{n} \sum_{ i \in [n]} w_i  
    %\\
    \calA_{\delta, \phi} = & \left \{\begin{aligned}
      &\forall i\in [n]
      & w_i^2
      & = w_i \\
      &\forall i\in [n] & w_i \Paren{ \Iprod{x_i, v}^2  - \phi   } & \leq w_i \delta^2  \\
      && \norm{v}_2^2 &=1 \\
    \end{aligned}\right\}
\end{split}
\end{equation}

\begin{corollary}[Certifying Anti-Concentration around a shift]
\label{cor:anti-concentration-thm-shift}
Given $\delta>0$, $\phi \geq 1$, and $n \geq n_0$ samples $\Set{x_1, x_2, \ldots , x_n} \subseteq \mathbb{R}^d$ for some $n_0 \ge d^t$, sampled from a $(\delta, \exp\paren{1/\delta^2}, O(1) )$-\textit{reasonably anti-concentrated} distribution $\calD$ (\cref{def:reasonably-anti-concentrated-dist}),  let $t = \mathcal{O}\Paren{\frac{\log(d) \cdot \phi }{\exp\Paren{\delta^4}^{ \exp\Paren{ 1/\delta^2} } } }$. Then, there is a degree-$t$ certificate of anti-concentration, i.e.  
\begin{equation*}
    \calA_\delta \sststile{t}{w,v} \Set{  \frac{1}{n} \sum_{i \in [n]} w_i  \leq \delta },
\end{equation*}
\end{corollary}
\begin{proof}
The proof is exactly the same as that of \cref{thm:main-anti-concentration-thm}, with the modification that we use the shifted version of $p^2$, with the shift being $\phi$. It follows from \cref{thm:anti-concentration-certificates-with-shifts} that $\expecf{i \in [n]}{p^2\Paren{ \Iprod{x_i ,v}^2 - \phi  } } \leq \delta$ as long as $p^2$ has degree $\bigO{ 1/\delta^2 + \phi}$.
\end{proof}

%% file: sos-anti-conc-dense-directions.tex
\section{Certifiable Anti-Concentration along analytically dense directions}
\label{sec:dense-cert-anti-conc}

% \ainesh{change all statements to reasonaly anti concentrated distribution}

% \ainesh{note to self. make a pass on this section}

We show that for reasonably anticoncentrated distributions, we can obtain a sum-of-squares certificate of anti-concentration along \textit{analytically-dense} directions, not necessarily centered around the origin:

\begin{restatable}[Anti-concentration certificates with bounded shifts]{theorem}{AntiConcBoundedShifts}
% \begin{theorem}[Anti-concentration certificates with bounded shifts]
\label{thm:anti-concentration-certificates-with-shifts}
Let $v$ be a vector such that $0< \norm{v}_2^2 \leq 1$. Let $\Set{x_i}_{i \in [n]}$  be $n$ iid samples from a distribution $\calD$ which is reasonably anticoncentrated. Let $p^2$ be a degree $t = O\Paren{ \log\Paren{ \frac{1}{ \delta} + \frac{\zeta}{\delta} }  \cdot \Paren{ \frac{1}{ 1/\delta^2  } + \Delta^2 } } $ polynomial such that for all $\zeta^2 \leq \Delta^2$ and 
 $\lambda \leq \frac{\delta}{ (2t)^{4t} \Delta^t }$, 
\begin{equation*}
    \Set{ \norm{v}_4^4 \leq \lambda \norm{v}_2^4 } \sststile{}{v} \Set{ \expecf{ }{p^2\Paren{  \Iprod{x_i,v} - \zeta } } \leq O\Paren{\delta} \norm{v}_2^2 }, 
\end{equation*}
Further, let $\calD'$ be a distribution such that $\norm{y}_2 \leq \Delta$ almost surely for all $y\sim \calD'$ and let $\Set{y_i}_{i \in [n]}$ be $n$ iid samples from $\calD'$. Then,
\begin{equation*}
    \Set{ \norm{v}_4^4 \leq \lambda \norm{v}_2^4 } \sststile{}{} \Set{  \frac{1}{n}\sum_{i \in [n]} {p^2\Paren{ \Iprod{x_i,v} - y_i } } \leq O\Paren{\delta} \norm{v}_2^2 }.
\end{equation*}
\end{restatable}
% \ainesh{ this statement is claiming that if the shifts are from a distribution that is  independent and the magnitude of the shift is bounded then, it must be that $p^2$ is still small. }

The rest of the section is devoted to proving this theorem and relies on the following strategy.
\begin{itemize}
    \item In 
Lemmas~\ref{lem:lower-bound-k-th-moment}-\ref{cor:upper-bound-k-th-moment}, we first obtain a SoS certificate of anti-concentration along a dense direction $v$ that satisfies $\Norm{v}_{4}^4 \leq \lambda \Norm{v}_{2}^4$. The analysis proceeds by utilizing an explicit polynomial $p^2$ that approximates the box-indicator, and is inspired by \cite{kindler2010ugc}. 
    \item In \cref{thm:certificates-with-shifts}, we show that given a sufficiently $\ell_4$-dense direction, we can obtain a certificate of anti-concentration around a point that need not be the origin.
\end{itemize}

We can re-write the final conclusion above as the following polynomial identity:
\begin{equation}
    \label{eqn:explicit-polynomial-identity-sos-cert}
    \delta \norm{v}^2  -\frac{1}{n} \sum_{i \in [n]} p\Paren{\Iprod{x_i, v}^2-y_i  } = \textrm{sos}(v) + z(v)\Paren{ \lambda \norm{v}_2^4 - \norm{v}_4^4 }
\end{equation}
where $\textrm{sos}(v), z(v)$ are sum-of-squares polynomials, $p^2$ has degree $t = O\Paren{ \log\Paren{ \frac{1}{ \delta} + \frac{\zeta}{\delta} }  \cdot \Paren{ \frac{1}{ 1/\delta^2  } + \Delta^2 } } $ and  $\lambda\leq \frac{\delta}{ (2t)^{4t} \Delta^t }$.
Further, $z(v)$ is explicit in our proof of \cref{thm:anti-concentration-certificates-with-shifts}.

\subsection{Anti-concentration certificates along dense directions}

The goal of this section is to obtain a SoS certificate of anti-concentration along the direction $v$ that satisfies $\Norm{v}_{4}^4 \leq \lambda \Norm{v}_{2}^4$.
Our approach is to get a lower and upper bound on all 'monomials', i.e. $\expecf{}{\Iprod{x,v}^{2k}}$ and then compute an explicit polynomial $p^2$ that approximates the box-indicator. 

% In this section, we extend anti-concentration certificates beyond the setting where each coordinate is sampled independently. 
% \ainesh{Add formal statement for certifiable anti concentration and proof using bounds on the coefficients of the Chebyshev polynomial. Define 'nice' distributions as almost $k$-wise independent. Show in next sub section, that uniform distributions over generalized convex bodies are 'nice'. }

\begin{restatable}[Lower Bound on $2k$-th Moment]{lemma}{lowerBoundMoment}
\label{lem:lower-bound-k-th-moment}
Let $d/2 \ge k\geq2$ and let $x \sim \calD$, where  $\calD$ is 
a reasonably anticoncentrated distribution.
%a uniform distribution over an unconditional convex body $\calK$. 
Further, for any set $S \subset [d]$ such that $|S|=k$, denote $\expecf{}{\prod_{i \in S } x_i^2 }= \Phi_{k}$. Let $v \in \R^d$ be such that $\Norm{v}_4^4 \leq \lambda \Norm{v}_2^4$ for some $\lambda \in (0, 2/k^2)$. Then, 
\begin{equation*}
    \Set{ \Norm{v}_4^4\leq \lambda \Norm{v}_2^4 } \sststile{4k}{v} \Set{ \expecf{}{\Iprod{x,v}^{2k}} \geq \Paren{\frac{ (2k)! }{2^{k} k!}}\cdot \Phi_k  \Paren{1 - \frac{(k-1) \lambda}{2}} \Norm{v}_2^{2k} }.
\end{equation*}
\end{restatable}

Next, we consider the upper bound certificate. We show that such a certificate holds whenever the distribution we consider is 
reasonably anticoncentrated.
%a uniform distribution over a Generalized Orlicz ball.

\begin{restatable}[Upper Bound on $2k$-th Moment]{lemma}{upperBoundMoment}
\label{cor:upper-bound-k-th-moment}
Let $k\geq2$ and let $x \sim \calD$ where $\calD$ is a
reasonably anticoncentrated distribution.
%uniform distribution over a Generalized Orlicz ball in $\R^d$ with mean $0$ \gnote{This lemma also needs $(c_k, k)$ hypercontractivity, and variance being bounded by $1$}. 
Let $v \in \R^d$ be such that $\Norm{v}_4^4 \leq \lambda \Norm{v}_2^4$ for some $\lambda \in (0, 1/(k^k c_k^k) )$. Then, 
\begin{equation*}
    \Set{ \Norm{v}_4^4\leq \lambda \Norm{v}_2^4 } \sststile{4k}{v} \Set{ \expecf{}{\Iprod{x,v}^{2k}} \leq \Paren{\frac{ (2k)! }{2^{k} k!}} \Paren{1 +  \Paren{ \lambda c_k^{k-2}k^k }} \Norm{v}_2^{2k} }.
\end{equation*}
\end{restatable}

% \gnote{doesn't this need var at most 1 for every coord?}
% \gnote{Why is $\Phi$ bold?}

We prove these lemmas in the subsequent sections.

\subsubsection{Warmup: The case of independent coordinates}

In this section, as a warmup, we consider the case when the $x_i$s are independent with mean $0$ and variance $1$ and prove the bounds from the previous section. 
To this end, we obtain lower and upper bounds on the $2k$th moments in \cref{lem: lower_bound_moment_independent} and \cref{lem: upper_bound_moment_independent} respectively.

We first obtain a lower bound on $\expecf{}{\Iprod{x,v}^{2k}}$ that is captured by sos. The proof is inspired by Lemma 2.6 in Kindler, Naor and Schechtman \cite{kindler2010ugc}.

\begin{lemma}[Lower Bound on $2k$-th moment]\label{lem: lower_bound_moment_independent}
Let $d/2\geq k\geq2$ and let $\{ x_i \}_{i \in [d]}$ be $d$ iid random variables with mean $0$ and variance $1$. Let $v \in \R^d$ be such that $\Norm{v}_4^4 \leq \lambda \Norm{v}_2^4$ for some $\lambda \in (0, 2/k^2)$. Then, 
% \begin{equation*}
    % \Set{ \Norm{v}_4^4\leq \lambda \Norm{v}_2^4 } \sststile{4k}{v} \Set{ \expecf{}{\Iprod{x,v}^{2k}} \geq \Paren{\frac{ (2k)! }{2^{k} k!}} \Paren{1 - \frac{(k-1) \lambda}{2}} \Norm{v}_2^{2k} \prod_{j \in [d]}\expecf{}{x_j^2}^k }
% \end{equation*}\ainesh{what is m??}
% \gnote{Here is the fixed result.}
\begin{equation*}
    \Set{ \Norm{v}_4^4\leq \lambda \Norm{v}_2^4 } \sststile{4k}{v} \Set{ \expecf{}{\Iprod{x,v}^{2k}} \geq \Paren{\frac{ (2k)! }{2^{k} k!}} \Paren{1 - \frac{k(k-1) \lambda}{2}} \Norm{v}_2^{2k} }
\end{equation*}
\end{lemma}

\begin{proof}
For ease of notation in this proof, denote by $w_i := x_iv_i$ for $i \le d$. Using the substitution rule, we have 
\begin{equation}
\label{eqn:lower_bound_main1}
    \begin{split}
    \sststile{4k}{v}\Biggl\{ \expecf{}{\Iprod{x,v}^{2k}} & =  \sum_{ \substack{\ell_1, \ldots \ell_d \in \Z_{\geq 0} \\\ell_1+ \ldots +\ell_d=k }} \frac{(2k)!}{ \prod_{j \in [d]} (2\ell_j)! } \expecf{}{\prod_{i \in [d]}  w_j^{2\ell_j}} \geq \Paren{\frac{(2k)!}{2^k}} \Paren{ \sum_{\substack{ S\subset [n]\\ |S|=k}} \prod_{j \in S} \expecf{}{ w_j^{2}} }  \Biggr\},
    \end{split}
\end{equation}
where the inequality follows from only summing over terms for which $\ell_j \le 1$ for all $j \in [d]$.
On the other hand, using the substitution rule and $\EE[w_j^2] = \EE[x_j^2v_j^2] = v_j^2$, we have 
\begin{equation}
\label{eqn:upper_bound_norm_v}
\begin{split}
    \sststile{4k}{v} \Biggl\{\Norm{v}_2^{2k} &= \Paren{ \sum_{j\in [d]} \expecf{}{ w_j^2 }  }^{k} =  \sum_{\substack{ \ell_1, \ldots, \ell_d \in \Z_{\geq0}\\ \ell_1 + \ldots+\ell_d=k }}\frac{k!}{\prod_{j \in [d]} \ell_j !} \prod_{j \in [d]} (\expecf{}{w_j^2})^{\ell_j}   \Biggr\} .
\end{split}
\end{equation}
We split the term above into terms where all the $\ell_j$'s are at most $1$ and the terms where at least one $\ell_j$ is at least $2$, to get
\begin{equation}
    \begin{split}
    \Norm{v}_2^{2k} &\leq  k! \sum_{\substack{S \subset [n]\\ |S|=k}} \prod_{j \in S} \expecf{}{w_j^2 }  + \frac{k}{2} \sum_{j \in [d]} \expecf{}{w_j^2} \sum_{\substack{r_1, \ldots r_d \in \Z_{\geq0}\\ r_1+\ldots+r_d =k-1 \\r_j\geq1 }} \frac{(k-1)!}{\prod_{i \in[d]} r_i! } \prod_{i \in [d]} (\expecf{}{w_i^2})^{r_i} \\
    & = k! \sum_{\substack{S \subset [n]\\ |S|=k}} \prod_{j \in S} \expecf{}{w_j^2 }  + \frac{k}{2} \sum_{j \in [d]} (\expecf{}{w_j^2})^2  \sum_{\substack{r_1, \ldots r_d \in \Z_{\geq0}\\ r_1+\ldots+r_d =k-2 }} \frac{(k-1)!}{\prod_{i \in[d]} r_i! } \prod_{i \in [d]} (\expecf{}{w_i^2})^{r_i}  \Biggr\} .
\end{split}
\end{equation}
Let us focus on upper bounding the second term above. Observe, the inner summation is independent of $j$ and thus

\begin{equation}
    \begin{split}
    \sststile{4k}{v} \Biggl\{ &  \sum_{j \in [d]} (\expecf{}{w_j^2})^2  \sum_{\substack{r_1, \ldots r_d \in \Z_{\geq0}\\ r_1+\ldots+r_d =k- 2  }} \frac{(k-1)!}{\prod_{i \in[d]} r_i! } \prod_{i \in [d]} (\expecf{}{w_i^2})^{r_i}   \\
    & = \Paren{ \sum_{j \in [d]}  (\expecf{}{w_j^2})^2 } \cdot \Paren{k-1} \sum_{\substack{r_1, \ldots r_d \in \Z_{\geq0}\\ r_1+\ldots+r_d =k-2  }} \frac{(k-2)!}{\prod_{i \in[d]} r_i! } \prod_{i \in [d]} (\expecf{}{w_i^2})^{r_i} \Biggr\} .
\end{split}
\end{equation}
We now bound each of these terms above as follows: using the bound on the $\ell_4^4$ norm, along with $\expecf{}{x_j^2}=1$, 
\begin{equation}
    \Set{\Norm{v}_4^4 \leq \lambda \Norm{v}_2^4 } \sststile{4}{v} \Set{ \sum_{j \in [d]} ( \expecf{}{w_j^2} )^2  = \sum_{j \in [d]}  v_j^4 \leq \lambda \Norm{v}_2^4 },
\end{equation}
\begin{equation}
    \begin{split}
    \sststile{4k}{v} \Biggl\{   \sum_{\substack{r_1, \ldots r_d \in \Z_{\geq0}\\ r_1+\ldots+r_d =k-2}} \frac{(k-2)!}{\prod_{i \in[d]} r_i! } \prod_{i \in [d]} \expecf{}{w_i^2}^{r_i} 
    &  = \Norm{v}^{2k-4}_2\Biggr\} .
\end{split}
\end{equation}
Combining the above, we conclude
\begin{equation}
     \begin{split}
     \sststile{4k}{v} \Biggl\{\Norm{v}_2^{2k} \leq  k! \sum_{\substack{S \subset [n]\\ |S|=k}} \prod_{j \in S} \expecf{}{w_j^2 }  + \frac{k(k-1)}{2} \Paren{ \lambda \Norm{v}_2^{2k}} \Biggr\}
    \end{split}
\end{equation}
Together with \eqref{eqn:lower_bound_main1}, we obtain the result.
\end{proof}

A similar proof technique can be used for the upper bound.
% Using Goutham's suggestion from the last meeting, we can obtain the following upper bound:

\begin{lemma}[Upper Bound on $2k$-th Moment]\label{lem: upper_bound_moment_independent}
\label{lem:upper-bound-independent}
Let $k\geq2$ and let $\{ x_i \}_{i \in [d]}$ be $d$ iid random variables drawn from a distribution $\calD$ with mean $0$ and variance $1$ that is $(c_k, k)$-certifiably hypercontractive. Let $v \in \R^d$ be such that $\Norm{v}_4^4 \leq \lambda \Norm{v}_2^4$ for some $\lambda \in (0, 1/(k^k c_k^k) )$. Then, 
% \begin{equation*}
    % \Set{ \Norm{v}_4^4\leq \lambda \Norm{v}_2^4 } \sststile{4k}{v} \Set{ \expecf{}{\Iprod{x,v}^{2k}} \leq \Paren{\frac{ (2k)! }{2^{k} k!}} \Paren{1 +  \Paren{ \lambda c_k^{k-2}k^k }} \Norm{v}_2^{2k} }.
% \end{equation*}
% \textcolor{red}{
\begin{equation*}
    \Set{ \Norm{v}_4^4\leq \lambda \Norm{v}_2^4 } \sststile{4k}{v} \Set{ \expecf{}{\Iprod{x,v}^{2k}} \leq \Paren{\frac{ (2k)! }{2^{k} k!} +   \lambda c_k^kk^k } \Norm{v}_2^{2k} }.
\end{equation*}
% }
\end{lemma}

\begin{proof}
    The proof is similar to the proof of \cref{lem: lower_bound_moment_independent}. We have
    \begin{equation}
\label{eqn:upper_bound_main}
    \begin{split}
    \sststile{4k}{v}\Biggl\{ \expecf{}{\Iprod{x,v}^{2k}} & =  \sum_{ \substack{\ell_1, \ldots \ell_d \in \Z_{\geq 0} \\\ell_1+ \ldots +\ell_d=k }} \frac{(2k)!}{ \prod_{j \in [d]} (2\ell_j)! } \prod_{i \in [d]} \expecf{}{ x_j^{2\ell_j}v_j^{2\ell_j}}\\
    & \leq \Paren{\frac{(2k)!}{2^k}}  \sum_{\substack{ S\subset [d]\\ |S|=k}} \prod_{j \in S} \expecf{}{ x_j^{2}v_j^2} + \sum_{j \in [d]}v_j^{4}\sum_{\substack{r_1, \ldots r_d \in \Z_{\geq0}\\ r_1+\ldots+r_d =k - 2}} \frac{(2k)!}{\prod_{i \in[d]} (2r_i)! }  \expecf{}{x_j^4\prod_{i \in [d]}x_i^{2r_i}v_i^{2r_i}} \Biggr\},
    \end{split}
\end{equation}

To bound the first term, we use
$\expecf{}{x_j^2 v_j^2} = \expecf{}{x_j^2} v_j^2= v_j^2$ to get

\begin{equation}
    \sststile{4k}{v} \Set{ \sum_{\substack{ S\subset [n]\\ \abs{S}=k}} \prod_{j \in S} \expecf{}{ x_j^{2}v_j^2}  \leq \norm{v}_2^{2k} },
\end{equation}
and therefore the first term is at most $\frac{ (2k)! }{2^{k} k!} \Norm{v}_2^{2k}$.
As for the second term, we will use the following argument. By  hypercontractivity of the true samples (see Definition~\ref{def:certifiable-hypercontractivity-of-linear}), we have
\begin{equation}
    \sststile{4r_i}{v}\Set{ \expecf{}{x_i^{2r}v_i^{2r}} = \expecf{}{x_i ^{2r}} v_i^{2r}\leq c_{k}^{r}\cdot  \expecf{}{x_i^2}^{r} v_i^{2r} \le c_{k}^{r}v_i^{2r}}
\end{equation}
and thus 
\begin{equation}
\begin{split}
    \sststile{4k}{v} \Biggl\{ \sum_{\substack{r_1, \ldots r_d \in \Z_{\geq0}\\ r_1+\ldots+r_d =k-2 }} \frac{(2k)!}{\prod_{i \in[d]} (2r_i)! } \prod_{i \in [d]} \expecf{}{x_j^4(x_iv_i)^{2r_i}} &  \leq \sum_{\substack{r_1, \ldots r_d \in \Z_{\geq0}\\ r_1+\ldots+r_d =k-2 }} \frac{(2k)! \cdot  c_k^k }{\prod_{i \in[d]} (2r_i)! } \prod_{i \in [d]} v_i^{2r_i} \\
    & \le \frac{c_k^k \prod_{i\in [k+1] }(2k-i) }{2^{k-2}}   \sum_{\substack{r_1, \ldots r_d \in \Z_{\geq0}\\ r_1+\ldots+r_d =k-2 }} \frac{(k-2)!  }{\prod_{i \in[d]} r_i! } \prod_{i \in [d]} v_i^{2r_i}\\
    & \leq \Paren{c_k^k k^k } \norm{v}_2^{2k-4} \Biggr\}.
\end{split}
\end{equation}
Finally, we also use $\sum_{j \in [d]} v_j^4 \le \lda \norm{v}_2^4$.
Combining the equations above,
\begin{equation}
\label{eqn:upper_bound_main-2}
    \begin{split}
    \sststile{4k}{v}\Biggl\{ \expecf{}{\Iprod{x,v}^{2k}} 
    & \leq \Paren{\frac{(2k)!}{2^k}}  \norm{v}_2^{2k} +  \Paren{\lambda c_k^k k^k } \norm{v}_2^{2k} \Biggr\},
    \end{split}
\end{equation}
which concludes the proof.
\end{proof}

\subsubsection{Proofs of \cref{lem:lower-bound-k-th-moment} and \cref{cor:upper-bound-k-th-moment}}

Using the same techniques from the previous section, we can complete the proof of \cref{lem:lower-bound-k-th-moment} and \cref{cor:upper-bound-k-th-moment}.

\lowerBoundMoment*

\begin{proof}[Proof of \cref{lem:lower-bound-k-th-moment}]
Using the multinomial theorem and substitution rule, we have 
\begin{equation}
\label{eqn:lower_bound_main-log-conc}
    \begin{split}
    \sststile{4k}{v}\Biggl\{ \expecf{}{\Iprod{x,v}^{2k}} & =  \sum_{ \substack{r_1, \ldots r_d \in \Z_{\geq 0} \\\ r_1+ \ldots +r_d=2k }} \frac{ (2k)!}{ \prod_{j \in [d]} r_j! } \expecf{}{\prod_{i \in [d]}  (v_j x_j)^{r_j}}    \\
    & =  \sum_{ \substack{r_1, \ldots r_d \in \Z_{\geq 0} \\\ r_1+ \ldots +r_d=2k }} \frac{ (2k)!}{ \prod_{j \in [d]} r_j! }  \Paren{ \prod_{i \in [d]} v_j^{r_j}} \cdot \Paren{ \expecf{}{ \prod_{j \in [d]}   x_j^{r_j}} }   \Biggr\},
    \end{split}
\end{equation}

Now, consider the monomials that appear with at least $1$ odd term above. We can represent them as a product of squares as and a product of literals as follows: $\prod_{i\in \calT} x_{i} \prod_{j\in[d]} x_{j}^{2\ell_j} $, where $T \subseteq [d]$ such that $|T|\leq k$, and $\ell_j \in \Z_{\geq 0}$. Since $\calK$ is unconditional, we know the expectation of this monomial under any signing of the coordinates remains the same. Let $\zeta \in \Set{-1,1}^d$ be such that for some $i^* \in \calT$, $\zeta_{i^*} = - \textrm{sign}(x_{i^*})$ and for all $i \neq i^* \in \calT$,   $\zeta_{i^*} =  \textrm{sign}(x_{i^*})$. Then, 
\begin{equation*}
\begin{split}
     \expecf{x \sim \calD }{ \prod_{j \in [d]}   x_j^{r_j}}   =\expecf{}{ \prod_{i\in \calT} x_{i} \cdot  \prod_{j\in[d]} x_{j}^{2\ell_j} } 
     & = \expecf{ x \sim \calD }{ \prod_{i\in \calT} \zeta_i x_{i} \cdot  \prod_{j\in[d]} \Paren{ \zeta_j x_{j} }^{2\ell_j} } \\
     & = \expecf{ x \sim \calD }{ -x_{i^*} \prod_{i \neq i^* \in \calT}  x_{i} \cdot  \prod_{j\in[d]} \Paren{  x_{j} }^{2\ell_j} } \\
     & = - \expecf{}{ \prod_{i\in \calT} x_{i} \cdot  \prod_{j\in[d]} x_{j}^{2\ell_j} }
\end{split}
\end{equation*}
Therefore, we can conclude that all terms appearing with at least one odd power must be $0$. Now, 
\begin{equation}
\label{eqn:lower_bound_main3}
    \begin{split}
    \sststile{4k}{v}\Biggl\{ \expecf{}{\Iprod{x,v}^{2k}} &     =  \sum_{ \substack{\ell_1, \ldots \ell_d \in \Z_{\geq 0} \\\ell_1+ \ldots +\ell_d=k }} \frac{(2k)!}{ \prod_{j \in [d]} (2\ell_j)! } \Paren{ \prod_{i \in [d]} v_j^{2\ell_j}} \cdot \Paren{ \expecf{}{ \prod_{j \in [d]}   x_j^{2\ell_j}} }   \\
    & \geq \Paren{\frac{(2k)!}{2^k}} \Paren{ \sum_{\substack{ S\subset [n]\\ |S|=k}} \Paren{\prod_{j \in S} v_j^2  } \Paren{  \expecf{}{ \prod_{j \in S}  x_j^{2}}} } \Biggr\},
    \end{split}
\end{equation}
% where the inequality follows again from only summing over terms for which $\ell_j = 1$ for all $j \in [d]$.
Again, using the substitution rule, we have 
\begin{equation}
\label{eqn:upper_bound_norm_v-2}
\begin{split}
    \sststile{4k}{v} \Biggl\{\Norm{v}_2^{2k} &= \Paren{ \sum_{i\in [d]}   v^2_i   }^{k}=  \sum_{\substack{ \ell_1, \ldots, \ell_d \in \Z_{\geq0}\\  \ell_1 + \ldots+\ell_d=k }}\frac{k!}{\prod_{j \in [d]} \ell_j !} \prod_{j \in [d]}  v_j^{2\ell_j}    \Biggr\} .
\end{split}
\end{equation}
We again split the term above into terms where all the $\ell_j$'s are at most $1$ and the terms where at one $\ell_j$ is at least $2$ to get,

\begin{equation}
    \begin{split}
    \sststile{4k}{v} \Biggl\{  \sum_{\substack{ \ell_1, \ldots, \ell_d \in \Z_{\geq0}\\ \ell_1 + \ldots+\ell_d=k }}\frac{k!}{\prod_{j \in [d]} \ell_j !} \prod_{j \in [d]} v_j^{2\ell_j} &\leq \Paren{ k! \sum_{\substack{S \subset [n]\\ |S|=k}} \prod_{j \in S} v_j^2  } + \frac{k}{2} \sum_{j \in [d]}v_j^2  \sum_{\substack{r_1, \ldots r_d \in \Z_{\geq0}\\ r_1+\ldots+r_d =k-1 \\r_j\geq1 }} \frac{(k-1)!}{\prod_{i \in[d]} r_i! } \prod_{i \in [d]} v_i^{2r_i}  \\
    &= \Paren{ k! \sum_{\substack{S \subset [n]\\ |S|=k}} \prod_{j \in S} v_j^2  }   + \frac{k}{2} \sum_{j \in [d]} v_j^4   \sum_{\substack{r_1, \ldots r_d \in \Z_{\geq0}\\ r_1+\ldots+r_d =k-2 }} \frac{(k-1)!}{\prod_{i \in[d]} r_i! } \prod_{i \in [d]} v_i^{2r_i}  \Biggr\} .
\end{split}
\end{equation}
In the second term, the inner summation is independent of $j$, so
\begin{equation}
    \begin{split}
    \sststile{4k}{v} \Biggl\{ &  \sum_{j \in [d]} v_j^4  \sum_{\substack{r_1, \ldots r_d \in \Z_{\geq0}\\ r_1+\ldots+r_d =k- 2  }} \frac{(k-1)!}{\prod_{i \in[d]} r_i! } \prod_{i \in [d]} v_i^{2r_i}    
     = \Paren{   \Norm{v}_4^4  } \cdot \Paren{k-1} \sum_{\substack{r_1, \ldots r_d \in \Z_{\geq0}\\ r_1+\ldots+r_d =k-2  }} \frac{(k-2)!}{\prod_{i \in[d]} r_i! } \prod_{i \in [d]} v_i^{2r_i} \Biggr\}
\end{split}
\end{equation}
% Recall $\sststile{}{}\Set{\Norm{v}_4^4 \leq \lambda \Norm{v}_2^4 }$. 
To bound the second term above, we use
\begin{equation}
    \begin{split}
    \sststile{4k}{v} \Biggl\{   \sum_{\substack{r_1, \ldots r_d \in \Z_{\geq0}\\ r_1+\ldots+r_d =k-2}} \frac{(k-2)!}{\prod_{i \in[d]} r_i! } \prod_{i \in [d]} v_i^{2r_i}  
    &  = \Norm{v}^{2k-4}_2\Biggr\} .
\end{split}
\end{equation}
to conclude
\begin{equation}
     \begin{split}
     \Set{\Norm{v}_4^4 \leq \lambda \Norm{v}_2^4 } \sststile{4k}{v} \Biggl\{ &  \sum_{\substack{ \ell_1, \ldots, \ell_d \in \Z_{\geq0}\\ \ell_1 + \ldots+\ell_d=k }}\frac{k!}{\prod_{j \in [d]} \ell_j !} \prod_{j \in [d]} v_j^{2\ell_j} \leq  k! \sum_{\substack{S \subset [n]\\ |S|=k}} \prod_{j \in S} v_j^2 + \frac{k(k-1)}{2} \Paren{ \lambda \Norm{v}_2^{2k}} \Biggr\} .
    \end{split}
\end{equation}
Therefore,
\begin{equation*}
  \Set{\Norm{v}_4^4 \leq \lambda \Norm{v}_2^4  } \sststile{}{}\Set{     \sum_{\substack{S \subset [n] \\ \abs{S}=k} } \prod_{j \in S} v_j^2   \geq \frac{1}{k!} \Paren{ 1- \frac{k(k-1)\lambda}{2}} \Norm{v}_2^{2k} }.
\end{equation*}

% \ainesh{It suffices to now bound  $\expecf{}{\prod_{j \in S} x^{2}_j}$ for any set $S \subset [d]$ of size $k$. Note, a real world proof for log-concave distributions suffices. }
Using the notation $\Phi_k = \expecf{}{\prod_{j \in S} x^{2}_j}$ for any set $S \subset [d]$ of size $k$ and rearranging, 
\begin{equation}
     \begin{split}
     \Set{\Norm{v}_4^4 \leq \lambda \Norm{v}_2^4 } \sststile{4k}{v} \Set{ \expecf{}{\Iprod{x,v}^{2k} } \geq  \Paren{\frac{(2k)!}{2^k k! }} \cdot \Phi_k \cdot  \Paren{ 1- \frac{k(k-1)\lambda}{2}} \Norm{v}_2^{2k}   } .
    \end{split}
\end{equation}
% \ainesh{For this bound to be useful, it better be the case that $\Phi_k = 1 - o(1) $.}
\end{proof}

% \ainesh{if the above works out, the upper bound can be obtained from observing that for $x$ sampled from a uniform distribution over an unconditional convex body, $\{ |x_i| \}_{i \in [d]}$ are negatively associated. }

\upperBoundMoment*

\begin{proof}[Proof of \cref{cor:upper-bound-k-th-moment}]
Let ${i_1, i_2, \ldots, i_t} \subset [d]$ and let $\ell_1,\ell_2, \ldots, \ell_t \in \Z_{\geq 1}$ such that $\sum_{j \in [t]}\ell_j = k$. 
% It follows from Theorem \ref{thm:negative-association} that the set $\Set{ \abs{x_{i_j}} }_{j \in [t]}$ is negatively associated. 
Setting $f(x_{i_1}) = x_{i_1}^{2\ell_1}$ and $g\Paren{x_{i_j}} = x_{i_{j}}^{2\ell_j} $ for $j \in [2,t]$, and observing that $x^2 = |x|^2$,
% it follows from Definition \ref{def:negative-association} that
\begin{equation*}
    \expecf{}{ \prod_{j \in [t]} x_{i_j}^{2\ell_j} } \leq \expecf{}{ x_{i_1}^{2\ell_1}} \cdot  \expecf{}{ \prod_{j \in [2, t]} x_{i_j}^{2\ell_j} }.
\end{equation*}
Repeating the argument above $t$ times, we can conclude  $\expecf{}{  \prod_{j \in [t]} x_{i_j}^{2\ell_j} }  \leq \prod_{j \in [t]} \expecf{}{ x_{i_j}^{2\ell_j}} $.

Then, 
\begin{equation}
    \begin{split}
    \sststile{4k}{v}\Biggl\{ \expecf{}{\Iprod{x,v}^{2k}} & =  \sum_{ \substack{\ell_1, \ldots \ell_d \in \Z_{\geq 0} \\\ell_1+ \ldots +\ell_d=k }} \frac{(2k)!}{ \prod_{j \in [d]} (2\ell_j)! } \expecf{} {\prod_{i \in [d]}  (v_j x_j)^{2\ell_j}} \leq \sum_{ \substack{\ell_1, \ldots \ell_d \in \Z_{\geq 0} \\\ell_1+ \ldots +\ell_d=k }} \frac{(2k)!}{ \prod_{j \in [d]} (2\ell_j)! } \prod_{i \in [d]} \expecf{}{ (v_j x_j)^{2\ell_j}} \Biggr\}
    \end{split}
\end{equation}
The rest of the argument is identical to Lemma \ref{lem:upper-bound-independent}.
\end{proof}

\subsection{Anti-Concentration Certificates with shifts.}

Next, we show that given a sufficiently $\ell_4$-dense direction, we can obtain a certificate of anti-concentration around a point that need not be the origin. Here, we use an explicit construction for the box indicator polynomial, as defined in \cref{subsec:boxpolynomial}.

\begin{theorem}[Certificates with Shifts]
\label{thm:certificates-with-shifts}
Let $v$ be a vector -valued indeterminate.  Given a distribution $\calD$ which is reasonably anticoncentrated and given a shift $\zeta >  0$,  if $\lambda \leq \frac{\delta}{ (2d)^{4d} \zeta^d }$, 
\begin{equation*}
    \Set{ \norm{v}_4^4 \leq \lambda \norm{v}_2^4 } \sststile{}{} \Set{ \expecf{ }{p^2\Paren{ \Iprod{x,v} - \zeta } } \leq O\Paren{\delta} \norm{v}_2^2 }, 
\end{equation*}
where $p^2\Paren{z}$ is a degree $t = O\Paren{ \log\Paren{ \frac{1}{ \delta} + \frac{\zeta}{\delta} }  \cdot \Paren{ \frac{1}{ 1/\delta^2  } + \zeta^2 } } $ polynomial.
% \ainesh{fix the assumption here.} \gnote{which one?}
\end{theorem}

In order to prove this theorem, we first show that $p^2$ is anticoncentrated around the shift when the input distribution is Gaussian. Then we show that sufficiently $\ell_4$ dense directions have low-degree moments that concentrate around those of a Gaussian, which suffice to prove the aforementioned theorem. 

Now, we are ready to provide a proof of the multi-variate inequality. At a high level, we reduce to the univariate inequality and then appeal to the fact that every univariate inequality has a sum-of-squares proof.

\begin{proof}[Proof of Theorem \ref{thm:certificates-with-shifts}]

Let $q\Paren{ \Iprod{\Sigma^{\dagger/2} x, v} - \zeta  } = \norm{v}^{t} p\Paren{ \Iprod{\Sigma^{\dagger/2} x, v} / \norm{v} - \zeta }  $. Since $p^2$ is an even polynomial, we have 
\begin{equation*}
    q^2\Paren{ \Iprod{\Sigma^{\dagger/2} x, v} - \zeta  }  = \norm{v}^{2t}_2 \cdot \Paren{ \sum_{j \in [t] } c_{2j} \Paren{ \frac{ \Iprod{ \Sigma^{\dagger/2} x, v}  }{ \norm{v}  } - \zeta }^{2j} }^2 
\end{equation*}
Observe, when $x \sim \calN(0,\Sigma)$, $\Sigma^{\dagger/2} x$ has variance $1$ in every direction that lies in the span of $\Sigma$, and further by rotational invariance of Gaussians, 
\begin{equation}
\begin{split}
    \expecf{x \sim \calN(0,\Sigma) }{ q^2\Paren{ \Iprod{\Sigma^{\dagger/2} x, v} - \zeta  } }  & = \norm{v}_2^{2t} \expecf{x \sim \calN(0,\Sigma) }{ \Paren{ \sum_{j \in [t]} c_{2j}\Paren{ \frac{ \Iprod{ \Sigma^{\dagger/2} x, v}  }{ \norm{v}  } - \zeta }^{2j}  }^2  }  \\
    & = \norm{v}^{2t} \expecf{g \sim \calN(0,1) }{ \Paren{ \sum_{j \in [t]}  c_{2j} \Paren{ g - \zeta } ^{2j}  }^2}.
\end{split}
\end{equation}
Therefore, $\expecf{x \sim \calN(0,\Sigma) }{ q\Paren{ \Iprod{ \Sigma^{\dagger/2} x, v } -\zeta } } $ is a polynomial in the formal variable $\norm{v}^2$ and using Theorem \ref{thm:anti-conc-with-shifts-gaussian}, we have,
\begin{equation}
    \expecf{x \sim \calN(0, \Sigma) }{q^2\Paren{ \Iprod{ \Sigma^{\dagger/2} x, v } -\zeta }  } \leq O(\delta) \norm{v}^{2t}_2
\end{equation}
Since the aforementioned inequality is a univariate polynomial in $\norm{v}^2$, invoking Fact \ref{fact:univariate-sos-proofs}, we have 
\begin{equation*}
    \sststile{\norm{v}^2}{ 4t } \Set{  \expecf{x \sim \calN(0, \Sigma) }{q^2\Paren{ \Iprod{ \Sigma^{\dagger/2} x, v } -\zeta }  } \leq O(\delta) \norm{v}^{2t}_2  }
\end{equation*}

Now, we reduce to the Gaussian certificate by bounding each low-degree moment, when the samples are drawn from a distribution $\calD$.  We begin with invoking Lemma \ref{lem:lower-bound-k-th-moment}, for all $k \in [t]$,

\begin{equation*}
    \Set{ \Norm{v}_4^4\leq \lambda \Norm{v}_2^4 } \sststile{4k}{v} \Set{ \expecf{x \sim \calD }{\Iprod{\Sigma^{\dagger/2} x,v}^{2k}} \geq \Paren{\frac{ (2k)! }{2^{k} k!}}   \Paren{1 - \frac{(k-1) \lambda}{2}} \Norm{v}_2^{2k} }.
\end{equation*}
Similarly, invoking Lemma \ref{cor:upper-bound-k-th-moment}, for all $k \in [t]$, 

\begin{equation*}
    \Set{ \Norm{v}_4^4\leq \lambda \Norm{v}_2^4 } \sststile{4k}{v} \Set{ \expecf{x\sim \calD }{\Iprod{\Sigma^{\dagger/2} x,v}^{2k}} \leq \Paren{\frac{ (2k)! }{2^{k} k!}} \Paren{1 +  \Paren{ \lambda c_k^{k-2}k^k }} \Norm{v}_2^{2k} }.
\end{equation*}

% \ainesh{the following part of the proof needs to be done carefully. }
Since each moment is tightly concentrated around the $k$-th moment of a Gaussian, we expand $q^2$ in it's monomial basis and bound each term :
\begin{equation}
\begin{split}
    \sststile{}{}  \expecf{x\sim \calD }{q^2\Paren{ \Iprod{ \Sigma^{\dagger/2} x,v} - \zeta  } } =    \sum_{i \in [t] }\sum_{j \in [t]} \hat{c}_{i,j} \expecf{x\sim \calD }{ \Iprod{\Sigma^{\dagger/2} x , v }^{i}    } \zeta^j 
\end{split}
\end{equation}
% \begin{equation*}
%     \expecf{x\sim \calD }{q^2\Paren{ \Iprod{ \Sigma^{\dagger/2} x,v} - \zeta  } } = \expecf{}{ \sum_{i\in[d]} \sum_{j \in [d]} c_{i,j} \Iprod{x,v}^i \zeta^j  } = \sum_{j \in [d]}\zeta^j  \sum_{i\in[d]}  c_{i,j}  \expecf{}{\Iprod{x,v}^i } ,
% \end{equation*}
where each $\abs{\hat{c}_{i,j}} \leq 2^{O(t)}$. Note that we know the coefficients $c_{i,j}$ apriori and can invoke the appropriate upper bound from above for each of the terms in the expansion. 
Therefore, we have 
\begin{equation}
    \Set{ \Norm{v}_4^4\leq \lambda \Norm{v}_2^4 } \sststile{4d }{ v} \Set{\expecf{x \sim \calD }{p^2\Paren{ \Iprod{x,v} - \zeta  } } \leq   \expecf{z \sim \calN(0,1) }{ p^2\Paren{ z - \zeta } } \norm{v}_2^2 + (c_d d)^{4d} \zeta^d \lambda \norm{v}^{2}_2   }
\end{equation}
Recall, $\lambda \leq \frac{\delta}{ (c_d d)^{4d}\zeta^d   }$, and thus both terms above can be upper bounded by $O(\delta)$, completing the proof.
% \ainesh{use the homogenoes polynomial from subspace recovery.}
\end{proof}

\subsection{Proof of \cref{thm:anti-concentration-certificates-with-shifts}}

Now, we have all the tools we need to finish the proof of \cref{thm:anti-concentration-certificates-with-shifts}, restated for convenience.

\AntiConcBoundedShifts*

\begin{proof}
We know that with probability at least $99/100$, for all $\zeta \in [-\Delta, \Delta]$, we have 
\begin{equation}
\label{eqn:anti-conc-gaussian-all-shifts}
   \frac{1}{n}\sum_{i \in [n]}  p^2\paren{ g_i  - \zeta  }   \leq O\Paren{\delta} 
\end{equation}
Observe, 
\begin{equation*}
\begin{split}
    \frac{1}{n}\sum_{i \in [n]}  p^2\paren{ g_i  - y_i   } & = \frac{1}{n}\sum_{i \in [n]}  \int_{y \in [-\Delta, \Delta] } p^2\paren{ g_i  -  y  } \indic\Paren{y= y_i } \mu(y) dy \\
    & \leq \frac{1}{n} \sum_{i \in [n]} \int_{y \in [-\Delta, \Delta]} p^2\Paren{ g_i - y } \mu(y) dy \\
    & = \int_{y \in [-\Delta, \Delta]}  \frac{1}{n} \sum_{i \in [n]}  p^2\Paren{ g_i - y } \mu(y) dy \\
    & \leq   O(\delta) \int_{y \in [-\Delta, \Delta]}   \mu(y) dy,
\end{split}
\end{equation*}
where the last inequality follows from invoking Equation \eqref{eqn:anti-conc-gaussian-all-shifts} for $\zeta = y$. To complete the proof, observe, we can repeat the analysis in the proof of Theorem \ref{thm:certificates-with-shifts}, where we showed that the precondition $\norm{v}_4^4\leq \lambda \norm{v}_2^2$ implies the low-degree moments of $\Iprod{x_i , v}$ are within tiny additive error of those of a Gaussian. 
\end{proof}

%% file: re-weighting-pseudo-distributions.tex
\section{Reweighting Pseudo-distributions over the sphere}
\label{sec:re-weighting-pseudo-dist}

% \ainesh{i am happy with this section.}

In this section, we show that given a pseudo-distribution $\mu$ over a vector valued indeterminate $v \in \mathbb{R}^d$, there exists a re-weighted pseudo-distribution $\mu'$ and a fixed vector $u_{\calS}$ supported on a small subset of coordinates $\calS \subset [d]$ such that either the vector $v- u_\calS$ is \textit{analytically dense}
or all the low-degree moments of $\norm{v-u_\calS}$ are bounded in pseudo-expectation.   

\begin{restatable}[Re-weighting pseudo-distributions over the sphere]{theorem}{Reweighting}
% \begin{theorem}[Re-weighting pseudo-distributions over the sphere]
\label{thm:key-re-weighting-theorem}
Let $k,d \in \mathbb{N}$ and $ \delta, \eta >0$, let $z \in \mathbb{N} \geq \log(d)$ and $t = \bigO{ z/(\eta^2  \delta^{4k+2} ) } $. Let $\mu$ be a degree-$t$ pseudo-distribution over the vector valued indeterminate $v \in \mathbb{R}^d$ satisfying $\norm{v}_2^2 =1$. Then, Algorithm~\ref{algo:re-weighting}  outputs  a  re-weighted pseudo-distribution $\mu'$ and a fixed constant $u_{\calS}$, supported on a subset $\calS \subset [d]$ of size at most $\bigO{ 1/(\delta^{4k+2}\eta^2 ) }$ and  satisfies one of the following:
\begin{enumerate}
    \item Hypercontractivity at non-trivial scale: \begin{equation*}
    \pexpecf{\mu'}{ \Norm{v -  u_{\calS}  }_{2z}^{2z} } \leq (2\delta)^{2z} \cdot  \pexpecf{\mu'}{ \Norm{v -  u_{\calS}  }_{2}^{2z} } \textrm{ and } \pexpecf{\mu'}{\Norm{v -u_{\calS}  }_2^{2 k}}  >  \eta,
\end{equation*}
    \item Bounded pseudo-moments of the norm: $\forall y  \in [k]$, \begin{equation}
      \pexpecf{\mu'}{\Norm{v -u_{\calS}  }_2^{2 y}} \leq \eta.
    \end{equation}
\end{enumerate}
\end{restatable}

% \pravesh{the current argument does not prove 4.2. You are left with a $\sim 1/d$ tiny vector that could nevertheless be sparse. You probably need a version of 4.2 that allows additive error.}

% \begin{mdframed}
%   \begin{algorithm}[SoS Relaxation of the Anti-Concentration Program]
%     \label{algo:sos-relaxatation-anti-conc}\mbox{}
%     \begin{description}
%     \item[Input:] $n$ samples $\calX =\Set{x_i}_{i \in [n]}$ from a distribution $\calD$. \textcolor{red}{A: parameters of the algorithm need to be set correctly.}
    
%     \item[Operation:]\mbox{}
%     \begin{enumerate}
%     \item Find a degree-$\bigO{ \frac{  \log(1/\delta) \log(d) }{\delta^4} }$ pseudo-distribution $\mu$ over $w, v$ satisfying $\calA_{\delta}$. 
%     \item Search over all subsets of size $O(1/\delta^2)$. For each such subset, $\calS_j$,:
%     \begin{enumerate}
%         \item Check if for all $\ell \in \calS_j$ $\pexpecf{\mu}{v_\ell }\geq $
%         \item 
%     \end{enumerate}

%     \item Let $\mu'$ be the pseudo-distribution obtained by re-weighting $\mu$ on the subset $\calS$.  Round the pseudo-distribution to obtain an assignment $\hat{w}_i = \pE_{\mu'}[w_i]$.
%     \end{enumerate}
%     \item[Output:] An assignment such that  $\frac{1}{n}\sum_{i \in [n]} \hat{w}_i \leq \bigO{\delta} $.
%     \end{description}
%   \end{algorithm}
% \end{mdframed} 

\begin{mdframed}
  \begin{algorithm}[Re-Weighting Pseudo-Distributions]
    \label{algo:re-weighting}\mbox{}
    \begin{description}
    \item[Input:] Integers $z, k, \in\mathbb{N}$ and $0< \eta, \delta <1$. A degree-$\phi$ pseudo-distribution over a vector valued indeterminate $v\in\mathbb{R}^d$ satisfying the constraint $\norm{v}^2=1$, such that $z\geq \log(d)$ and $\phi= \Omega\Paren{z/\delta^{4k+2} \eta^2 }$.
    
    \item[Operation:]\mbox{}
    \begin{enumerate}
    \item Let $u_S=0$, $v^{(0)} = v$ and $\mu_0 = \mu$. Let $L = \bigO{ 1/(\delta^{4k+2}\eta^2 ) }$.   
    \item For $\ell \in [L]$, 
    \begin{enumerate}
        \item If $\forall y \in [k]$,  $\pexpecf{\mu}{ \norm{v^{(\ell-1)} }^{2y}_2 } \leq \eta$ or  $\pexpecf{\mu}{ \norm{v^{(\ell-1)} }^{2z}_{2z} } \leq (2 \delta)^{2z} \pexpecf{\mu}{ \norm{v^{(\ell-1)} }_2^{2z} }$, stop.
        \item Else, find a coordinate $i_\ell \in [d]$ such that $$\pexpecf{\mu_{\ell-1} }{ 
 \Paren{ v_{i_\ell}^{(\ell-1)} }^{2z}  } > \delta^{2z} \cdot  \pexpecf{\mu_{\ell-1} }{ \norm{v^{(\ell-1)}}_2^{2z} }.$$ 
        \item Compute a re-weighted pseudo-distribution $\mu_\ell$ such that \cref{lem:key-coordinate-lemma} is satisfied. 
        \item Set $v^{(\ell)} =v^{(\ell)}  - \pexpecf{\mu_\ell }{ v_{i_{\ell}} } $ and $u_\calS = u_{\calS} + \pexpecf{\mu_\ell }{ v_{i_{\ell}} }$. 
    \end{enumerate}
    \end{enumerate}
    \item[Output:] A re-weighted pseudo-distribution $\mu'$ and a fixed constant $u_{\calS}  \in \mathbb{R}^d$ supported on at most $\bigO{ 1/\delta^{4k+2}\eta^2 }$ coordinates such that the conclusion of Theorem~\ref{thm:key-re-weighting-theorem} holds. 
    \end{description}
  \end{algorithm}
\end{mdframed}

We will prove this theorem in the rest of the section. But first, we need a few technical ingredients, starting with the following re-weighting lemma from Barak, Kothari and Steurer~\cite{barak2017quantum}:

% \begin{theorem}[Subspace Reweighting, Lemma 6.6~\cite{barak2017quantum}]
% \label{thm:subspace-reweighting}
% Given a constant $c \geq 1$ and $t\in\mathbb{N}$, let $\mu$ be a degree  $t= \bigO{\frac{d \log(d)^{c'} }{\delta} }$ pseudo-distribution over the vector valued indeterminate $v \in \mathbb{R}^d$, such that $\mu$ satisfies $\Set{ \norm{v}^2 \leq 1 }$. If $\pexpecf{\mu}{ \norm{ v }^2 } \geq d^{-c}$, there exists a degree-$t/100$ re-weighting $\mu'$ of $\mu$ such that  
% \begin{equation*}
%     \Norm{ \pexpecf{\mu'}{ v^2 } } \geq \Paren{1-\delta} \pexpecf{\mu}{  \Norm{ v^2 }  }.
% \end{equation*}
% Further, the re-weighting polynomial $q(v) =  \mu'/\mu$ can be found in $2^{O(t)}$ time, has coefficients bounded by  $2^{O(t)}$ and satisfies $\norm{q(v)} \leq t^{O(t)} \norm{v}^{t} $. 
% \end{theorem}

\begin{lemma}[Scalar Reweighting, Lemmas 6.1, 6.3~\cite{barak2017quantum}]
\label{lem:scalar-reweighting}
Given $t\in \mathbb{N}$ let $\mu$ be a pseudo-distribution over $\mathbb{R}$ and $z$ be an indeterminate such that $\pexpecf{\mu}{z^2 }\geq1$ and $\Set{ z^2 \leq \Delta^2 }$, for some $\Delta \geq 1$.  Given $\eps>0$, there is a degree-$\Paren{ct \log(\Delta)/\eps^2 }$ re-weighting of $\mu$, for a fixed constant $c\ge1$, denoted by $\mu'$ such that 
\begin{enumerate}
    % \item Bounding the variance: 
    % \begin{equation*}
    %     \pexpecf{\mu' }{ \Paren{  z -   \pexpecf{\mu'}{z }   }^{t}  } \leq \eps^{t} \pexpecf{\mu'}{z }^{t},
    % \end{equation*}
    \item Fixing the moments:  For some $m\in \mathbb{R}$ be such that $\abs{m} \geq \sqrt{ \pexpecf{\mu'}{z^2} }$,
    \begin{equation*}
    \pexpecf{\mu' }{ \Paren{  z -  m }^{2t}  } \leq \eps^{2t} m^{2t},
    \end{equation*}
    \item Concentration : 
    \begin{equation*}
        \Paren{1+\eps} m  \geq \pexpecf{\mu'}{z} \geq \Paren{1-\eps} m .
    \end{equation*}
\end{enumerate}
\end{lemma}

We note that the last inequality above  follows from observing that by Holder's inequality for pseudo-distributions,
\begin{equation*}
    \Paren{\pexpecf{\mu'}{ z -  m }}^{2t} \leq  \pexpecf{\mu'}{ \Paren{ z -  m }^{2t} } \leq \eps^{2t} m^{2t},
\end{equation*}
where the last inequality follows from property (2) in 
 \cref{lem:scalar-reweighting}. Taking the $2t$-th root implies $\pexpecf{\mu'}{z}  = \Paren{1\pm\eps}  m$. 
Next, we prove the following lemma that shows that we can fix each coordinate, given that a large enough pseudo-moment of the coordinate is non-trivially large. 

% \pravesh{1. Add $||v||_2^2\leq 1$ as a constraint. 2. Assume $k$ is a power of $2$. 3. Make sure that the scaling of $\delta$ is right across the whole proof. }

\begin{lemma}[Key Coordinate Re-weighting Lemma]
\label{lem:key-coordinate-lemma}
Let $\mu$ be a degree-$\phi$ pseudo-distribution over $v \in \mathbb{R}^d$ satisfying $\Set{\norm{v}^2  \leq 1}$ and let $t \in \mathbb{N}$. Further, assume $\pexpecf{\mu}{ \norm{v}_2^{2k} } \geq \Delta$ and that there exists an $i \in [d]$ and $k \in \mathbb{N}$ such that  $\pexpecf{\mu}{v_i^{2k}} \geq \delta^{k} \pexpecf{\mu}{ \norm{v}_2^{2k}  }$. Given $\eps>0$, there exists a re-weighting of $\mu$, denoted by $\mu'$ such that 
\begin{enumerate}
    \item $\Paren{ \pexpecf{\mu'}{ v_i }}^2 \geq \Paren{1-\eps}^{8}    \delta    \Paren{ \pexpecf{\mu}{\norm{v}_2^{2k} } }^{1/k} $, and,
    \item $ \pexpecf{\mu'}{ \Paren{ v_i 
 -  \pexpecf{\mu'}{v_i}  }^{2t} } \leq  \eps^{2t} \Paren{\pexpecf{\mu'}{v_i} }^{2t}$,
    \item $\forall \ell \in  [k]$, let $q$ be the smallest power of $2$ larger than $2\ell$. Then, $\pexpecf{\mu'}{ v_i^{2\ell} } \geq \Paren{1-\eps}^8 \delta^{ q } \Delta$. 
\end{enumerate}
assuming  $\phi \geq \Paren{ c k +   \Paren{ ct \log\Paren{\frac{1}{\delta \Delta  }} } / \Paren{\eps^2\delta^2} }$, for a large enough fixed constant $c$.
\end{lemma}
\begin{proof}
% \pravesh{The degree bound shouldn't need $\pE_{\mu} \norm{v}_2^{2k}]$ in the denominator here.}
% \ainesh{k needs to be a pwer of 2}

First, we consider the case where $k =1$, i.e. $\pexpecf{\mu}{ v_i^2 } \geq \delta \pexpecf{\mu}{ \norm{v}_2^{2} } \geq \delta \Delta$. Then, consider the indeterminate $z_i = \frac{ v_i }{ \sqrt{ \delta   \Delta} }$. We have 
\begin{equation}
\label{eqn:variance-lower-bound-k=2}
   \pexpecf{\mu }{z_i^2 } = \frac{\pexpecf{\mu}{v_i^2} }{\delta\Delta } \geq 1 ,
\end{equation}
where the inequality follows from our assumption.
Further,  since $\Set{\norm{v}^2 \leq 1} \sststile{4}{v} \Set{v_i^2 \leq 1 }$, for all $i \in [d]$, we obtain the following: 
\begin{equation}
    \sststile{}{} \Set{ z_i^2 \leq \frac{ v_i^2}{ \delta  \Delta  } \leq   \frac{ 1 }{ \delta \Delta } }.
\end{equation}
% First, we consider the case where $k= 2$. Consider the indeterminate $z_i = \frac{ v_i^2 }{\delta \sqrt{ \pexpecf{}{\norm{v}_2^4}}}$. Then,
% \begin{equation}
% \label{eqn:variance-lower-bound-k=2}
%    \pexpecf{\mu }{z_i^2 } = \frac{\pexpecf{\mu}{v_i^4} }{\delta^2 \pexpecf{\mu}{\norm{v}_2^4} } \geq 1 
% \end{equation}
% where the inequality follows from our assumption. Further, since $\Set{\norm{v}^2 \leq 1} \sststile{4}{v} \Set{v_i^2 \leq 1 }$, for all $i \in [d]$, we have 
% \begin{equation}
%     \sststile{}{} \Set{ z_i \leq \frac{ v_i^2}{ \delta \sqrt{ \pexpecf{\mu}{ \norm{v}^{4}_2 } } } \leq   \frac{ 1 }{ \delta \pexpecf{\mu}{ \norm{v}^4_2 } } }.
% \end{equation}
Invoking Property (1) in \cref{lem:scalar-reweighting}, since $\mu$ has degree at least $\Omega( t\log(1/\delta\Delta)/\eps^2)$,  we have 
\begin{equation}
    \pexpecf{\mu'}{ \Paren{ v_i  -   \pexpecf{\mu}{v_i}   }^{2t} } \leq \eps^{2t} \pexpecf{}{v_i}^{2t} , 
\end{equation}
as desired. 
% we have
% \begin{equation*}
%     \pexpecf{\mu' }{ \Paren{  \frac{v_i^2}{\delta  
%  \sqrt{ \pexpecf{\mu}{ \norm{v}^4 } } } - \sqrt{ \pexpecf{\mu}{  \frac{v_i^4}{\delta^2 \pexpecf{\mu}{ \norm{v}^4 }}  } } }^{2t}  } \leq \eps^{2t} \frac{ \pexpecf{\mu}{v^4_i }^{2t} }{ \delta^{2t} \Paren{ \pexpecf{\mu}{ \norm{v}^4 }}^{t}  },  
% \end{equation*}
Further, using Property (2) in \cref{lem:scalar-reweighting}, we have
\begin{equation}
\begin{split}
    \Paren{ \pexpecf{\mu'}{ v_i  } }^2  & \geq \Paren{1-\eps}^2    \pexpecf{\mu'}{ v_i^2 }   \geq  \Paren{1-\eps}^2   \pexpecf{\mu}{ v_i^2 }  \geq \Paren{1-\eps}^2   \delta    \pexpecf{\mu}{ \norm{v}_2^{2} } ,
\end{split}
\end{equation}
where the second inequality follows from Fact~\ref{fact:scalar-reweighting-monotone}  and the last inequality follows from the definition of the $i$-th coordinate.

Now, consider the case where $k > 1$ and the indeterminate $\frac{ v_i^{k} }{ \delta^{k/2}  \sqrt{ \Delta } } $. Again, by assumption, we have 
$\frac{\pexpecf{\mu}{ v_i^{2k} } }{ \delta^{k} \Delta } \geq 1$. 
Further, $\sststile{}{} \Set{   \frac{v_i^{2k}}{\delta^{k} \Delta }  \leq \frac{1}{\delta^{k} \Delta }}$. 
Using Property (2) in \cref{lem:scalar-reweighting}, and recalling that $k$ is even, we can set $m = \sqrt{ \pexpecf{\mu_1}{v_i^{2k}} }$ and obtain a re-weighting $\mu_1$ of $\mu$ such that 
\begin{equation}
\label{eqn:lower-bound-pE-v_i^k}
\begin{split}
       \pexpecf{\mu_1}{   v_i^{k}   }   & \geq \Paren{1-\eps} \sqrt{ \pexpecf{\mu_1}{  v_i^{2k} }    }  \geq \Paren{1-\eps} \sqrt{ \pexpecf{\mu}{  v_i^{2k} }    } \geq \Paren{1-\eps} \Paren{ \delta^{k/2} \sqrt{\Delta }},
\end{split}
\end{equation}
where the second inequality follows from invoking Fact~\ref{fact:scalar-reweighting-monotone}  and the last inequality follows from our assumption about coordinate $i$.
Further,  for all even $\ell\in [k , 2k-2]$, we have
\begin{equation}
    \pexpecf{\mu_1}{ v_i^{\ell } } \geq   \pexpecf{\mu_1}{ v_i^{2k } } \geq \pexpecf{\mu}{ v_i^{2k } }  \geq \delta^{k} \Delta,
\end{equation}
where the first inequality follows from the fact that $\Set{\norm{v}^2 \leq 1 } \sststile{2}{v} \Set{ v_i^2 \leq 1 }$, the second follows from Fact~\ref{fact:scalar-reweighting-monotone} and the last one follows from our assumption on coordinate $i$.

We then repeat this argument by considering the indeterminate $  \frac{ v_i^{k/2} }{ \Paren{ \delta^{k}\Delta }^{1/4} }$, and obtain a pseudo-distribution  $\mu_2$ which is a re-weighting of $\mu_1$ such that 
\begin{equation}
    \pexpecf{\mu_2}{ v_i^{k/2} } \geq \Paren{1-\eps} \sqrt{ \pexpecf{\mu_1}{ v_i^{k} }  } \geq \Paren{1-\eps}^{1.5} \delta^{k/4 }  \Delta^{1/4},
\end{equation}
where the last inequality follows from plugging in the lower bound in Equation~\eqref{eqn:lower-bound-pE-v_i^k}. Further, for all even $\ell \in [k/2, k]$
\begin{equation}
    \pexpecf{\mu_2}{ v_i^{\ell}  } \geq  \pexpecf{\mu_2}{ v_i^{k}  } \geq \pexpecf{\mu_1}{ v_i^{k}  } \geq \Paren{1-\eps}\Paren{ \delta^{k/2}  \sqrt{\Delta} },
\end{equation}
where the first inequality follows from $\Set{\norm{v}^2 \leq 1}$, the second follows from Fact~\ref{fact:scalar-reweighting-monotone} and the last follows from Equation \eqref{eqn:lower-bound-pE-v_i^k}.

Repeating this argument $L = O(\log(k))$ times, we obtain a pseudo-distribution $\mu_{L}$ such that 
\begin{equation}
   \pexpecf{\mu_{ L } }{v_i^2 } \geq  \Paren{1-\eps}^{4} \delta   \Delta^{1/k}.
\end{equation}
Since we can set $\eps < 0.01$,  we can invoke the argument for the case of $k=1$, and conclude 
\begin{equation*}
    \pexpecf{\mu_L}{ \Paren{v_i - \pexpecf{\mu}{v_i } }^{2t} } \leq \eps^{2t} \Paren{ \pexpecf{\mu}{ v_i  } }^{2t},
\end{equation*}
and 
\begin{equation}
    \Paren{ \pexpecf{\mu_{L }}{v_i  } }^2 \geq \Paren{1-\eps}^{8}  \delta   \Paren{ \pexpecf{\mu}{ \norm{v}_2^{2k} }}^{1/k} .
\end{equation}
Finally, for all $\ell \in [k]$, let $q$ be the next largest power of two. Then, 
\begin{equation}
    \pexpecf{\mu_L}{ v_i^{\ell }} \geq \Paren{1-\eps}^4 \delta^{q/2} \Delta.
\end{equation}
Setting $\mu'=\mu_L$ concludes the proof.
\end{proof}

Next, we also need the following lemma to show that if the $2z$-th pseudo-moment of each coordinate is bounded, then the $\ell_{2z}^{2z}$-norm of the resulting vector must be bounded, as long as $z\geq \log(d)$.

\begin{lemma}[Bounded Coordinates Implies $z$ to $2$ Hypercontractivity]
\label{lem:bounded-coordiates-to-hypercontractivity}
Given $d\in\mathbb{N}$,   let $z \geq \log(d)$ and let $\mu$ be a degree $2z$  pseudo-distribution on a $d$-dimensional vector valued indeterminate $v$. If for all $i \in [d]$,  $\pexpecf{\mu}{v_i^{2z}} \leq \delta^{2z} 
\pexpecf{\mu}{ \norm{v}^{2z}_2 }$, then $\pexpecf{\mu}{ \norm{v}_{2z}^{2z} } \leq (2\delta)^{2z}\pexpecf{\mu}{ \Norm{v}_2^{2z}} $.
\end{lemma}
\begin{proof}

Summing over all $i\in[d]$, using the hypothesis, we have
\begin{equation*}
\begin{split}
    \pexpecf{\mu}{\norm{v}_{2k}^{2k}} = \sum_{i\in [d]} \pexpecf{\mu}{ v_i^{2k}} & \leq d \delta^{2z}  \pexpecf{\mu}{ \norm{v}_2^{2z}}   \leq (\delta)^{2z} \frac{d}{2^{2z}} \pexpecf{\mu}{ \norm{v}_2^{2z}}  \leq \Paren{2\delta}^{2z} \pexpecf{\mu}{ \norm{v}_2^{2z}},
\end{split}
\end{equation*}
which concludes the proof.
\end{proof}

% \ainesh{leave this note and hard instance in?}
% The degree bound in the aformentioned lemma seems necessary since we can construct the following hard example for distributions over unit vectors that demonstrates that in expectation, each coordinate being small in magnitude, does not imply that the $\ell_{4}^4$ norm is bounded. 

% \begin{lemma}[Hard Example]
% Let $\mu$ be the uniform distribution over the standard basis vectors. Then, for all $i \in [d]$, for any $k \in \mathbb{N}$, $\expecf{\mu}{ v_i^{2k} } =1/d$ but $\expecf{\mu}{ \norm{v}_4^4 } = \expecf{\mu}{ \norm{v}_2^4 }$. 
% \end{lemma}

Next, we show that a re-weighting that fixes the variance of a coordinate continues to remain fixed under further re-weightings, as long as the the resulting re-weighting polynomial has non-trivial pseudo-expectation and is mildly hypercontractive. 

\begin{lemma}[Variance remains small in subsequent re-weightings]
\label{lem:fixing-thevariance}
Let $\mu$ be a pseudo-distribution of degree $\geq t$ a $d$-dimensional vector valued indeterminate $v$ satisfying $\Set{\norm{v}_2^2 \leq  1}$ and for some $i\in [d]$ let $\pexpecf{\mu}{ \Paren{ v_i - \pexpecf{\mu}{v_i} }^{2t} } \leq \eps^{2t}$ for all $t\in [k]$. Suppose that for some $\delta>0$, there exists a  coordinate, $j \in [d]$ such that $\pexpecf{\mu}{ v_j^{2k}} \geq \delta^{2k} \pexpecf{\mu}{ \norm{v}_2^{2k}}$ and let $\mu'$ be a re-weighting of $\mu$ satisfying the guarantees of Lemma~\ref{lem:key-coordinate-lemma}. Then, for $t = \Omega\Paren{ k \cdot \log\Paren{  \pexpecf{\mu}{ \norm{v}_2^{2k} } } }$, we have
\begin{equation*}
    \pexpecf{\mu'}{ \Paren{ v_i - \pexpecf{\mu}{v_i} } }^2 \leq \eps^{2}/\delta^2. 
\end{equation*}
\end{lemma}
\begin{proof}
Recall, by definition, the re-weighting polynomial can be assumed to be $v_j^{2k}$, i.e. 
\begin{equation}
\begin{split}
    \pexpecf{\mu'}{ \Paren{ v_i - \pexpecf{\mu}{v_i} } }^2 & = \Paren{ \frac{ \pexpecf{\mu}{  \Paren{ v_i - \pexpecf{\mu}{v_i} } v_{j}^{2k} } }{ \pexpecf{\mu}{v_j^{2k} } } }^2 \\
    & \leq  \frac{ \pexpecf{\mu}{  \Paren{ v_i - \pexpecf{\mu}{v_i} }^{2t}   }^{\frac{1}{t}}  \cdot \pexpecf{\mu}{ v_j^{ \frac{2k t}{t-1} } }^{2- \frac{2}{t}} }{ \pexpecf{\mu}{v_j^{2k} } }  \\
    & \leq \eps^2 \cdot \pexpecf{\mu}{ v_i }^2   \cdot  \pexpecf{\mu }{ v_j^{2k} }^{- \frac{2}{t}}\\
    & \leq \eps^2 \Paren{  \delta^{2k}  \pexpecf{\mu}{ \norm{v}_2^{2k} } }^{-2/t}  \pexpecf{\mu}{ v_i}^2 ,
\end{split}
\end{equation}
where the first inequality follows from applying H\'older's inequality for psuedo-distributions, and the last inequality follows from out assumption on $v_j$. Plugging in $t = \Omega\Paren{ k \pexpecf{\mu}{\norm{v}_2^{2k} } }$ yields the claim.
\end{proof}

Using these lemmas, we are ready to prove our main theorem on reweighting, restated for convenience.

\Reweighting*

\begin{proof}
We iteratively re-weight the pseudo-distribution till either the norm of the resulting vector is tiny, or we have \textit{fixed} all the large coordinates in $v$ and the resulting vector is $\ell_{2z \to 2}$-hypercontractive. 
Let $\mu_0=\mu$ be a pseudo-distribution that is consistent with the program defined in Equation \eqref{eqn:cons-system}. Let $z = 2\log(d)$, and consider the case where for all $i \in [d]$,  $\pexpecf{\mu_0}{ v_i^{2z} } \leq \delta^{2z} \pexpecf{\mu_0}{ \norm{v}_2^{2z}}$. Then, it follows from Lemma~\ref{lem:bounded-coordiates-to-hypercontractivity}, that
\begin{equation}
    \pexpecf{\mu_0}{ \norm{v}_{2z}^{2z} } \leq \Paren{2\delta}^{2z} \pexpecf{\mu_0}{ \norm{v}_2^{2z} },
\end{equation}
and therefore the claim follows from setting $\calS= \Set{\emptyset}$, and $v_{\calS} =0$.

Otherwise, we know there exists some coordinate $i_1 \in [d]$ such that $\pexpecf{\mu_0}{ v_{i_1}^{2z} } >  \delta^{2z} \pexpecf{\mu_0}{ \norm{v}^{2z}_2}$. We can then apply Lemma~\ref{lem:key-coordinate-lemma} to obtain a re-weighted pseudo-distribution $\mu_1$ such that 
\begin{equation}
\label{eqn:expectation-v-i1-large}
    \pexpecf{\mu_1}{ v_{i_1}  }^2 \geq  \Paren{1-\eps}^{4} \delta \Paren{\pexpecf{\mu_1}{ \norm{v}^{2z}}  }^{1/z} \geq \delta^2 \pexpecf{\mu_1}{ \norm{v}^2_2 }
\end{equation}
and the last inequality follows from setting $\eps$ to be appropriately small and invoking  Fact \ref{fact:pseudo-expectation-holder}. 
 
% We add the index $i$ corresponding to $v_{i_1}$ above to the set $\calS_1 = \Set{i_1}$, and let $v_{\calS_1} = v_{i_1}$. 

Now, consider the vector of indeterminates $v^{(1)} =v-\pexpecf{\mu_1}{v_{i_1} }$. More generally, we use the notation $v^{(\ell)} = v - \pexpecf{\mu_1}{v_{i_1} } -  \pexpecf{\mu_2}{v^{(1)}_{i_2} } - \ldots - \pexpecf{\mu_{\ell}}{ v^{(\ell-1)}_{i_\ell} }$ to denote the residual indeterminate after $\ell$ iterations. 

Throughout the re-weighting process, we assume that for some $y \in [k]$,  $v^{(\ell)}$ satisfies $\pexpecf{\mu_{\ell}}{ \Norm{  v^{(\ell)} }_2^{2y} } \geq \eta$, otherwise we simply stop since condition (2) holds. We show that this process must terminate since all the pseudo-moments of the $\ell_2$ norm are decreasing monotonically.

Consider the first step of this process. For any $y \in [k]$,
\begin{equation}
\label{eqn:expanding-v-1}
\begin{split}
    & \pexpecf{\mu_1}{ \Norm{  v^{(1)} }_2^{2y} } \\
    & =  \pexpecf{\mu_1}{ \Paren{ \sum_{ i \neq i_1 } v_i^2  +   \Paren{ v_{i_1} - \pexpecf{\mu_1}{v_{i_1} } }^2 }^{y} }  \\
    &  = \pexpecf{\mu_1}{ \Paren{   \norm{ v }^2  +   \Paren{ v_{i_1} - \pexpecf{\mu_1}{v_{i_1} } }^2 - v_{i_1}^2 }^{y} }    \\
    & \leq    \pexpecf{\mu_1}{ \Paren{   \norm{ v }^2   - v_{i_1}^2 }^{y} }  +  y \underbrace{ \pexpecf{\mu_1}{ \Paren{ v_{i_1} - \pexpecf{\mu_1}{v_{i_1} } }^2  \Paren{\norm{v}^2 - v_{i_1}^2 }^{2y-2} } }_{\eqref{eqn:expanding-v-1}.1} + \underbrace{ \pexpecf{\mu_1}{ \Paren{v_{i_1} - \pexpecf{\mu_1}{v_{i_1}} }^{2y} } }_{\eqref{eqn:expanding-v-1}.2}
    \\
    & \leq  \pexpecf{\mu_1}{ \Paren{   \norm{ v }^2   - v_{i_1}^2 }^{y} } + y  \pexpecf{\mu_1}{ \Paren{v_{i_1} - \pexpecf{\mu_1}{v_{i_1}}}^{4}  }^{1/2}\pexpecf{\mu_1}{\norm{v}^2 } + \eps^{2y} \pexpecf{\mu_1}{v_i}^{2y} \\
    & \leq \pexpecf{\mu_1}{ \norm{v}^{2y}  }  - \pexpecf{\mu_1}{ v_{i_1}^{2y} } + y \eps^2 \pexpecf{\mu_1}{v_{i_1}}^{2} + \eps^{2y}\pexpecf{\mu_1}{v_{i_1}}^{2y}.
\end{split}
\end{equation}
where the first inequality follow from the fact that $\sststile{}{}\Set{ (a+b)^{2t} \leq a^{2t} + 2t b a^{2t-1} + b^{2t} }$, the second inequality follows from applying SoS Cauchy-Schwarz to term \eqref{eqn:expanding-v-1}.1 and using property (2) from Lemma~\ref{lem:key-coordinate-lemma} to bound term \eqref{eqn:expanding-v-1}.2. The last inequality follows from using property (2) again and invoking Fact~\ref{fact:moment-inequality-sos}.

It follows from the constraint that $\norm{v}_2^2 =1$  and H\"older's inequality for pseudo-distributions that $\pexpecf{\mu_1}{ v_{i_1} }^2 \leq  \pexpecf{\mu_1}{ v_{i_1}^2 } \leq 1$ and similarly $\pexpecf{\mu_1}{ v_{i_1} }^{2y} \leq  \pexpecf{\mu_1}{ v_{i_1}^{2y} } \leq \pexpecf{\mu_1}{  \norm{v}^{2y}_2 } $. Further, by Lemma \ref{lem:key-coordinate-lemma}, 

\begin{equation*}
    \pexpecf{\mu_1}{  v_{i_1}^{2y} } \geq  (1-\eps)^4 \delta^{2y} \Paren{ \pexpecf{\mu}{ \norm{v}_2^{2y} } }.
\end{equation*}
Combining these bounds and recalling the constraint that $\norm{v}_2^2=1$, we have
\begin{equation}
\begin{split}
    \pexpecf{\mu_1}{ \Norm{  v^{(1)} }_2^{2y} } &\leq \pexpecf{\mu}{ \norm{v}_2^{2y}  } - ( 0.99 \delta^{2y} - \eps^{2y}) \pexpecf{\mu_1}{ \norm{v}^{2y}_2 } + \eps y  \\
    & \leq 1 - 0.9 \delta^{2y}. 
\end{split}
\end{equation}
where the last inequality follows from setting $\eps \ll \delta^{2y} /y$. Therefore, in a single iteration of the re-weighting we have made progress.

Assuming the iterative process does not halt due to either the norm becoming small or the vector satisfying $\ell_{2z \to 2}$-hypercontractivity, we show that in each subsequent iteration, the norm of the residual vector decreases additively. Consider the $\ell$-th iteration: repeating a similar argument to Equation \eqref{eqn:expanding-v-1}, we have 
\begin{equation}
\label{eqn:expanding-v-ell-iteratively}
    \begin{split}
        \pexpecf{\mu_\ell }{ \norm{v^{(\ell)}}^{2y} } & =   \pexpecf{\mu_\ell }{ \Paren{ \norm{v^{(\ell-1)}}^2 +  \Paren{ v_{i_\ell} - \pexpecf{\mu_\ell }{v_{i_\ell}}}^2 -  v_{i_\ell}^2 }^{2y} } \\
 & \leq \underbrace{ \pexpecf{\mu_\ell }{ \norm{ v^{(\ell-1)} }_2^{2y} } }_{\eqref{eqn:expanding-v-ell-iteratively}.(1) } - \pexpecf{\mu_\ell }{ v_{i_\ell}^2 }^{2y} + y\eps^2 \pexpecf{\mu_1}{v_{i_1}}^2 + \eps^{2y} \pexpecf{\mu_\ell }{v_{i_\ell}}^{2y} \\
    \end{split}
\end{equation}

We now unroll the indeterminate $v^{(\ell-1)}$, and show that at each step, the variance (and higher moments) of all the coordinates fixed in previous steps remains bounded.  
Focusing on Term \eqref{eqn:expanding-v-ell-iteratively}.(1),  it suffices to show that for any $m \in [1,\ell]$, we have: 
\begin{equation}
\label{eqn:unrolling-general-m}
\begin{split}
    \pexpecf{\mu_\ell}{ \Norm{ v^{(\ell-m)} }^{2y} } & =  \pexpecf{\mu_\ell }{ \Paren{ \norm{v^{(\ell-m-1)}}^2 +  \Paren{ v_{i_{\ell-m} } - \pexpecf{\mu_{\ell-m} }{v_{i_{\ell-m} }}}^2 -  v_{i_{\ell-m} }^2 }^{y} }   \\
    & \leq  \pexpecf{\mu_\ell }{   \norm{ v^{(\ell-m-1)} }^{2y} }   - \pexpecf{\mu_{\ell}}{ v_{i_{\ell-m} }^{2y}  }  + \underbrace{ \pexpecf{\mu_\ell }{ \Paren{v_{i_{\ell-m} } - \pexpecf{\mu_1}{v_{i_{\ell-m} }} }^{2y} } }_{\eqref{eqn:unrolling-general-m}.1}   \\
    & \hspace{0.2in}+   y \underbrace{ \pexpecf{\mu_\ell }{ \Paren{ v_{i_{\ell-m} } - \pexpecf{\mu_\ell }{v_{i_{\ell -m} } } }^2  \Paren{\norm{v^{(\ell-m)} }^2 - v_{i_{\ell-m}}^2 }^{2y-2} } }_{\eqref{eqn:unrolling-general-m}.2}   \\
\end{split}
\end{equation}
where the inequality follows from the fact that $\sststile{}{}\Set{ (a+b)^{2t} \leq a^{2t} + 2t b a^{2t-1} + b^{2t} }$ and Fact~\ref{fact:moment-inequality-sos}. We can now bound terms \eqref{eqn:unrolling-general-m}.1 and \eqref{eqn:unrolling-general-m}.2 by repeatedly invoking Lemma~\ref{lem:fixing-thevariance} as follows: 
\begin{equation}
\begin{split}
     \pexpecf{\mu_\ell }{ \Paren{v_{i_{\ell-m} } - \pexpecf{\mu_1}{v_{i_{\ell-m} }} }^{2y} } & \leq \Paren{0.1\delta^2} \pexpecf{\mu_{\ell-1} }{ \Paren{v_{i_{\ell-m} } - \pexpecf{\mu_1}{v_{i_{\ell-m} }} }^{2yt}  }^{1/t}  \\
     & \leq \Paren{0.1\delta^2}^2 \pexpecf{\mu_{\ell-2} }{ \Paren{v_{i_{\ell-m} } - \pexpecf{\mu_1}{v_{i_{\ell-m} }} }^{4yt}  }^{1/2t} \\
     & \qquad \vdots \\
     & \leq \Paren{0.1\delta^2}^m \pexpecf{\mu_{\ell-m} }{ \Paren{v_{i_{\ell-m} } - \pexpecf{\mu_1}{v_{i_{\ell-m} }} }^{2^m y t }  }^{2/ (2^m t) } \\
     & \leq \Paren{0.1\delta^2}^m \eps^{2y} \pexpecf{\mu_{\ell-m}}{ v_{i_{\ell-m}} }^{2y}
\end{split}
\end{equation}
and similarly, 
\begin{equation}
\begin{split}
    & \pexpecf{\mu_\ell }{ \Paren{ v_{i_{\ell-m} } - \pexpecf{\mu_\ell }{v_{i_{\ell -m} } } }^2  \Paren{\norm{v^{(\ell-m)} }^2 - v_{i_{\ell-m}}^2 }^{2y-1}}  \\
    & \leq \pexpecf{\mu_{\ell} }{ \Paren{ v_{i_{\ell-m} } - \pexpecf{\mu_\ell }{v_{i_{\ell -m} } } }^4  }^{1/2} \cdot \pexpecf{\mu_\ell }{ \Paren{ \norm{v^{(\ell-m)} }^2 - v_{i_{\ell-m}}^2 }^{4y-4} }^{1/2} \\
    &  \leq \pexpecf{\mu_{\ell-1} }{ \Paren{ v_{i_{\ell-m} } - \pexpecf{\mu_\ell }{v_{i_{\ell -m} } } }^{4t}  }^{1/2t}\\
    & \leq (0.01\delta^2 )^m \eps^{4} \pexpecf{\mu_{\ell-m} }{v_{i_{\ell-m} }}^4
\end{split}
\end{equation}
Combining the two equations above, and substituting back into Equation \eqref{eqn:unrolling-general-m} we have
\begin{equation}
\label{eqn:unrolling-general-m-2}
\begin{split}
    & \pexpecf{\mu_\ell}{ \Norm{ v^{(\ell-m)} }^{2y} } \\
    & \leq  \pexpecf{\mu_\ell }{   \norm{ v^{(\ell-m-1)} }^{2y} }   - \pexpecf{\mu_\ell}{ v_{i_{\ell-m} }^{2y}  }  +  \Paren{0.1\delta^2}^m \eps^{2y} \pexpecf{\mu_{\ell-m}}{ v_{i_{\ell-m}} }^{2y} + (0.01\delta^2 )^m \eps^{4} \pexpecf{\mu_{\ell-m} }{v_{i_{\ell-m} }}^4 \\
    & \leq \pexpecf{\mu_\ell }{   \norm{ v^{(\ell-m-1)} }^{2y} } - \Paren{ 0.99 \delta^{2y}  - \Paren{0.1\delta^2}^m \eps^{2y} 
  } \pexpecf{\mu_{\ell-m} }{ \norm{v^{(\ell-m-1)}}^{2y} } + (0.01\delta^2 )^m \eps^{4}  \\
  & \leq \pexpecf{\mu_\ell }{   \norm{ v^{(\ell-m-1)} }^{2y} } - 0.9 \delta^{2y} \eta
\end{split}
\end{equation} 
where the last inequality follows from recalling that $\eps\ll \delta^{2k} \eta / k$. 
Substituting the above back into Equation \eqref{eqn:expanding-v-ell-iteratively}, we have
\begin{equation}
\begin{split}
    \pexpecf{\mu_\ell}{ \Norm{v^{(\ell)} }_2^{2y}  } & \leq  \Paren{ \pexpecf{\mu_\ell}{ \Norm{v^{(\ell-1)} }^{2y}_2 }  -  0.9 \delta^{2y}  \eta } \\
    & \leq  \Paren{  \pexpecf{\mu_\ell}{ \Norm{v^{(\ell-2)} }^{2y}_2 } - 2 (0.9) \delta^{2y} \eta }  \\
    & \qquad \vdots \\
    & \leq  \Paren{  \pexpecf{\mu_\ell}{ \Norm{v }^{2y}_2 } -  \ell (0.9) \delta^{2y} \eta } \\
    & \leq  \Paren{ 1 -  0.9 \ell \delta^{2y}  \eta  }
\end{split}
\end{equation}
where the subsequent inequalities follow from repeatedly invoking Equation \eqref{eqn:unrolling-general-m}, and the last inequality uses that $\pexpecf{\mu_\ell}{\norm{v}} = 1$, since re-weightings satisfy equality constraints. Therefore, setting $\ell = \Omega\Paren{ k/(\delta^{4k+2} \eta^2) }$, either the aforementioned process stops and the resulting vector is $\ell_{2z\to 2}$-hypercontractive or for all $y \in [k]$ , $\pexpecf{\mu_\ell}{\norm{v^{(\ell)}}^{2y}  } \leq \eta $, as desired.
\end{proof}

%% file: clustering-mixtures.tex
\section{Clustering a Spectrally Separated Mixture}
\label{sec:clustering-warmup}

In this section, we demonstrate the effectiveness of our anti-concentration program to the problem of clustering a mixture of two 
reasonably anti-concentrated distributions with mean $0$ and arbitrary covariances. In the next section, we will show how to apply our techniques to $k$-clustering in general whenever they are clusterable.

In particular, let $\calM = \frac{1}{2}\calD_1( 0, \Sigma_1)  + \frac{1}{2} \calD_2( 0, \Sigma_2)$, where $\calD_1(0, I)$ and $\calD_2(0, I)$ satisfy \cref{def:reasonably-anti-concentrated-dist}. Let $X = \Set{ x_1, x_2, \ldots , x_n}$ be a set of $n$ i.i.d samples from $\calM$.

\begin{restatable}[Clustering Mixtures with Spectral Separation]{theorem}{TwoCluster}
% \begin{theorem}[Clustering Mixtures with Spectral Separation]
\label{thm:clusterin-2-spectrally-separated}
Given $0< \delta <1$, and $n \geq n_0$ samples from an isotropic mixture of two  $(\delta, \exp\paren{1/\delta^2}, O(1) )$-reasonably anti-concentrated distributions $\calD_1(0, \Sigma_1)$ and $\calD_2(0, \Sigma_2)$, let $u$ be a direction  such that $ u^\top \Sigma_1 u \leq \delta^3 u^\top \Sigma_2u$. Then, there exists an algorithm that runs in $\bigO{n^{ \Paren{  \log^2(d) \exp\Paren{1/\delta^4}^{\exp \Paren{1/\delta^2} }  }
}}$ time and with probability at least $99/100$ outputs a direction $z$ such that $z^\top \Sigma_1 z \leq \bigO{\delta } \cdot z^\top \Sigma_2 z$, whenever $n_0\geq d^{\log(d)\exp\Paren{1/\delta^4}^{\exp \Paren{1/\delta^2} }  }$.
\end{restatable}

We will prove this theorem in the rest of this section.

\begin{mdframed}
  \begin{algorithm}[Clustering Spectrally Separated Mixtures ]
    \label{algo:clustering-2-mixtures}\mbox{}
    \begin{description}
    \item[Input:] Samples $\Set{x_i}_{i \in [n]}$ from an isotropic mixture of distributions and $0< \delta < 1$. Let $r = \Omega\Paren{ \log(d)/\exp\Paren{\delta^4}^{\exp\Paren{1/\delta^2} } }$. 
    \item[Operation:]\mbox{}
    \begin{enumerate}
    \item Let $t= \bigO{r \log(d)} $. Compute a degree-$t$ pseudo-distribution $\mu$ consistent with the following constraints:
\begin{equation}
\label{eqn:cons-system-clustering-algo}
\begin{split}
    \calC_\delta = & \left \{\begin{aligned}
      &\forall i\in [n]
      & w_i^2
      & = w_i \\
      && \sum_{i \in [n]} w_i & = (1-\delta^2)n/2\\
      & \forall i\in [n] & w_i \Iprod{x_i, v}^2 & \leq c  \delta  w_i \norm{v}^2    \\
      && \norm{v}_2^2 &=1 \\
    \end{aligned}\right\}
\end{split}
\end{equation}
    \item Pick $x_j$ uniformly at random from $\Set{x_i}_{i \in [n]}$. Re-weight $\mu$ by $\Iprod{x_j , v }^{2t}$. 
    \item Let $\widetilde{\Sigma } = \pexpecf{\mu}{ vv^\top  }$. Sample $g \sim \calN\Paren{0, \Sigma}$. 
    \end{enumerate}
    \item[Output:]  $z=g/\norm{g}$ such that $ z^\top \Sigma_1 z  < \bigO{\delta} \Paren{ z^\top \Sigma_2 z  }   $. 
    \end{description}
  \end{algorithm}
\end{mdframed}

\begin{lemma}[Feasibility]
With at least constant probability, the constraint system $\calC_{\delta}$ is feasible. 
\end{lemma}
\begin{proof}
Since the mixture is equi-weighted, there is a constant probability that there are at least $n/2$ samples from $\calD_1(0,\Sigma_1)$, the cluster with smaller variance along the separating direction $u$. Let $w_i$'s denote the indicators for the points sampled from $\calD_1(0,\Sigma_1)$, denoted by $\calC_1$. By Markov, for any $x_i \in \calC_1$,
\begin{equation*}
    \Pr\left[\iprod{ x_i , u }^2 > \frac{ 1 } {\delta^2} \Paren{  u^\top \Sigma_1 u } \right]  \leq \delta^2,     
\end{equation*}
and thus for all but $1- \delta^2$ fraction of points in $\calC_1$, $\Iprod{x_i, u}^2 \leq \frac{ 1 } {\delta^2} \Paren{  u^\top \Sigma_1 u }  \leq 2 \delta \norm{u}^2 $.
\end{proof}

Next, we prove the following key sum-of-squares identity. 

\begin{lemma}[Key SoS Identity]
\label{lem:key-sos-identity}
Let $r$ be the degree of the certificate required in \cref{thm:main-anti-concentration-thm}. Then, 
\begin{equation*}
\begin{split}
\calC_\delta \sststile{r}{w,  v} \Biggl\{ & \Iprod{ \Sigma_1 , vv^\top }^2  \\
& \leq \mathcal{O}\Paren{   c_4^4    \Paren{ \delta v^\top \Sigma_2 v   - \sum_{i  \in \calC_2} q^2_i\Paren{ v,w }\cdot   w_i  \Paren{ v^\top \Sigma_2 v  - \frac{\delta \norm{v}^2 }{100  } } } +  \delta^2   }     \Biggr\} ,    
\end{split}
\end{equation*}
where $c_4^4$ is the certifiable hyper-contractivity constant. 
\end{lemma}
\begin{proof}
\begin{equation}
\label{eqn:expand-mat-inner-prod}
\begin{split}
    \calC_\delta \sststile{}{ } \Biggl\{ \Iprod{ \Sigma_1, vv^\top  }^2   
    & =  \Iprod{ 
 \frac{1}{\abs{\calC_1 }} \sum_{i \in \calC_1 } x_i x_i^\top  , vv^\top  }^2    \\
 & \leq  \underbrace{  2 \Paren{   \frac{1 }{\abs{\calC_1} } \sum_{i \in \calC_1 } w_i \Iprod{x_i , v}^2  }^2  }_{\eqref{eqn:expand-mat-inner-prod}.(1)} +2   \underbrace{ \Paren{ \frac{1}{\abs{\calC_1} } \sum_{i \in \calC_1 } (1-w_i) \Iprod{x_i , v}^2 }^2  }_{\eqref{eqn:expand-mat-inner-prod}.(2
 )} \Biggr\} 
\end{split}
\end{equation}
Using the constraint that $w_i \Iprod{x_i, v}^2 \leq \frac{ c w_i \norm{v}^2 }{\Delta} $, we can bound \eqref{eqn:expand-mat-inner-prod}.(1) as follows: 
\begin{equation}
\label{eqn:bounding-var-smaller-component}
    \calC_\delta \sststile{}{} \Set{\frac{1}{\abs{\calC_1} } \sum_{i \in \calC_1 } w_i \Iprod{x_i , v}^2 \leq  \frac{c \norm{v}^2 }{ \Delta } } 
\end{equation}
Using certifiable hyper-contractivity, we can bound \eqref{eqn:expand-mat-inner-prod}.(2) as follows: 
\begin{equation}
\label{eqn:invoking-hypercontractivity}
\begin{split}
    \calC_\delta \sststile{}{} \Biggl\{  \Paren{    \frac{1}{\abs{\calC_1} } \sum_{i \in \calC_1 } (1-w_i) \Iprod{x_i , v}^2 }^2  & \leq   \Paren{ \frac{1}{\abs{\calC_1} } \sum_{i \in \calC_1 } (1-w_i)^2 } \cdot \Paren{ \frac{1}{\abs{\calC_1} } \sum_{i \in \calC_1 }  \Iprod{x_i , v}^4 }   \\
    & \leq c_4^4 \cdot \underbrace{ \Paren{ \frac{2}{ n  }  \sum_{i \in \calC_2} w_i } }_{\eqref{eqn:invoking-hypercontractivity}.(1)} \cdot  \Paren{ v^\top \Sigma_1 v}^2   \Biggr\} ,
\end{split}
\end{equation}
where the last inequality follows from recalling our constraint that $\sum_{i \in \calC_1 } w_i +\sum_{i \in \calC_2 } w_i = n/2$. 
We focus on bounding term \eqref{eqn:invoking-hypercontractivity}.(1) using our anti-concentration certification on the cluster $\calC_2$. Invoking the certificate from  \cref{thm:main-anti-concentration-thm} along the direction $v$, 
\begin{equation}
     \sststile{}{w,v} \Set{ \delta -  \frac{2}{n} \sum_{i \in \calC_2 } w_i -  \sum_{i \in \calC_2 } q^2_i\Paren{ v,w }  
 w_i \Paren{ \delta v^\top \Sigma_2 v -  \Iprod{ x_i , v }^2  } \geq 0 }
\end{equation}
Further, the constraint system implies $$\calC_\Delta \sststile{}{w}  \Set {  w_i \Iprod{ x_i, v }^2 \leq \frac{w_i }{\Delta} \norm{v}^2 \leq  \frac{w_i \delta}{100}  \norm{ v}^2 }$$
Therefore, 
\begin{equation}
    \calC_\Delta \sststile{}{w} \Set{ \frac{2}{n}\sum_{i \in \calC_2} w_i \leq \delta  - \sum_{ \in \calC_2} q^2_i\Paren{ v,w }\cdot   w_i  \Paren{ v^\top \Sigma_2 v  - \frac{\delta \norm{v}^2 }{100  } }   }
\end{equation}
Plugging this back in \cref{eqn:invoking-hypercontractivity} and combining it with \cref{eqn:bounding-var-smaller-component,eqn:expand-mat-inner-prod}, we can conclude that 
\begin{equation*}
\begin{split}
\calC_\Delta  \sststile{}{w,  v} \Biggl\{ & \Iprod{ \Sigma_1 , vv^\top }^2  \\
& \leq \mathcal{O}\Paren{   c_4^4 v^\top \Sigma_1 v  \Paren{ \delta v^\top \Sigma_2 v   - \sum_{ i \in \calC_2} q^2_i\Paren{ v,w }\cdot   w_i  \Paren{ v^\top \Sigma_2 v  - \frac{\delta \norm{v}^2 }{100  } } } +  \delta \norm{v}^4   }     \Biggr\}   
\end{split}
\end{equation*}
The claim follows from recalling that $v^\top \Sigma_1 v \leq 2 \norm{v}^2_2$ and $\norm{v}_2^2  =1$. 
\end{proof}

As a thought experiment, assume we know $v^\top \Sigma_2 v$. Then we re-weight the pseudo-distribution $\mu$ with $(v^\top \Sigma_2 v)^{2t}$ for large $t$. Let the resulting pseudo-distribution be $\mu'$.  We show that under the re-weighting, we can lower bound expressions of the form $\pexpecf{\mu'}{  q_i^2 \Paren{ w,v } v^\top \Sigma_2 v }$. 

\begin{lemma}["Fixing" the quadratic form]
\label{lem:fixing-the-quad-form}
Let $q^2_i(v,w)$ be are degree-$r$ sum-of-squares polynomials such that $\norm{q_i^2 }^2_2 \leq d^{\mathcal{O}(r)}$.
Let $\mu$ be a degree-$\Omega\Paren{ r \log(d) }$ distribution consistent with $\calC_\Delta$. Let $\mu'$ be the pseudo-distribution obtained by re-weighting $\mu$ by $(v^\top \Sigma_2 v)^{2t}$, for $t= \bigO{r\log(d)} $. Then,
\begin{equation*}
    \expecf{\mu' }{ q_i^2\Paren{v,w} v^\top \Sigma_2 v } \geq 0.9 \expecf{\mu'}{ q_i^2\Paren{v,w} \norm{v}^2  }.
\end{equation*}
\end{lemma}
\begin{proof}
Using the scalar fixing lemma (\cref{lem:scalar-reweighting}), we have that for $\ell \in [t/100]$
\begin{equation}
    \pexpecf{\mu'}{  \Paren{ v^\top \Sigma_2 v  - \pexpecf{\mu}{ v^\top \Sigma_2 v }  }^{2\ell}  } \leq \eps^{2\ell}   \Paren{ \pexpecf{\mu}{ v^\top \Sigma_2 v }  }^{2\ell}.
\end{equation}

Let $\mu''$ be the pseudo-distribution obtained by re-weighting $\mu'$ by the unknown sum-of-squares polynomial $q_i^2(v,w)$. 

Then, applying Cauchy-Schwarz for pseudo-distributions, 
\begin{equation}
\begin{split}
     \pexpecf{\mu''}{    \Paren{ v^\top \Sigma_2 v  - \pexpecf{\mu}{ v^\top \Sigma_2 v }  }  }^{4\ell} 
    &  \leq \Paren{ \pexpecf{\mu'' }{  \Paren{ v^\top \Sigma_2 v  - \pexpecf{\mu}{ v^\top \Sigma_2 v }  }^{2\ell}    }     }^{2 } \\
    & = \Paren{ \frac{ \pexpecf{\mu' }{  q_i^2\Paren{w,v} \Paren{ v^\top \Sigma_2 v  - \pexpecf{\mu}{ v^\top \Sigma_2 v }^{2\ell} }
    } } {\pexpecf{\mu'}{ q_i^2\Paren{ w, v }  } }    }^{2 } \\
    & \leq \frac{ \pexpecf{\mu'}{  q_i^4\Paren{ w, v } }   \cdot \pexpecf{\mu'}{  \Paren{ v^\top \Sigma_2 v  - \pexpecf{\mu}{ v^\top \Sigma_2 v }  }^{4\ell} } }{  \pexpecf{\mu'}{ q_i^2\Paren{ w, v }  }^{2} }    \\
    & \leq \frac{\pexpecf{\mu'}{  q_i^4\Paren{ w, v } } }{ \pexpecf{\mu'}{ q_i^2\Paren{ w, v }  }^{2}  }    \cdot  \eps^{4\ell} \cdot  \pexpecf{\mu}{ v^\top \Sigma_2 v }^{4\ell} \\
    & \leq d^{\mathcal{O}(r)} \cdot 
    \eps^{4\ell} \cdot  \pexpecf{\mu}{ v^\top \Sigma_2 v }^{4\ell} ,
\end{split}
\end{equation}
where the last inequality follows from observing that $\pexpecf{\mu'}{ q_i^4(w,v) } \leq \norm{q_i}_2^2 \pexpecf{\mu'}{ q_i^2(w,v) }^2 \leq  d^{\mathcal{O}(r)} \pexpecf{\mu'}{ q_i^2(w,v) }^2 $.  Setting $\eps = 0.01$ and $\ell = \Omega( r \log(d) )$, and taking the $4\ell$-th root,  we get that 
\begin{equation*}
    \Abs{ \pexpecf{\mu''}{    \Paren{ v^\top \Sigma_2 v  - \pexpecf{\mu}{ v^\top \Sigma_2 v }  }  }  }  \leq  0.1  \pexpecf{}{ v^\top \Sigma_2 v }
\end{equation*}
Plugging in the definition of $\mu''$, we have 
\begin{equation*}
\begin{split}
    \pexpecf{\mu'}{ q_i^2 (w,v)    \Paren{ v^\top \Sigma_2 v  }  }   & \geq \pexpecf{\mu'}{  q_i^2 (w,v)} \pexpecf{\mu}{ v^\top \Sigma_2 v } -      0.1  \pexpecf{\mu'}{ q_i^2(w,v) } \pexpecf{\mu}{ v^\top \Sigma_2 v } \\
    & \geq 0.9 \cdot \pexpecf{\mu'}{  q_i^2 (w,v)} \pexpecf{\mu}{ v^\top \Sigma_2 v }.
\end{split}
\end{equation*}
\end{proof}

Next, we show that picking a random sample $x_j$ and re-weighting with $\Iprod{x_j,v}^{2t}$ instead, we obtain can obtain the same conclusion as \cref{lem:fixing-the-quad-form}.

\begin{lemma}[Simulating "fixing" the quadratic form]
\label{lem:simulating-fixing-the-quad-form}
Let $q^2_i(v,w)$ be are degree-$r$ sum-of-squares polynomials such that $\norm{q_i^2 }^2_2 \leq d^{\mathcal{O}(r)}$. Let $\mu$ be a degree-$\Omega\Paren{ r \log(d) }$ distribution consistent with $\calC_\Delta$. Let $j \in [n]$ be picked uniformly at random and let $\mu'$ be the pseudo-distribution obtained by re-weighting $\mu$ by $\Iprod{ x_j , v}^{2t}$, for $t = \bigO{ r\log(d) }$. Then with probability at least $1/100$, we have
\begin{equation*}
    \expecf{\mu' }{ q_i^2\Paren{v,w} v^\top \Sigma_2 v } \geq 0.1 \expecf{\mu'}{ q_i^2\Paren{v,w} \norm{v}^2  }.
\end{equation*}
\end{lemma}
\begin{proof}
Consider the following expression: 
\begin{equation*}
    \begin{split}
        \expecf{ x_j  \in \calC_2 }{  \pexpecf{\mu }{ q_i^2(v,w)  \cdot v^\top \Sigma_2 v \cdot \Iprod{x_j , v}^{2t}   }     } & = \pexpecf{\mu}{ q_i^2(v,w)  \cdot v^\top \Sigma_2 v \cdot \expecf{x_j\in\calC_2 }{ \Iprod{ x_j x_j^\top , vv^\top}^{2t}   }  } \\
        & \geq  \pexpecf{\mu}{ q_i^2(v,w)  \cdot v^\top \Sigma_2 v \cdot  (v^\top \Sigma_2 v)^{2t}  }   \\
        & \geq 0.9 \pexpecf{\mu}{  q_i^2(w,v) \cdot \norm{v}^2_2 \cdot  (v^\top \Sigma_2 v)^{2t}   },
    \end{split}
\end{equation*}
where the last inequality follows from clearing the denominator in \cref{lem:fixing-the-quad-form}.
Then, by Markov's inequality for at least $0.05n$ points $x_j \in \calC_2$,  it follows that 
\begin{equation*}
       \pexpecf{\mu }{ q_i^2(v,w)  \cdot v^\top \Sigma_2 v \cdot \Iprod{x_j , v}^{2t}   }  \geq 0.1 \pexpecf{\mu}{  q_i^2(w,v) \cdot \norm{v}^2_2 \cdot  (v^\top \Sigma_2 v)^{2t}   }.
\end{equation*}
Therefore, with probability at least $1/100$, we sample a good $x_j$, and therefore, 
\begin{equation*}
\begin{split}
    \pexpecf{\mu' }{ q_i^2(v,w)  \cdot v^\top \Sigma_2 v    }  & = \frac{ \pexpecf{\mu }{ q_i^2(v,w)  \cdot v^\top \Sigma_2 v \cdot \Iprod{x_j , v}^{2t}   }  } {  \pexpecf{\mu}{ \Iprod{x_j , v}^{2t}   }  } \\
    & \geq 0.1 \cdot \frac{  \pexpecf{\mu}{  q_i^2(w,v) \cdot \norm{v}^2_2 \cdot  (v^\top \Sigma_2 v)^{2t}   }  } {  \pexpecf{\mu}{ \Iprod{x_j , v}^{2t}   }  }  \\
    & = 0.1 \pexpecf{\mu'}{ q_i^2(v,w) \cdot \norm{v}^2  }.
\end{split}
\end{equation*}

\end{proof}

We can now show that can recover a nearly-separating direction via Gaussian rounding.

\begin{lemma}[Rounding]
\label{lem:rounding-the-pd}
\cref{algo:clustering-2-mixtures} outputs a direction $z$ such that $z^\top \Sigma_1 z \leq \sqrt{\delta}$. 
\end{lemma}
\begin{proof}
We round to the pseudo-expected density matrix $\pexpecf{\mu'}{ vv^\top}$. It is easy to check that $\tr\Paren{ \pexpecf{\mu'}{ vv^\top} } = 1 $ and   $\pexpecf{\mu'}{ vv^\top} \succeq 0$. Further, it follows from H\"older's inequality for pseudo-distributions and \cref{eqn:expand-mat-inner-prod} that with probability at least $1/100$ over the choice of $x_j$, 
\begin{equation}
\begin{split}
    \Iprod{ \Sigma_1 , \pexpecf{\mu'}{v  v^\top  }   }^2 & \leq \pexpecf{\mu'}{ \Iprod{ \Sigma_1 , v  v^\top  }^2   } \\
    & \leq \bigO{ 1 } \pexpecf{\mu'}{   \Paren{ \delta v^\top \Sigma_2 v   - \sum_{ i \in \calC_2} q^2_i\Paren{ v,w }\cdot   w_i  \Paren{ v^\top \Sigma_2 v  - \frac{\delta \norm{v}^2 }{100  } } } +   \delta  } \\
    & \leq \bigO{ \delta + \delta } .  
\end{split}
\end{equation}
Taking square-roots,  $\Iprod{ \Sigma_1 , \pexpecf{\mu}{v  v^\top  }   } \leq \bigO{ \sqrt{\delta} }$.  Let $\widetilde{\Sigma}  \sim \calN\Paren{ 0,\pexpecf{\mu}{v  v^\top  }   } $, and observe that $\expecf{}{gg^\top} = \pexpecf{\mu}{v  v^\top  }$. Therefore, 
\begin{equation}
    \expecf{g}{ g^\top \Sigma_1 g  } = \Iprod{ \Sigma_1, \expecf{}{gg^\top} } = \Iprod{ \Sigma_1, \pexpecf{}{ vv ^\top} } \leq \bigO{ \sqrt{\delta} }. 
\end{equation}
The claim follows from applying Markov's inequality. 
% Let $z_1, \ldots , z_n$ be the eigenvectors of $\pexpecf{\mu}{ vv^\top}$ and $\lambda_1, \ldots, \lambda_n$ be the corresponding eigenvalues. Then, we can re-write the above as 
% \begin{equation*}
%     \sum_{i \in [n]} \lambda_i \cdot z_i^\top \Sigma_1 z_i \leq \bigO{ \sqrt{\delta} } .
% \end{equation*}
% Since it is a density matrix, we know that the $\lambda_i$'s form a probability distribution we have $\expecf{\lambda}{ z_i \Sigma_1 z_i  } \leq \bigO{\sqrt{\delta} }$ and therefore at least one eigenvector $z$ must satisfy   $ z^\top \Sigma_1 z  \leq \bigO{\sqrt{\delta}}$. Algorithm bla tries all of them, which concludes the proof. 
\end{proof}

We finally complete the proof of \cref{thm:clusterin-2-spectrally-separated}, restated for convenience.

\TwoCluster*

\begin{proof}
    We use \cref{algo:clustering-2-mixtures}.
    We compute the degree $t$ pseudo-distribution satisfying $C_{\delta}$, which by \cref{lem:key-sos-identity} is feasible. The reweighting in step $2$ simulates fixing the quadratic form, as per \cref{lem:simulating-fixing-the-quad-form}. This gives rise to the lower bounds used in the analysis of \cref{lem:rounding-the-pd}, which will show that Gaussian rounding outputs a valid separating direction (where we use the fact that the distribution is isotropic).
\end{proof}

\section{Clustering Mixtures of Reasonably Anti-Concentrated distributions}
\label{sec:clustering-with-all-separation}

In this section, we extend our techniques to cluster a mixture of $k$ distributions, whenever they are \textit{clusterable} according to the following notion: 

\begin{definition}[$\Delta$- Clusterable Mixture]
\label{def:clusterable-mixture}
Given $\Delta \gg 1$ and an equi-weighted mixture of $k$ \textit{reasonably anti-concentrated} distributions, denoted by $\calM = \sum_{ i \in [k]} \calD\paren{\mu_i , \Sigma_i}$, we call the mixture \textit{clusterable} if for every distinct pair $i, j \in [k]$, at least one of the following holds:
\begin{enumerate}
    \item Mean Separation: $\exists v$ such that $\Iprod{\mu_i - \mu_j, v  }^2 \geq \Delta v^\top \Paren{ \Sigma_1 + \Sigma_2 } v$. 
    \item Spectral Separation: $\exists v$  such that $\Delta v^\top \Sigma_i v \leq v^\top \Sigma_j v$. 
    \item Frobenius Separation: $\norm{ \Sigma_i^{\dagger/2} \Sigma_j \Sigma_i^{\dagger/2} - I }_F^2 \geq \Delta \norm{\Sigma_i^{\dagger/2 } \Sigma_j^{1/2}  }^4_{\textrm{op}}$, where $M^{\dagger}$ denotes the Moore-Penrose psuedo-inverse.
\end{enumerate}
\end{definition}

The main result we prove in this section is as follows: 

\begin{restatable}[Clustering Mixtures of Reasonably Anti-Concentrated distributions]{theorem}{KCluster}
\label{thm:main-clusterin-theorem}
% \begin{theorem}[Clustering Mixtures of Reasonably Anti-Concentrated distributions]
% \label{thm:main-clusterin-theorem}
Given $0< \eps< 1$ and  $n \geq n_0$ samples from a mixture of $k$ $(\mathcal{O}_{k,\eps}(1), \mathcal{O}_{k,\eps}(1),\mathcal{O}(1) )$-reasonably anti-concentrated distributions that also has $\mathcal{O}_{k,\eps}(1)$-certifiably hypercontractive quadratic forms (see~\cref{def:certifiable-hypercontractivity}) denoted by $\calM= \sum_{i \in k} \calD(\mu_i,\Sigma_i)$, such that $\calM$ is $\mathcal{O}_{k,\eps}(1)$-clusterable, there exists an algorithm that runs in time $n^{ \mathcal{O}_{k,\eps}(\log^2(d) ) }$ and outputs a clustering $\Set{ \hat{ \calC }_i }_{i \in [k]}$ such that with probability at least $99/100$, for all $i \in [k]$, $\abs{ \hat{ \calC }_i \cap \calC_i } \geq (1-  \eps )n/k $, where $\calC_i$ is the set of points from component $\calD(\mu_i, \Sigma_i)$, whenever $n_0 = d^{ \Omega_{k,\eps}\Paren{\log(d)} }$. 
\end{restatable}

We begin by recalling the notion of a partial clustering, where if any pair of components are mean, spectral or Frobenius separated, we can partition the mixture into two sets such that the partition respects the true clustering and is non-trivial.

\begin{definition}[Partial Clustering]
\label{def:partitial-clustering}
Given $n$ samples from a mixture of $k$ distributions, denoted by $\calM= \frac{1}{k} \sum_{i \in [k]} \calD(\mu_i, \Sigma_i)$, let $\Set{ \calC_i }_{i\in[k]}$ denote the true clustering, i.e. $\calC_i$ corresponds to the samples generated from $\calD(\mu_i, \Sigma_i)$. A bi-partition $\calP_1, \calP_2$ of the points is $\eps$-approximate partial clustering if 
\begin{enumerate}
    \item Partition respects clustering: for all $i \in [k]$
    \begin{equation*}
          \max \Set{  \frac{k}{n} \abs{ \calC_i \cap \calP_1  } ,   \frac{k}{n} \abs{ \calC_i \cap \calP_2  }    }  \geq 1- \eps,
    \end{equation*}
    \item The partition is non-trivial: 
    \begin{equation*}
        \max_{i \in [k]} \frac{k}{n} \abs{ \calC_i \cap \calP_1  } , \max_{i \in [k]} \frac{k}{n} \abs{ \calC_i \cap \calP_2  } \geq 1- \eta. 
    \end{equation*}
\end{enumerate}
\end{definition}

\begin{mdframed}
  \begin{algorithm}[Clustering Spectrally Separated Mixtures ]
    \label{algo:clustering-k-mixtures}\mbox{}
    \begin{description}
    \item[Input:] Target accuracy  $0< \eps < 1$. Samples $\Set{x_i}_{i \in [n]}$ from an isotropic mixture of $k$ distributions $\calM = \sum_{i \in [k]} \calD(\mu_i, \Sigma_i)$ such that there exists a pair $i \neq j \in [k]$ and a direction $u$ with $ \Delta u^\top \Sigma_i u < u^\top \Sigma_j u$, where $\Delta = 1/\eps^4$. 
    Let $r = \Omega_{k,\eps}\Paren{ \log(d) }$. 
    \item[Operation:]\mbox{}
    \begin{enumerate}
    \item Run the clustering sub-routines corresponding to Theorem 4.2 in~\cite{bakshi2022robustly}. %\ainesh{needs to change.}
    \item Guess the threshold at which the resulting mixture is concentrated along the spectrally separating direction, and let this be $\delta_k$. Guess the fraction of points that are concentrated, i.e. $\Iprod{x_i, v}^2 \leq \delta_k \norm{v}^2$ for the separating direction $v$ and denote this by $\eta_k$.
    \item Let $t= \bigO{r \log(d)} $. Compute a degree-$t$ pseudo-distribution $\mu$ consistent with the following constraints:
\begin{equation}
\label{eqn:cons-system-clustering}
\begin{split}
    \calC_{\delta_k, \eta_k} = & \left \{\begin{aligned}
      &\forall i\in [n]
      & w_i^2
      & = w_i \\
      && \sum_{i \in [n]} w_i & = \eta_k  n\\
      & \forall i\in [n] & w_i \Iprod{x_i, v}^2 & \leq c  \delta_k  w_i \norm{v}^2    \\
      && \norm{v}_2^2 &=1 \\
    \end{aligned}\right\}
\end{split}
\end{equation}
    \item Terminate if the system is infeasible and return the current set as the clustering.
    \item Pick $x_j$ uniformly at random from $\Set{x_i}_{i \in [n]}$. Re-weight $\mu$ by $\Iprod{x_j , v }^{2t}$. 
    \item Let $\widetilde{\Sigma } = \pexpecf{\mu}{ vv^\top  }$. Sample $g \sim \calN\Paren{0, \Sigma}$. Consider the projected instance  $\Iprod{  x_i,  g/ \norm{g} }$, and let $\calP_1$ be the set of points such that $\Iprod{ x_i , g/\norm{g}} \leq c \delta_k$. 
    \item Recurse on $\calP_1$ and $\Set{x_i} \setminus \calP_1$.  
    \end{enumerate}
    \item[Output:]  A clustering $\Set{ \hat{C}_{i} }_{i \in [k]}$. 
    \end{description}
  \end{algorithm}
\end{mdframed}

Next, we recall results from \cite{bakshi2022robustly} that demonstrate how to obtain a partial clustering if the mixture is mean or Frobenius separated. These results do not rely on certificates of anti-concentration, and therefore can be used black box. Further, in our setting, we do consider corruptions, so the sample complexity and running time only depend on the desired accuracy.

\begin{lemma}[Partial Clustering under Mean/ Frobenius Separation,Thm 4.2~\cite{bakshi2022robustly}]
\label{lem:partial-clustering-mean-frob}
Given $0<\eps<1$ and $n \geq n_0$ samples $\Set{x_i}_{i \in [n]}$ from a mixture of $k$ \textit{reasonably anti-concentrated} distributions, $\calM = \frac{1}{k} \calD(\mu_i, \Sigma_i) $ , let $i \neq j \in [k]$ be such that $\calD(\mu_i, \Sigma_i)$ and $\calD( \mu_j, \Sigma_j )$ are either mean-separated or Frobenius separated with $\Delta = (k/\eps)^{\mathcal{O}(k)}$. Let $\Set{ \calC_i }_{i \in [k]}$ be the true clustering.
Then, there exists an algorithm that runs in $\bigO{n^{(k/\eps)^{\mathcal{O}(k)}}}$ time and with probability at least $99/100$, outputs an $\eps$-partial clustering (see \cref{def:partitial-clustering}) whenever $n_0 = \Omega\Paren{d^{(k/\eps) ^{\mathcal{O}(k)}}}$. 
\end{lemma}

% Next, we show that clustering a one-dimensional instance that is either mean separated or spectrally separated is straight-forward. 

% \begin{lemma}[$1$-dimensional Clustering]
% \label{lem:1-dim-clustering}
% Given $0<\eps<1$ $n \geq n_0$ samples from mixture of $k$ \textit{reasonably anti-concentrated} distributions in one dimension, let $\Set{ \calC_{i} }_{i \in [k]}$ be the true clustering. Then, there exists an algorithm that runs in time $??$ and outputs a clustering $\Set{ \hat{\calC}_{i} }_{i \in [k]}$ such that with probability at least $99/100$, for all $i \in [k]$, $ \frac{k}{n} \abs{ \hat{\calC}_i \cap \calC_i  } \leq \eps$, up to a permutation of the clusters. 
% \end{lemma}

Next, we show that the algorithm from \cref{sec:clustering-warmup} can be adapted to work whenever there is a multiplicative gap in the second moments, for an equi-weighted mixture of $k$ reasonably anti-concentrated distributions, where the mixture is in isotropic position.

\begin{lemma}[Anti-Concentration program]
\label{lem:anti-conc-k-clusters}
Given $n \geq n_0$ samples from a mixture of $k$ reasonably anti-concentrated distributions, $\calM = \frac{1}{k} \sum_{i \in [k]} \calD(\mu_i, \Sigma_i)$, in isotropic position, such that $\norm{\mu_i}^2 \leq \poly(k)$, let $v$ be a direction satisfying $ ( v^\top \Sigma_1 v) \leq \eps k^{-\poly(k)}  v^\top \Sigma_j v$ for some $j \in [k]$. Then, there exists an algorithm that runs in time $n^{ \mathcal{O}_{k,\eps}(\log^2(d))}$ and with probability at least $99/100$, outputs an $\eps$-partial clustering of the samples.  
\end{lemma}

\begin{proof}
Let $\sigma_1, \sigma_2, \ldots , \sigma_k$ be the variance of the unpaired mixture along $v$. We know that $\sigma_1 < k^{- k^{100}}$ and $\sigma_k \leq k$. Consider the case where $\sigma_k < 1/k^k$. Then, the must be some mean $\mu_i$ such that $\abs{ \Iprod{ \mu_i , v} } \geq 1/2$, since the mixture is isotropic and the mixture variance along $v$ is $1$ and since the mean of the mixture is $0$, there must be some $\mu_j $ such that $\abs{ \Iprod{ \mu_j , v} } \geq 1/k$ on the opposite side of $\mu_i$, along $v$. 
Therefore, there exists a pair $i,j$ such that $\Iprod{ \mu_i - \mu_j, v}^2 \geq k^{\poly(k) }  \sigma_{k}$, i.e. the mixture is mean separated and this is a contradiction. 

It suffices to consider the case where $\sigma_k \geq 1/k^k$. 
Since $\sigma_1 \leq 1/k^{k^{100}}$, observe there exists a variance  $\sigma_t \in [k^{-k^{100}}, k]$ such that $\sigma_t < k^{-k^{10}} \sigma_{t+1}$ and guess $\delta_k = \sigma_{t} \pm \eps  k^{-k^{10}} $.  
Let $\calS_1$ denote the effective component corresponding to points from $\calC_1 \ldots \calC_t$ and $\calS_2$ be the complement. The empirical mean of $\calS_1$ is bounded, so we can guess it by discretizing the interval $[-k, k]$ at granularity $\eps k^{-\poly(k)}$. Note, this only makes the success probability  $\eps k^{-\poly(k)}$, but we can boost this up by repeating the algorithm several times. Assuming we have guessed the mean, $\mu_{\calS_1}$, correctly, we can shift the samples by $\mu_{\calS_1}$ so that $\calS_1$ is mean $0$. Further, $\calS_1$ and $\calS_2$ have a spectral gap of $k^{-k^2}$ and we can essentially reduce to the $2$ component case.

% We begin by randomly pairing up the samples and taking the difference, i.e. $ x_i - x_j$. Let $\Set{y_i}_{i \in [n/2]}$, be resulting set of samples. Note, $\Set{y_i }_{i\in [n/2]}$ is now a mixture of $\binom{k}{2}$ components. We show that  there exists a pair of components that is spectrally separated in the new mixture, and one of them has mean $0$. First, note in the difference mixture, the means are still bounded, i.e. for all $i \in [ \binom{k}{2} ] $,  $\norm{ \mu_i} \leq 1/\poly(k)$. Next, observe, since the original mixture was spectrally separated, any pair of points in the component with smaller variance have mean $0$ and the covariance is at most $2$ times the original. Now, we argue that there must be a multiplicative gap of $\Omega(\Delta)$ in the second moments of all the components of the paired up mixture. 

Next, we show that we can modify the proof of  \cref{thm:clusterin-2-spectrally-separated} to work in this setting.  Let $k \cdot \eta_k$ be the guess of how many components lie in $\calS_1$. 
Next, we note that since the mixture is isotropic, $\norm{ \mu_i } \leq k$, and thus $\norm{\mu_{\calS_2} } \leq k$. 
Now, we observe that we can generalize \cref{lem:key-sos-identity} as follows: repeating the previous analysis, we have 
\begin{equation*}
\begin{split}
    \calC_{\delta_k, \eta_k}\sststile{}{} \Biggl\{   \Iprod{ \Sigma_{\calS_1}, vv^\top  } 
    & \leq \Paren{ \frac{1}{\abs{ \calS_1 } } \sum_{i \in \calS_1 } (1-w_i)  } \cdot \Paren{ c_4^4 (v^\top \Sigma_{\calS_1} v)^2} \\
    & \leq  c_4^4  \Paren{ \frac{1}{ \abs{ \calS_1} } \sum_{i \in \calS_2} w_i  } \cdot \Paren{ v^\top \Sigma_{\calS_1} v}^2 
    \Biggr\}, 
\end{split}
\end{equation*}
where the last inequality follows from $\sum_{i \in \calS_1} w_i + \sum_{i \in \calS_2 } w_i = \eta_k n $ and $\sum_{i \in \calS_1} 1  = \eta_k n$. Let $\mu_{\calS_2}$ be the mean of the components in $\calS_2$. Invoking our the shifted anti-concentration certificate from \cref{cor:anti-concentration-thm-shift}, with $\phi = \Omega( \norm{\mu_{\calS_2}}^2 )$ for the cluster $\calS_2$, we have
\begin{equation*}
\begin{split}
    \sststile{}{w,v}\Biggl\{ \delta_k - \frac{1}{\abs{\calS_2} } \sum_{i \in \calS_2} w_i & \geq \sum_{i \in \calS_2 } q^2(v,w) w_i \Paren{ \delta_k v^\top \Sigma_{\calS_2} v - \Paren{ \Iprod{ x_i -\mu_{\calS_2} , v}^2 - \phi \norm{v}^2 }  } \\
    & \geq \sum_{i \in \calS_2 } q^2(v,w) w_i \Paren{ \delta_k v^\top \Sigma_{\calS_2} v - \Paren{ \Iprod{ x_i , v}^2 + \norm{\mu_{\calS_2} }^2 \norm{v}^2 - \phi \norm{v}^2 }  } \\
    & \geq \sum_{i \in \calS_2 } q^2(v,w) w_i \Paren{ \delta_k v^\top \Sigma_{\calS_2} v - \delta \norm{v}^2/100 } \Biggr\},
\end{split}
\end{equation*}
where the last inequality follows from our constraints, and $\phi = \Omega(\norm{\mu_{\calS_2}}^2)$. Since $\abs{ \calS_1} $ and $\abs{\calS_2}$ are within a factor of $k$ of each other, we can conclude that 
\begin{equation*}
    \calC_{\delta_k, \eta_k}\sststile{}{} \Set{ \Iprod{\Sigma_{\calS_1}, vv^\top  }  \leq \bigO{ k \delta_k  }  -  \sum_{i \in \calS_2 } q^2(v,w) w_i \Paren{ \delta v^\top \Sigma_{\calS_2} v - \delta_k \norm{v}^2/100  } }.
\end{equation*}
Now, with probability at least $1/k$, we sample a point $x_i \in \calS_2$ and as before, re-weight by $\Iprod{x_i, v}^{2t}$. In expectation, this fixes $v^\top \Sigma_{\calS_2} v$, and repeating the argument in \cref{lem:simulating-fixing-the-quad-form,lem:rounding-the-pd}, we can find a direction $z$ such that 
\begin{equation*}
    z \Sigma_{\calS_1} z \leq \delta_k.
\end{equation*}
We then simply project onto the rounded direction $z$ and get an $\eps$-approximate partial clustering, where components corresponding to $\sigma_1, \ldots, \sigma_t$ are $\calP_1$ and the rest are $\calP_2$.  
\end{proof}

Now, we are ready to combine these algorithms together and complete the proof of \cref{thm:main-clusterin-theorem}, restated for convenience.

\KCluster*

\begin{proof}[Proof of \cref{thm:main-clusterin-theorem}]
Given $n$ samples from a mixture in isotropic position, we run the algorithm corresponding to the case where the mixture is mean or Frobenius separated. It follows from \cref{lem:partial-clustering-mean-frob} that if algorithm succeeds, we can simply recurse on the resulting partial clustering. 

Instead, consider the case where the mixture is not mean or Frobenius separated. Then, the means, $\mu_i$ are bounded by $\poly(k)$. Further, we know that for some pair must be spectrally separated and we can invoke \cref{lem:anti-conc-k-clusters} on the mixture and again obtain an $\eps$-partial clustering, allowing us to recurse until there's only a single component left.  
The running time is dominated by the anti-concentration program, and this concludes the proof. 
\end{proof}

%% file: list-decodable-regression.tex
\section{List-Decodable Regression}
\label{sec:list-decodable-regression}
% \ainesh{this section is also done. }
In this section, we show that our certificate can be used to obtain list-decodable regression algorithms for \textit{reasonably anti-concentrated} distributions.

\begin{model}[List-Decodable Regression]\label{model:regression}
Given $0<\alpha<1$ and $\ell^* \in \mathbb{R}^d$ such that $\norm{\ell^* } =1$, we define list-decodable learning w.r.t. a distribution $\calD$ as, denoted by $\textsf{Lin}_{\calD}(\alpha, \ell^*)$ as follows: generate $\alpha n$ i.i.d. samples, $x_i$, from $\calD$ and set $y_i = \Iprod{ x_i, \ell^*} $. Let these $\Set{ (y_i, x_i) }_{i \in [\alpha n]}$ be denoted by $\calI$. Generate the remaining  $(1-\alpha)n$ pairs arbitrarily and potentially adversarially w.r.t. $\calI$. 
\end{model}

We show that we can output a short list containing a vector close to the true regressor for any \textit{reasonably anti-concentrated} distribution. This is the task of list-decodable linear regression.

\begin{theorem}[List-decodable Regression]
\label{thm:list-decodable-regression}
For any $0< \alpha, \eta < 1$, there exists an algorithm that takes $n \geq n_0$ samples from $\textsf{Lin}_{\calD}(\alpha, \ell^*)$, where $\calD$ is a \textit{reasonably anti-concentrated} distribution, runs in $\bigO{n^{ \log(d) \exp\Paren{1/ \alpha^8 \eta^8 }^{\exp\Paren{1/\alpha^4 \eta^4} } }}$ time and outputs a list of $\bigO{1/\alpha}$ vectors such that with probability at least $99/100$, the list contains at least one vector $\hat{\ell}$ such that $\norm{ \hat{\ell} - \ell }_2^2\leq \eta$, when $n_0 \geq d^{\log(d) \exp\Paren{1/ \alpha^8 \eta^8 }^{\exp\Paren{1/\alpha^4 \eta^4} } }$. 
\end{theorem}

We consider the following constraint system: 
\begin{equation}
\label{eqn:cons-system-regression}
\begin{split}
    \calC_\alpha = & \left \{\begin{aligned}
      &\forall i\in [n]
      & w_i^2
      & = w_i \\
      && \sum_{i \in [n]} w_i & = \alpha n\\
      & \forall i\in [n] & w_i ( y_i -   \Iprod{x_i, \ell } )^2  & =  0     \\
      && \norm{v}_2^2 &=1 \\
    \end{aligned}\right\}
\end{split}
\end{equation}

\begin{lemma}[Key SoS Identity, analog of Lemma 4.1 in~\cite{karmalkar2019list}]
Given $0<\alpha, \eta <1$, let $\delta  = \alpha^2 \eta^2 /4$ and let $r = \Omega\Paren{ \log(d) \cdot \exp\Paren{ 1/\delta^4  }^{\exp\Paren{1/\delta^2}} }$. Then, 
\begin{equation*}
    \calC_{\alpha} \sststile{r }{w, v} \Set{ \frac{1}{\calI } \sum_{ i \in \calI 
    } w_i \norm{ \ell - \ell^* }^2 \leq  \delta }
\end{equation*}
\end{lemma}
\begin{proof}
Observe, for the inliers, our certificate on anti-concentration states 
\begin{equation*}
    \sststile{}{w,v} \Set{ \delta -  \frac{1}{ \abs{ \calI } } \sum_{i \in \calI } w_i -  \sum_{i \in \calI } q^2_i\Paren{ v,w }  
 w_i \Paren{ \delta v^\top \Sigma v -  \Iprod{ x_i , v }^2  } \geq 0 } 
\end{equation*}
Consider $v = \ell- \ell^*$ and observe that $w_i \Iprod{ x_i , \ell - \ell^* }^2 = 0$. 
% \ainesh{this is only true when $y_i = \Iprod{x_i , \ell^*}$ }. 
Invoking the above identity along $ v =
 \ell - \ell^*$, observe $ \delta q_i^2(v,w) w_i (v^\top \Sigma v)  $ is a sum-of-squares polynomial and thus $\sststile{}{w} \Set{ \frac{1}{\abs{\calI} } \sum_{i \in \calI }  w_i  \leq \delta }   $. By sum-of-squares triangle inequality, $\Set{ \norm{ \ell - \ell^*}^2 \leq 2 } $ since both are unit vectors and therefore
\begin{equation*}
    \calC_{\alpha} \sststile{r }{w, v} \Set{ \frac{1}{\calI } \sum_{ i \in \calI 
    } w_i \norm{ \ell - \ell^* }^2 \leq  2\delta } ,
\end{equation*}
which yields the claim. 
\end{proof}

The theorem follows from repeating the remaining arguments in \cite{karmalkar2019list}. 

%% file: distributions.tex
\section{Reasonably anti-concentrated Distributions}
\label{sec:distribution_families}
% \ainesh{can someone clean up this section? }
% \ainesh{we also need to prove closure properties, like $x_i - x_j$ remain reasonably anti concentrated, etc.  }
% \ainesh{the clustering section needs that if $\calD_1$ and $\calD_2$ are reasonably anti-concentrated separately, the uniform mixture is also anti-concentrated.  }

In this section, we show that uniform distributions over $\ell_p$ balls for $p > 1$ are reasonably anti-concentrated.

\begin{theorem}[Uniform distributions over $\ell_p$ balls]\label{thm:anti_concentrated_distributions}
For any  $p \in [1 + \alpha, \infty)$, for a fixed constant $\alpha$, and any $0<\delta<1$, the uniform distribution over the unit $\ell_p$ ball is  $\delta$-reasonably anti-concentrated and satisfies certifiable hypercontractivity of quadratic forms.
\end{theorem}
% \ainesh{add that these distributions also satisfy hypercontractivity of deg 2 polynomials.}

We need to show that the uniform distribution on $\ell_p$ balls satisfies anti-concentration, certifiable hypercontractivity,  almost product-ness and strictly sub-exponential tails.

First, the uniform distribution over any convex body is log-concave, and therefore, for any univariate projection, the PDF is bounded by $e^{ - \abs{ \Iprod{ x, v}}/c }$, for a fixed universal constant $c$, along direction $v$. It is easy to check that such a distribution is anti-concentrated. 
Next, we show that the distribution is strictly sub-exponential.

\begin{lemma}[Strictly sub-exponential]
\label{lem:strictly-sub-exp}
For any $p \in [1+\alpha, \infty)$, where $\alpha$ is a fixed constant, let $x$ be a sample from the uniform distribution on the unit $\ell_p$ ball. Then, for all unit vectors $v \in \R^d$ and $t > 0$, we have
\begin{equation*}
    \Pr[ \abs{ \Iprod{x, v} } \geq t \sqrt{v^T \Sigma v} ] \leq \exp\Paren{ -t^{1+\alpha} / c  },
\end{equation*}
where $\Sigma$ is the covariance of the distribution.
\end{lemma}

\begin{proof}
This result follows from \cite[Proposition 10]{barthe2005probabilistic}, where for $p \ge 2$, we can choose $\alpha = 1$ because the distribution is sub-Gaussian, i.e. the $\psi_2$ norm is bounded by a fixed universal constant. For $p \in (1+\alpha, 2]$, the $\psi_p$ norm is bounded by a fixed constant. 

\end{proof}

To show almost-productness, we first show that the uniform distribution over $\ell_p$ balls satisfy
negative association.  We invoke a theorem from \cite{pilipczuk2008negative}, which shows this property for a larger class of uniform distributions over Generalized Orlicz balls.

% \begin{definition}[Unconditional Convex Body]
% A convex body $\calK \subset \R^d $ is unconditional if it is bounded, centrally symmetric (if $x \in \calK$, then $-x \in \calK$), has a non-empty interior, and for all $\zeta \in \{-1, 1\}^d$, if $x \in \calK$, $\Paren{ \zeta_1 x_1, \zeta_2 x_2, \ldots, \zeta_d x_d} \in \calK$.  
% \end{definition}

\begin{definition}[Generalized Orlicz Ball]
A function $f:\R^+ \to \R^+ \cup \{\infty\}$ is a Young function if $f$ is convex, $f(0)=0$, $\exists x$  such that $f(x)\neq 0$, and $\exists x \neq 0$ such that $f(x)\neq \infty$. Let $\calF = \{ f_i \}_{i \in [d]}$ be a set of $d$ Young functions. Then, 
\begin{align*}
    \calK =\{ x \; :\; \sum_{i \in [d]} f_i(|x_i|) \leq 1  \} 
\end{align*}
is a Generalized Orlicz Ball. 
\end{definition} 

Note in particular that $\ell_p$ balls are Generalized Orlicz balls.

\begin{definition}[Negative Association]
\label{def:negative-association}
A set of random variables $\{x_i\}_{i \in [d]}$ is negatively associated if for any coordinate-wise increasing, bounded, functions $f, g$ and disjoint sets $\calS, \calT \subset [d]$ such that $\calS =\{i_1, i_2, \ldots i_s\},\calT =\{ j_1, j_2, \ldots j_t\}$,
\begin{equation*}
    \expecf{}{ f\Paren{ x_{i_1}, x_{i_2}, \ldots, x_{i_s} } \cdot g\Paren{ x_{j_1}, x_{j_2}, \ldots, x_{j_s} } } \leq \expecf{}{ f\Paren{ x_{i_1}, x_{i_2}, \ldots, x_{i_s} }  } \cdot \expecf{}{ g\Paren{ x_{j_1}, x_{j_2}, \ldots, x_{j_s} } }.
\end{equation*}
\end{definition}

\begin{theorem}[Generalized Orlicz Ball in the Positive Orthant, Theorem 1.2 in \cite{pilipczuk2008negative}]
\label{thm:negative-association}
Let $\calK$ be a Generalized Orlicz ball in $\R^d$ and let $\calD$ be the uniform distribution over $\calK$. If $x \sim \calD$, then the set of random variables $\Set{ \abs{ x_i }  }_{i \in [d]}$ is negatively associated. 
\end{theorem}

This establishes one side of almost $k$-wise independence. We establish the other side via standard integration techniques. In particular, we show that several families of distributions satisfy the following criterion:

\begin{lemma}[The solid $\ell_p$ ball is almost-product]\label{lem: expectation_of_coordinate_prods_lpball}
Let $x\sim \calD$ where $\calD$ is the uniform distribution on the solid $\ell_p$ ball in $\R^d$, that is, $\mathbb{B}_p = \{x \in \R^d, \norm{x}_p^p \le 1\}$, for any $p \geq 1$. Then, for any subset $S\subset [d]$, such that $|S|=k$,
\[
\expecf{}{\prod_{j \in S} x_j^2} \ge \Paren{ 1 - o_d(1)} \prod_{j \in S}\expecf{}{x_j^2} ,
\] 
where $k\ll d$.
\end{lemma}
% As we noted above, enough to consider surface of $l_p$ ball, so $\{x : \norm{x}_p = 1\}$. So, $\sum x_j^p = 1$.

% But actually, we can use \cite{calafiore1998uniform} directly.They give a procedure to generate a random point on the surface of an $l_p$ ball and then multiply but an according radius vector. In particular, to generate a uniform point on the solid $l_p$ ball of radius $1$, we first sample $y = (y_1, \ldots, y_n)$ iid from the generalized normal distribution with density $f(t) = \frac{p}{2\Gamma(1/p)}e^{-|t|^p}$ and independently sample $w \sim Unif([0, 1])$ and set
% \[x = w^{1/d}\frac{y}{\norm{y}_p}\]

\begin{proof}
To prove this, we can compute these expectations explicitly as follows: 
\[
    \expecf{x\sim \calD}{ \prod_{i \in S} x_i^2} = \frac{\Gam(3/p)^k\Gam(d/p  + 1)}{\Gam(1/p)^k\Gam((d + 2k)/p + 1)}.
\]
The proof of this equality follows \cite{wang2005volumes} so we only give a brief sketch.
Let the uniform distribution on the standard $\ell_2$ ball $B_2^d$ be $\calD'$. Let $\phi(y) = (y_1^{2/p}, \ldots, y_d^{2/p})$ with the corresponding Jacobian determinant $J\phi(y) = (2/p)^dy_1^{2/p-1} \ldots y_d^{2/p - 1}$. Then,
    \[\int_{\calD} f(x)\,dx = \int_{\calD'} f(x) |J\phi(y)|\,dy = (2/p)^d\int_{\calD'} f(\phi(y)) |y_1|^{ 2/p - 1}\ldots |y_d|^{2/p - 1} \,dy\]
    Therefore,
    \[\int_{\calD} (\prod_S x_i^2) \,dx = (2/p)^d\int_{\calD'} (\prod_S |y_i|^{4/p}) |y_1|^{ 2/p - 1}\ldots |y_d|^{2/p - 1} \,dy\]
    
    For reals $\al_1, \ldots, \al_n > -1$, we have
    \[\int_{\calD'} |y_1|^{\al_1}\ldots |y_d|^{\al_d}\,dy =  \frac{\Gamma(\beta_1)\ldots \Gamma(\beta_d)}{\Gamma(\beta_1 + \ldots + \beta_d + 1)}\]
    where $\beta_i = (\al_i + 1) / 2$ for all $i$. Using this, we get
    \[\int_{\calD} (\prod_S x_i^2) \,dx = (2/p)^d\frac{\Gam(3/p)^k\Gam(1/p)^{d - k}}{\Gam((d + 2k)/p + 1)}\]
    Therefore,
    \[\EE_{\calD}[\prod_{S} x_i^2] = \frac{\int_{\calD} (\prod_S x_i^2) \,dx}{\int_{\calD} 1 \,dx} = \frac{\Gam(3/p)^k\Gam(d/p  + 1)}{\Gam(1/p)^k\Gam((d + 2k)/p + 1)}\]

Next, we use the above characterization to complete the proof. 
    \begin{align*}
    \frac{\EE[\prod_{j \in S} x_j^2]}{\prod_{j \in S} \EE[x_j^2]} &= \frac{\Gam(d/p+1 + 2/p)^k}{\Gam(d/p + 1)^{k - 1}\Gam((d/p + 1) + (2k/p))}
    \end{align*}
    
    Because the gamma function is strictly logarithmically convex on the positive reals, we have
    \[\Gam(x_1)^t\Gam(x_2)^{1 - t} \ge \Gam(tx_1 + (1 - t)x_2)\]
    for reals $x_1, x_2 > 0$ and $t \in [0, 1]$. Setting $x_1 = d/p + 1 + 2/p, y = d/p + 1, t = k$, we get
    \[\frac{\Gam(d/p+1 + 2/p)^k}{\Gam(d/p + 1)^{k - 1}} \ge \Gam((d/p + 1) + (2k/p) + (1 - k))\]
    Assuming $A = (d/p + 1) + (2k/p) \ge k - 1$ (follows if $d \gg pk$), we get
    \begin{align*}
    \frac{\EE[\prod_{j \in S} x_j^2]}{\prod_{j \in S} \EE[x_j^2]} &\ge \frac{\Gam(A + (1 - k))}{\Gam(A)}\\
    &= \frac{1}{(A - (k - 1))(A - (k - 2)) \ldots A}
    \end{align*}
    If we assume $p, k$ are constants and $d$ is growing, then $A \gg k$ in which case this will be $1 - o_d(1)$.
\end{proof}

% \begin{lemma}
%     Let $x\sim \calD$ where $\calD$. Then, for any subset $S\subset [d]$, such that $|S|=k$,
    
% \end{lemma}

% \begin{proof}
    
% \end{proof}

Next, we show that we can obtain certifiable hypercontractivity of linear forms.

Finally, we show that distributions that satisfy a a Poincar\'e inequality also have certifiably-hypercontractive quadratic forms (recall ~\cref{def:certifiable-hypercontractivity}). 

\begin{lemma}[Certifiable Hypercontractivity from Poincar\'e]
\label{lem:cert-hyper-poincare}
A distribution $\calD$  over $\mathbb{R}^d$ is $\sigma$-Poincar\'e if for all differentiable functions $f: \mathbb{R}^d \to \mathbb{R}$, 
\[ 
\expecf{x \sim \calD }{  \Paren{ f(x) - \expecf{x \sim \calD }{ f(x) } }^2   } \leq \sigma^2  \norm{ \nabla  f(x) }_2^2,
\]
for some fixed constant $\sigma$. Any $\sigma$-Poincar\'e distribution has $t^2$-certifiably $(32\sigma t )^t$-hyper-contractive quadratic forms, for all even $t$.
\end{lemma}
\begin{proof}
We proceed via induction, and consider the hypothesis that for even $t$,  
$$ \sststile{4t^2 }{Q} \Set{ \expecf{\calD}{  \Paren{ \Iprod{x, Qx} - \expecf{x\sim \calD }{ \Iprod{x, Qx} } }^{t}  } \leq (32\sigma t)^t \norm{Q}_F^{t} } $$ 

Consider the base case, i.e.
for a matrix-valued indeterminate $Q$, let $f_0(x) = \Iprod{x, Qx} - \expecf{x\sim \calD }{ \Iprod{x, Qx} }$. 
Then, 
\begin{equation*}
    \expecf{x \sim \calD }{ (f_0(x))^2  } \leq \sigma^2 \norm{ \nabla f_0(x) }_2^2 ,
\end{equation*}
where 
\begin{equation*}
    \norm{ \nabla f_0(x) }_2^2 = 4 \expecf{x\sim \calD }{ 
 x^\top Q^\top Q x } = 4 \Iprod{ QQ^\top , \Sigma } = 4 \norm{Q}_F^2, 
\end{equation*}
and thus 
\begin{equation*}
    \sststile{}{Q} \Set{ \expecf{x \sim \calD }{ \Paren{ \Iprod{x, Qx} - \expecf{x\sim \calD }{ \Iprod{x, Qx} } }^2  } \leq (2\sigma)^2 \norm{Q}_F^2 }.
\end{equation*}
Note, the above inequality admits a sum-of-squares proof since it is a homogeneous degree-$2$ polynomial inequality and can be reformulated as the quadratic form of a PSD matrix. 
Next, we can re-arrange the Poincar\'e inequality as follows: 

\begin{equation}
\label{eqn:re-arranged-poincare}
    \expecf{x\sim \calD }{f(x)^2 }\leq \expecf{x \sim \calD}{ f(x) }^2 + \sigma^2 \norm{\nabla f(x) }^2_2
\end{equation}
Next, following a simple calculation
\begin{equation*}
     \nabla f(x) = t   \Paren{ \Iprod{x , Qx} - \expecf{}{\Iprod{x , Qx}}^{t-1} } Qx 
\end{equation*}
and therefore, 
\begin{equation}
\label{eqn:bound-on-gradient}
\begin{split}
    \sststile{}{Q} \Biggl\{ \norm{\nabla f(x) }_2^{2t} & = t^{2t} \Paren{ \expecf{}{ \Paren{ \Iprod{x , Qx}  - \expecf{}{\Iprod{x , Qx} }}^{2t-2}  x^\top Q^\top Q x }  }^t \\
    & \leq t^{2t}  \Paren{ \expecf{}{ \Paren{ \Iprod{x , Qx} - \expecf{}{\Iprod{x , Qx} } }^{2t} }  }^{t-1} \Paren{\expecf{}{ \Paren{ x^\top Q^\top Q x }^t} }  \Biggr\}, 
\end{split}
\end{equation}
where the last inequality follows from applying sos H\"older's.

Then, applying \cref{eqn:re-arranged-poincare} with $f= \Paren{ \Iprod{x , Qx} - \expecf{}{ \Iprod{x , Qx} } }^t$,

\begin{equation}
    \begin{split}
        \sststile{}{Q} \Biggl\{ & \Paren{  \expecf{}{ \Paren{\Iprod{x , Qx} - \expecf{}{\Iprod{x , Qx} } }^{2t}    }  }^t \\
        & \leq \Paren{   \Paren{ \expecf{}{ \Paren{\Iprod{x , Qx} - \expecf{}{\Iprod{x , Qx} }}^t  } }^2 + \sigma^2 \norm{\nabla f(x) }^2_2    }^t  \\
        & \leq 2^t \Paren{ \expecf{}{ \Paren{\Iprod{x , Qx} - \expecf{}{\Iprod{x , Qx} }}^t  } }^{2t} + (2\sigma)^{2t} \norm{\nabla f(x) }^{2t}_2 \\
        & \leq \Paren{ 4 \sigma^{2t} \norm{Q}_F^{2t} }^t + (2\sigma t )^{2t}   \Paren{ \expecf{}{ \Paren{ \Iprod{x , Qx} - \expecf{}{\Iprod{x , Qx} } }^{2t} }  }^{t-1} \Paren{\expecf{}{ \Paren{ x^\top Q^\top Q x }^t} }\Biggr\} ,
    \end{split}
\end{equation}
where the second inequality follows from sos triangle inequality, and  last inequality follows from applying the inductive hypothesis to the first term and the bound obtained in \cref{eqn:bound-on-gradient} on the second. We now focus on bounding  $\Paren{\expecf{}{ \Paren{ x^\top Q^\top Q x }^t} }$ as follows: 
\begin{equation}
    \begin{split}
        \Paren{\expecf{}{ \Paren{ x^\top Q^\top Q x 
  }^t} } &  \leq 2^t   \Paren{\expecf{}{ \Paren{ x^\top Q^\top Q x 
   - \expecf{}{x^\top Q^\top Q x } }^t  } } + 2^t \Paren{ \expecf{}{x^\top Q^\top Q x}^t}\\
    & \leq 2^t \Paren{ \sigma^t \norm{Q^\top Q }_F^{t}  + \norm{Q}_F^{2t} } \leq (4\sigma)^t \norm{Q}_F^{2t}.
    \end{split}
\end{equation}

Therefore, we can re-arrange the above as:
\begin{equation}
\label{eqn:simplified-2t-th-moment}
\begin{split}
    \sststile{}{Q} \Biggl\{ & \Paren{  \expecf{}{ \Paren{\Iprod{x , Qx} - \expecf{}{\Iprod{x , Qx} } }^{2t}    }  }^t \\
    & \leq  \Paren{ 4 \sigma^{2t} \norm{Q}_F^{2t} }^t  + \Paren{  \expecf{}{ \Paren{\Iprod{x , Qx} - \expecf{}{\Iprod{x , Qx} } }^{2t}    }  }^{t-1} \cdot (16 \sigma t)^{2t} \norm{Q}_F^{2t} \Biggr\}
\end{split}
\end{equation}

Finally, applying \cref{eqn:t-th-order-inequality} to \cref{eqn:simplified-2t-th-moment} we can conclude that 
\begin{equation}
\sststile{}{Q} \Set{ \Paren{  \expecf{}{ \Paren{\Iprod{x , Qx} - \expecf{}{\Iprod{x , Qx} } }^{2t}    }  }^t  \leq \Paren{ \Paren{ 32 \sigma t }^{2t} \norm{Q}_F^{2t} }^t } ,
\end{equation}
which concludes the proof. 
 
% Let $f_{2t}(x) =  \Paren{ \Iprod{x, Qx} - \expecf{x\sim \calD }{ \Iprod{x, Qx} } }^{2t}$.  Next,
% \begin{equation}
% \label{eqn:2t-moments}
%     \begin{split}
%     \expecf{x\sim \calD}{ \Paren{  f_{2t}(x) }^{2} } & = \expecf{x\sim \calD}{ \Paren{ \Paren{ f_{2t}(x) - \expecf{x\sim \calD}{ f_{2t}(x) } + \expecf{x\sim \calD}{ f_{2t}(x) }  } }^{2} } \\
%     & \leq 2 \underbrace{ \expecf{x \sim \calD }{  \Paren{ f_{2t}(x) - \expecf{}{f_{2t}(x) } }^{2}    } }_{\eqref{eqn:2t-moments}.(1) } + 2   \underbrace{ \expecf{x \sim \calD }{ f_{2t}(x) } ^2 }_{ \eqref{eqn:2t-moments}.(2) }
%     %&  \leq 2 \sigma^2 \norm{ \nabla f_t(x) }_2^2 + 2 \expecf{x\sim \calD}{f_t(x)^2 }^2 \\ 
%     \end{split}
% \end{equation}
% We can bound term \eqref{eqn:2t-moments}.(1) as follows: 
% \begin{equation*}
% \begin{split}
%     \expecf{x \sim \calD }{  \Paren{ f_{2t}(x) - \expecf{}{f_{2t}(x) } }^{2} }  & \leq  \sigma^2 \norm{ \nabla f_{2t}(x) }_2^2 \\
%     & =  \sigma^2 (2t)^2 \expecf{}{  \Paren{ f_{2t-1}(x)^2 \cdot  \norm{Qx}_2^2      }}
% \end{split}
% \end{equation*}
% \ainesh{What now...?}
\end{proof}

\begin{fact}[$t$-th order inequality]
\label{eqn:t-th-order-inequality}
For any even $t$, 
    $$ \Set{ x^{t} \leq  \Paren{ a^t + x^{t-1} b } } \sststile{}{x} \Set{ x^t \leq 2^t \Paren{a^t + b^t} } .$$
\end{fact}
\begin{proof}
Observe, we can re-write $x^t \leq a^t + x^{t-1}\cdot  b $ as $x \leq  \frac{a^t}{x^{t-1}} + b $. Assume $x > (a+b)$. Then,  $x \leq \frac{a^t}{(a+b)^{t-1}} + b \leq a+b $, and thus $x^t \leq 2^t (a^t + b^t)$. Since this is a univariate inequality, it admits a sos proof.   
\end{proof}

\begin{lemma}[Certifiable Hypercontractivity of $\ell_p$ balls]
\label{lem:certifiable-hyperconstractivity-of-linear-forms}
For any $p \in [1, \infty)$, the uniform distribution on the unit $\ell_p$ ball has  $2k$-certifiably $c$-hypercontractive linear and quadratic forms for any $k \in \mathbb{N}$.
\end{lemma}
\begin{proof}
This proof follows from two known facts and \cref{lem:cert-hyper-poincare}. First, Kothari and Steinhardt showed that for any distribution that satisfies a Poincar\'e inequality with a dimension-independent constant has certifiably-hyper-contractive linear forms ~\cite[Theorem 1.1]{kothari2018agnostic}. \cref{lem:cert-hyper-poincare} shows this implication for quadratic forms. Second, for $p \in [1,2]$, Sodin~\cite{sodin2008isoperimetric} showed that uniform distribution over the $\ell_p$ ball satisfies a Poincar\'e inequality with a dimension-independent constant. For $p \in (2,\infty)$, this was shown by Lata{\l}a and  Wojtaszczyk~\cite{latala2008infimum}. 
\end{proof}

%% file: appendix.tex
\section{Bounding the coefficients of the certificates}\label{sec: bounding_coeffs}

%\ainesh{wip}
% \gou{Seems unnecessary with the modified argument going from SoS refutations to SoS proofs using properties of the reweighted polynomials (can show hypercontractivity for the reweighted polynomials)}

In this section, we argue the polynomials that arise in our SoS certificate have polynomially bounded bit complexity, which is needed for efficient algorithms. The strategy is to instantiate the SoS identity for various values of $v, w$ and then analyze the identities together to bound the coefficients of $q_i$.

In the previous sections, we showed the existence of the identity
\begin{align*}
    \delta n - \sum_{i = 1}^n w_i &= \left(\sum_j p_{0, j}(v, w)^2\right) + \sum_{i \le n} \left(\sum_j p_{i, j}(v, w)^2\right) (\delta w_i \cdot v^T \Sigma v - w_i \ip{v}{x_i}^2)\\
    &\qquad + \sum_i q_i(v, w) (w_i^2 - w_i) + r(w, v) (\norm{v}^2 - 1)
\end{align*}
where $p_{0, j}, p_{i, j}, q_i, r$ are polynomials of degree at most $t$.

\begin{lemma}\label{lem:bitcomplexity}
    $\norm{p_{0, j}}^2, \norm{p_{i, j}}^2, \norm{q_i}^2, \norm{r}^2 \le d^{2t}$
\end{lemma}

\begin{proof}
We will show the argument for $\norm{p_{1, j}}^2$, the rest are similar.
Assume without loss of generality that $\Sigma = I$. Fix a subset $U$ of indices in $[n]$ of size at most $t$. We will consider the coefficients of the terms which depends on $(w_i: i \in U)$ and not on $(w_i: i \not\in U)$. To do this, set $w = w^*$ where $w^*_i = 1$ if $i = 1$ or $i \in U$ and $0$ otherwise. 
% \gnote{Conditioning might be problematic for anti-concentration, use $N(0, I - x_1x_1^T)$ instead? Argument may go through for the first term. Another idea is to invoke distribution on $x_i$s -- choose $v$ first, then pick $w^*$ carefully. Firstly, for any collection of $w_a, w_b, w_c$, there exists some $v$ such that each of $\delta w_a - w_a \ip{v}{x_a}^2 \ge 0.001 \delta$. Therefore, choose $v$ randomly and then pick $w_i$ to be from among those that satisfy this inequality. Issue here is, it's no longer $\EE_v p(v, w^*)^2$ but it's actually $\EE_v p(v, w(v))^2$ and carbery-wright.}
Also sample $v$ from a Gaussian distribution with $\EE\norm{v}^2 = 1$ such that $\ip{v}{x_i}^2 \le \frac{\delta^2}{100}$ w.h.p. for all $i \in U \cup \{1\}$.
Note that this can be done because there are only $t+1$ such constraints.
Moreover, we can choose the mean of $v$ to be $0$ and  covariance $\Theta$ of the distribution of $v$ to be $\eta^2= \delta^2/(400 \log t)$ in the space corresponding to the span of $\{x_i: i \in U \cup \{1\}\}$ and $1$ in every direction orthogonal to this subspace. Then, we scale the covariance by a constant factor so as to satisfy the norm constraint. %For instance, if we had $t = 1$ and $\ip{v}{x_1}^2 \le \frac{\delta^2}{100}$, 
%we sample $v \sim N(0, I - (1 - poly(\delta))x_1x_1^T)$ up to constant scaling.
With this choice, then the identity simplifies to (after taking expectation)
\begin{align*}
    \delta n - |U \cup \{1\}| &= \left(\sum_j \EE_v p_{0, j}(v, w^*)^2\right) + \sum_{i \in U \cup \{1\}} \left(\sum_j p_{i, j}(v, w^*)^2\right) \big(\frac{99 \delta}{100}\big)
\end{align*}
This implies
\begin{align}
    \delta n  \ge \frac{99\delta}{100} \sum_{i \in U \cup \{1\}} \EE_v\left(\sum_j p_{i, j}(v, w^*)^2\right) \label{eq:bitcomplex:eq1}
\end{align}
where many terms got zeroed out because of our choice of $w^*$ and for the nonzero $w^*_i$, we invoked the constraint bound on the terms $\delta w_i \cdot v^T \Sigma v - w_i \ip{v}{x_i}^2$.

It now suffices to prove the following claim.

\noindent {\bf Claim.} {\em For some constant $c=c(t)>0$ 
\begin{equation}\label{eq:bitcomplex:eq2}
\EE_{v}[p_{ij}(v,w^*)^2] \ge \Var[p_{ij}(v,w^*)] \ge \Big( \frac{c}{\delta^t}\Big) \|p_{ij} \|^2,
\end{equation}
where $\| p_{ij}\|^2$ denotes the sum of the squares of the coefficients. }

Once we establish the claim,  \eqref{eq:bitcomplex:eq2} together with \eqref{eq:bitcomplex:eq1} establishes Lemma~\ref{lem:bitcomplexity}. We now prove the above claim. 

Let $z=\Theta^{-1/2} v$, so $z$ is a standard Gaussian with identity covariance. For simplicity, let the polynomial $f(v) \coloneqq p_{ij}(v,w^*)$ and let $f(v)=\sum_{j=0}^t f_j(v)$ where $f_j$ is a homogenous polynomial of degree $j$. Finally let $\forall j \in \{0,1,\dots, t\}$ let $h_j(z)=f_j(\Theta^{1/2} z)$, and let $h(z) \coloneqq \sum_{j=0}^t h_j(z) = f(\Theta^{1/2} z)$.

From standard facts about standard Gaussians, when $z \sim N(0,I)$ we have
\begin{equation}\Var_v[p_{ij}(v,w^*)]=\Var_v[f(v)]=\Var_{z}[h(z)]=\norm{h}^2=\sum_{j=0}^t \norm{h_j}^2,\label{eq:bitcomplexity:3}
\end{equation}
where $\norm{h}^2$ denotes the sum of the squares of the coefficients  of $h$. 

To relate $\norm{h_j}$ and $\norm{f_j}$, let us denote by $T^{(f)}_j$ the symmetric tensor that satisfies $f_j(v)=\iprod{T^{(f)}_j, v^{\otimes j}}$. We have 
$$ f_j(v)=\iprod{T^{(f)}_j, v^{\otimes j}} = \Big\langle (\Theta^{1/2})^{\otimes j} ~ T^{(f)}_j, z^{\otimes j} \Big\rangle = h_j(z).$$ 
Hence it follows that
\begin{align*}
\norm{h_j}^2 &\ge \Big\langle (\Theta^{1/2})^{\otimes j} ~ T^{(f)}_j, (\Theta^{1/2})^{\otimes j} ~T^{(f)}_j \Big\rangle \ge \sigma_{\min}(\Theta)^{j} \cdot \norm{T^{(f)}_j}_F^2 \ge \eta^j \cdot \frac{\norm{f_j}^2}{j!}.\\
\text{Hence } \norm{f_j}^2&\le \frac{j!}{\eta^j} \norm{h_j}^2 \le \frac{t!}{\eta^t} \norm{h_j}^2 \le \frac{c_t}{\delta^t} \cdot \norm{h_j}^2 ~~~\forall j \in \{0,1,\dots, t\},
\end{align*}
and some constant $c_t>0$ that depends on $t$. Summing over all $j \in \{0,1,\dots,t\}$ and using \eqref{eq:bitcomplexity:3} establishes the claim. This completes the proof of Lemma~\ref{lem:bitcomplexity}. 
\end{proof}